\documentclass[aos]{imsart}

\arxiv{math.PR/0000000}

\RequirePackage{amsthm,amsmath}%,natbib}
\usepackage{amsfonts}   % Include AMS fonts
\usepackage{amssymb}    % Include AMS symbols
\usepackage{mathtools}
\startlocaldefs

\numberwithin{equation}{section}
\input aad.tex

\def\aa{\underline a}
\def\bb{\underline b}
\def\AA{\underline A}
\def\asc{\mathfrak{\underline{a}}}
\def\bsc{\mathfrak{\underline{b}}}
\def\fsc{\mathfrak{\underline{f}}}
\let\tilde\widetilde

\def\dsum{\mathop{\mathchar"1350\mathchar"1350}\limits}
\def\ksum{\mathop{\mathchar"1350\mathchar"1350\cdots\mathchar"1350}\limits}

\endlocaldefs

\begin{document}

\begin{frontmatter}
\title{Minimax Estimation of a Functional on a Structured High-dimensional Model (corrected version)}
%\title{Estimation of a Functional on a Structured Model under Low Regularity}
\runtitle{Minimax Estimation on a Structured Model}

\begin{aug}
\author{James M. Robins\ead[label=e3]{robins@hsph.harvard.edu}},
\author{Lingling Li\ead[label=e2]{lingling07.li@gmail.com}},
\author{Lin Liu\ead[label=e5]{linliu.tju@gmail.com}},
\author{Rajarshi Mukherjee\ead[label=e5]{rajmrt23@gmail.com}},
\author{Eric Tchetgen Tchetgen\ead[label=e1]{etchetge@hsph.harvard.edu}}
\and 
\author{Aad van der Vaart\ead[label=e4]{avdvaart@math.leidenuniv.nl}}
\thankstext{}{The research leading to these results has 
received funding from the European Research Council 
under ERC Grant Agreement 320637, an NWO Spinoza grant,
and from the National Institutes of Health grants AI112339,  AI32475, AI113251 and ES020337.}
\runauthor{Robins et al.}
\affiliation{Harvard University, Karyopharm Therapeutics Inc.,
CMA-Shanghai, Harvard University
and TU Delft}
\address{James M. Robins\\
Harvard T.H. Chan School of Public Health\\
\printead{e3}}
\address{Lingling Li\\
Karyopharm Therapeutics Inc.\\
\printead{e2}}
\address{Lin Liu\\
Institute of Natural Sciences\\
School of Mathematical Sciences\\
CMA-Shanghai, MOE-LSC\\
\printead{e5}}
\address{Rajarshi Mukherjee\\
Department of Biostatistics\\
Harvard Unversity\\
\printead{e5}}
\address{Eric Tchetgen Tchetgen \\
The Wharton School of Business\\
University of Pennsylvania\\
\printead{e1}}
\address{Aad van der Vaart\\
Delft institute of Applied Mathemats\\
TU Delft\\
The Netherlands\\
\printead{e4}}
\end{aug}

\begin{abstract}
We introduce a new method of estimation of parameters in
semiparametric and nonparametric models. The method is
based $U$-statistics that are based on higher order
influence functions that extend
ordinary linear influence functions of the parameter of interest,
and represent higher derivatives of this parameter.
For parameters for which the representation cannot be perfect
the method often leads to a bias-variance trade-off, and 
results in estimators that converge at a slower than
$\sqrt n$-rate. In a number of examples the resulting rate
can be shown to be optimal. We are particularly interested
in estimating parameters in models with
a nuisance parameter of high dimension or
low regularity, where the parameter of interest
cannot be estimated at $\sqrt n$-rate, but we also consider
efficient $\sqrt n$-estimation using novel nonlinear estimators.
The general approach is applied in detail to
the example of estimating a mean response
when the response is not always observed.
\end{abstract}

\begin{keyword}[class=AMS]
\kwd[Primary ]{62G05, 62G20, 62G20, 62F25}
%\kwd[; secondary ]{}
\end{keyword}

\begin{keyword}
\kwd{Nonlinear functional, nonparametric estimation,
$U$-statistic, influence function, tangent space}
\end{keyword}

\end{frontmatter}

\maketitle

\date{July 2022}

\section{Introduction}
\label{SectionIntroduction}
Let $X_1,X_2,\ldots,X_n$ be a random sample from a density
$p$  relative to a measure $\m$ on a sample space $(\X,\A)$. 
It is known that $p$ belongs to a collection  $\P$ of densities,
and the problem is to estimate the value $\chi(p)$ of a functional $\chi: \P\to \RR$.
Our main interest is in the situation of a semiparametric
or nonparametric model, where $\P$ is infinite dimensional, and especially
in the case when the model is described through parameters of low
regularity. In this case the parameter $\chi(p)$ may not be estimable 
at the ``usual'' $\sqrt n$-rate.

In low-dimensional semiparametric models 
estimating equations have been found a good strategy 
for constructing estimators \cite{BKRW,vdLaanRobins,vanderVaart98}.
In our present setting it will be more convenient to consider one-step versions
of such estimators, which take the form 
\begin{equation}
\label{EqLinearEstimator}
\hat\chi_n=\chi(\hat p_n)+\PP_n \chi_{\hat p_n},
\end{equation}
for $\hat p_n$ an initial estimator for $p$ and
$x\mapsto \chi_p(x)$ a given measurable function, for each $p\in \P$,
and $\PP_nf$ short hand notation for $n^{-1}\sumin f(X_i)$.

One possible choice in (\ref{EqLinearEstimator})
is $\chi_p=0$, leading to the
plug-in estimator $\chi(\hat p_n)$. However, unless the
initial estimator $\hat p_n$ possesses special properties, this
choice is typically suboptimal. Better functions $\chi_p$ can be
constructed by consideration of the \emph{tangent space} of the model.
To see this, we write (with $P\chi_{\hat p}$ shorthand for
$\int \chi_{\hat p}(x)\,dP(x)$)
\begin{equation}
\label{EqBiasVarianceLinearEstimator}
\hat\chi_n-\chi(p)=\bigl[\chi(\hat p_n)-\chi(p)+
P \chi_{\hat p_n}\bigr]+(\PP_n-P) \chi_{\hat p_n}.
\end{equation}
Because it is properly centered, we may expect
 the sequence $\sqrt n(\PP_n-P) \chi_{\hat p_n}$
to tend in distribution to a mean-zero normal distribution.
The term between square brackets on the right of 
(\ref{EqBiasVarianceLinearEstimator}), which we shall
refer to as the \emph{bias term},
depends on the initial estimator $\hat p_n$,
and it would be natural to construct the function $\chi_p$ 
such that this term does not contribute to the limit
distribution, or at least is not dominating the
expression. Thus we would like to choose this
function such that the ``bias term'' is
no bigger than of the order $O_P(n^{-1/2})$. 
A good choice is to ensure that the term 
$P \chi_{\hat p_n}$ acts as minus the first
derivative of the functional $\chi$ in the ``direction'' $\hat p_n-p$.
Functions $x\mapsto \chi_p(x)$ with this property are
known as \emph{influence functions} in semiparametric theory
\cite{KoshevnikLevit,Pfanzagl82,vanderVaart91,vdVStFlour,BKRW},
go back to the \emph{von Mises calculus} due to \cite{vonMises}, and
play an important role in robust statistics \cite{Huber,Hampel}, or \cite{vanderVaart98}, Chapter~20.

For  an influence function we may expect
that the ``bias term'' is quadratic in the
error $d(\hat p_n,p)$, for an appropriate distance $d$. 
In that case it is certainly negligible as soon as this error
is of order $o_P(n^{-1/4})$. Such a ``no-bias'' condition
is well known in semiparametric theory
(e.g.\ condition (25.52) in \cite{vanderVaart98} or (11) in \cite{MurphyvdV00}).
However, typically it requires that the model $\P$ be ``not too big''. 
For instance, a regression or density function on
$d$-dimensional space can be estimated at rate $n^{-1/4}$ 
if it is a-priori known to have at least  $d/2$ derivatives
(indeed $\a/(2\a+d)\ge 1/4$ if $\a\ge d/2$).
The purpose of this paper is to develop 
estimation procedures for the case that no
estimators exist that attain a $O_P(n^{-1/4})$ rate of
convergence. The estimator (\ref{EqLinearEstimator}) is then
suboptimal, because it fails to make a proper trade-off between
``bias'' and ``variance'': the two terms in
(\ref{EqBiasVarianceLinearEstimator}) have different magnitudes. 
Our strategy is to replace the
linear term $\PP_n\chi_p$ by a general 
$U$-statistic $\UU_n\chi_p$, for an appropriate 
$m$-dimensional \emph{influence function}
$(x_1,\ldots,x_m)\mapsto \chi_p(x_1,\ldots, x_m)$, chosen 
using a type of von Mises expansion of $p\mapsto \chi(p)$.  
Here the order $m$ is adapted to the size of the model $\P$ and
the type of functional to be estimated.

Unfortunately, ``exact'' higher-order influence functions
turn out to exist only for special functionals $\chi$. 
To treat general functionals $\chi$ we approximate these by
simpler functionals, or use approximate influence functions.
The rate of the resulting estimator is then determined
by a trade-off between bias and variance terms.
It may still be of order $1/\sqrt n$, but it is typically
slower. In the former case, surprisingly, one may obtain
semiparametric efficiency by estimators whose variance is
determined by the linear term, but whose bias is corrected
using higher order influence functions. The latter case will be of more
interest.

The conclusion that the ``bias term'' in
(\ref{EqBiasVarianceLinearEstimator}) is quadratic in the estimation
error $d(\hat p_n,p)$ is based on a worst case analysis. First, there
exist a large number of models and functionals of interest that permit
a first order influence function that is unbiased in the nuisance
parameter. (E.g.\ adaptive models as considered in \cite{Bickel82},
models allowing a sufficient statistic for the nuisance parameter as
in \cite{vanderVaart86,vdV88}, mixture models as considered in
\cite{Lindsay,Pfanzagl89,vanderVaart96}, and convex-linear
models in survival analysis.)  In such models there is no need for
higher-order influence functions.  Second, the analysis does not take
special, structural properties of the initial estimators $\hat p_n$
into account. An alternative approach would be to study the bias of a
particular estimator in detail, and adapt the influence function to
this special estimator. The strategy in this paper is not to use such
special properties and focus on influence functions
that work with general initial estimators $\hat p_n$.

The motivation for our new estimators stems from studies in
epidemiology and econometrics that include
covariates whose influence on an outcome of
interest cannot be reliably modelled by a simple model. These
covariates may themselves not be of interest, but are
included in the analysis to adjust the analysis for possible bias. For instance, the
mechanism that describes why certain data is missing is in terms
of conditional probabilities given several covariates, but the
functional form of this dependence is unknown.  Or, to permit a
causal interpretation in an observational study one
conditions on a set of covariates to control for 
confounding, but the form of the dependence on the confounding variables is
unknown. One may hypothesize in such situations that the
functional dependence on a set of (continuous) covariates is smooth
(e.g. $d/2$ times differentiable in the case of $d$ covariates), or even linear. Then
the usual estimators will be accurate (at order $O_P(n^{-1/2})$) if the
hypothesis is true, but they will be badly biased in the other
case. In particular, the usual normal-theory based confidence
intervals may be totally misleading: they will be both too
narrow and wrongly located. The methods in this paper yield estimators
with (typically) wider corresponding confidence intervals, but they are correct
under weaker assumptions.

The mathematical contributions of the paper  are to provide a heuristic for constructing
minimax estimators in semiparametric models, and to apply this
to a concrete  model, which is a template for a number of other models
(see \cite{RobinsetalFreedman,vdVStatScience}).
The methods connect to earlier work \cite{HasminskiiIbragimov,Nemirovski} on the estimation of functionals 
on nonparametric models, but differ by our focus on
functionals that are defined in terms of the structure of a semiparametric model.
This requires an analysis of the inverse map from the density of the observations
to the parameters, in terms of the semiparametric tangent spaces of the models.
Our second order estimators are related to work on quadratic functionals, or functionals that are well approximated
by quadratic functionals, as in 
\cite{DonohoNussbaum90, KerkPicard96,BickelRitov,BirgeMassart95,Laurent97,LaurentMassart,CaiLow05,CaiLow06}.
While we place the construction of minimax estimators for these special functionals
in a wider framework, our focus differs by going beyond quadratic estimators
and to consider semiparametric models. 

Our mathematical results are in part conditional on a scale of regularity parameters (through
the dimension given in (\ref{EqDefinitionk}) and a partition of this dimension that depends on two of these parameters). 
We hope to discuss adaptation to these parameters in future work.

General heuristics of our construction are given in Section~\ref{SectionHeuristics}. 
Sections~\ref{SectionMarFirstOrder}--\ref{SectionMinimaxRate}
are devoted to constructing new estimators for the mean response
effect in missing data problems.  The latter are introduced in 
Section~\ref{SectionMissingData}, so that they can serve as illustration
to the general heuristics in Section~\ref{SectionHeuristics}. 
%From the general perspective of representation of derivatives, there is nothing special about the
%quadratic order. 
In Section~\ref{SectionOtherExamples} (in the supplement \cite{Supplement}) we briefly discuss  other problems, including
estimating a density at a point, where already
first order influence functions do not exist and our heuristics
naturally lead to projection estimators, and estimating a quadratic functional, where our approach produces standard estimators
from the literature in a natural way. Section~\ref{SectionProofs} (partly in the supplement \cite{Supplement})
collects technical proofs. Sections~\ref{SectionInfluenceFunctions},~\ref{SectionProjections} and~\ref{SectionU} 
(in the supplement \cite{Supplement})
discuss three key concepts of the paper: influence functions, 
projections and $U$-statistics.

\section{Notation}
\label{SectionNotation}
Let $\UU_n$ denote the \emph{empirical $U$-statistic} measure, 
viewed as an operator on functions.
For given $m\le n$ and a function $f: \X^m\to\RR$ on the sample space
this is defined by
$$\UU_n f=\frac1{n(n-1)\cdots (n-m+1)}\ksum_
{1\le i_1\not=i_2\not=\cdots\not=i_m\le n} f(X_{i_1}, X_{i_2},\cdots,X_{i_m}).$$
We do not let the order $m$ show up in the notation $\UU_nf$.
This is unnecessary, as the notation is consistent
in the following sense: if a function $f: \X^l\to\RR$ of $l<m$ 
arguments is considered a function of $m$ arguments that is constant
in its last $m-l$ arguments, then the right side of
the preceding display is well
defined and is exactly the corresponding $U$-statistic of order
$l$. In particular, $\UU_nf$ is the \emph{empirical distribution}
$\PP_n$ applied to $f$ if $f: \X\to\RR$ depends on only one argument.  

We write $P^n\UU_n f=P^mf$ for the expectation of 
$\UU_nf$ if $X_1,\ldots, X_n$ are distributed according to 
the probability measure $P$, and for the expectation of $f$ under the product measure $P^m$ of $m$
copies of $P$. We also use this operator
notation for the expectations of statistics in general. If the distribution of the
observations is given by  a density $p$, then we use $P$ as the measure
corresponding to $p$, and use the preceding notations likewise.
Finally $\UU_n-P^m$ denotes the centered $U$-statistic empirical measure,
defined by $(\UU_n-P^m)f=\UU_nf-P^mf$, for any integrable function $f$.

We  call $f$ \emph{degenerate} relative to $P$ 
if $\int f(x_1,\ldots,x_m)\,dP(x_i)=0$ for every $i$ and every $(x_j: j\not=i)$,
and we call $f$ \emph{symmetric} if $f(x_1,\ldots, x_m)$ is invariant under
permutation of the arguments $x_1,\ldots, x_m$. 
Given an arbitrary measurable function $f: \X^m\ra\RR$ we can form a
function that is degenerate relative to $P$ by subtracting
the orthogonal projection in $L_2(P^m)$
onto the functions of at most $m-1$ variables.
This degenerate function can be written in the form (e.g.\
\cite{vanderVaart98}, Lemma~11.11)
\begin{align}
\label{EqMakeDegenerate}
(D_Pf)(X_1,\ldots, X_m)
&\!=\!\!\!\!\!\!\sum_{A\subset \{1,\ldots, m\}}\!\!\!\!\!
(-1)^{m-|A|}\E_P\Bigl[f(X_1,\ldots, X_m)\given X_i: i \in A\Bigr],\!\!
\end{align}
where the sum is over all subsets $A$ of $\{1,\ldots, m\}$,
including the empty set. Here the conditional expectation
$\E\bigl[f(X_1,\ldots, X_m)\given X_i: i\in \emptyset\bigr]$
is understood to be the unconditional expectation $\E f(X_1,\ldots, X_m) =P^mf$.
If the function $f$ is symmetric, then so is the function $D_Pf$.

Given two functions $g, h: \X\to\RR$ we write $g\times h$ for
the function $(x,y)\mapsto g(x)h(y)$. More generally, given
$m$ functions $g_1,\ldots, g_m$
we write $g_1\times\cdots\times g_m$ for the tensor product
of these functions. Such product functions are degenerate
iff all functions in the product have mean zero.

A \emph{kernel operator} $K: L_r(\X,\A,\m)\to L_r(\X,\A,\m)$ takes
the form $(Kf)(x)=\int \bar K(x,y)f(y)\,d\m(y)$ for some
measurable function $\bar K: \X^2\to \RR$. We shall
abuse notation in denoting the operator $K$ and the \emph{kernel} $\bar K$ 
with the same symbol: $K=\bar K$. A (weighted) projection onto
a finite-dimensional space is a kernel operator. We discuss such
projections in Section~\ref{SectionProjections}.

The set of measurable functions whose $r$th absolute
power is $\m$-integrable is denoted $L_r(\m)$, with norm $\|\cdot\|_{r,\m}$,
or $\|\cdot\|_r$ if the measure is clear; or also as $L_r(w)$ with norm $\|\cdot\|_{r,w}$ 
if $w$ is a density relative to a given dominating measure. For $r=\infty$
the notation $\|\cdot\|_\infty$ refers to the uniform norm.

\section{Estimating the mean response in missing data models}
\label{SectionMissingData}
In this section we introduce our main example, which will be used as a running example
in the next section. We also summarize the results obtained for this example in the remainder
of the paper.

Suppose that a typical observation is distributed as $X=(YA,A,Z)$, for
$Y$ and $A$ taking values in the two-point set $\{0,1\}$ and
conditionally independent given $Z$. 

This model is standard in biostatistical applications, with $Y$ an  ``outcome''
or ``response variable'', which is observed only if the indicator $A$ takes the value
$1$.  The covariate $Z$ is chosen such that it contains all
information on the dependence between the response and the missingness
indicator $A$, thus making the response \emph{missing at random}. 
Alternatively, we think of $Y$ as a
``counterfactual'' outcome if a treatment were given ($A=1$) and
estimate (half) the treatment effect under the assumption of 
\emph{no unmeasured confounders}. (The results also apply without the ``missing-at-random''
assumption, but with a different interpretation; see Remark~\ref{RemarkNonMAR}.)

The model can be parameterized by the marginal density $f$ of $Z$
(relative to some dominating measure $\n$) and the probabilities
$b(z)=\Pr(Y=1\given Z=z)$ and $a(z)^{-1}=\Pr(A=1\given Z=z)$. (Using
$a$ for the inverse probability simplifies later
formulas.)  Alternatively, the model can be parameterized by the pair $(a,b)$
and the function $g=f/a$, which is the conditional density of $Z$
given $A=1$, up to the norming factor $\Pr(A=1)$.  Thus the density $p$
of an observation $X$ is described by the triplet $(a,b,f)$, or
equivalently the triplet $(a,b,g)$.  For simplicity of notation we
write $p$ instead of $p_{a,b,f}$ or $p_{a,b,g}$, with the implicit
understanding that a generic $p$ corresponds one-to-one to a generic
$(a,b,f)$ or $(a,b,g)$.

We wish to estimate the \emph{mean response} $\E Y=\E b(Z)$, i.e.\ the functional 
$$\chi(p)=\int bf\,d\n=\int ab g\,d\n.$$
Estimators that are $\sqrt n$-consistent and asymptotically
efficient in the semiparametric sense have been 
constructed using a variety of methods (e.g.
\cite{RobinsRotnitzkyRegressionMissing,RotnitzkyRobins95}, or see Section~\ref{SectionMarFirstOrder}), but 
only if $a$ or $b$, or both, parameters are restricted to
sufficiently small regularity classes.  For instance, if the covariate $Z$ ranges
over a compact, convex subset $\Z$ of $\RR^d$, then 
the mentioned papers provide $\sqrt n$-consistent estimators 
under the assumption that $a$ and $b$ belong to H\"older classes
$C^\a(\Z)$ and $C^\b(\Z)$ with $\a$ and $\b$ large enough that 
\begin{equation}
\label{EqEnoughSmoothnessForRootnByLinear}
\frac\a{2\a+d}+\frac\b{2\b+d}\ge \frac12.
\end{equation}
(See e.g.\ Section~2.7.1 in \cite{vdVWellner} for the definition of
H\"older classes.)
For moderate to large dimensions $d$ this is a restrictive
requirement. In the sequel we consider estimation
for arbitrarily small $\a$ and $\b$.

\subsection{Summary of results}
Throughout we assume that the parameters $a$, $b$ and $g$ are contained in
H\"older spaces $C^\a(\Z)$, $C^\b(\Z)$  and $C^\g(\Z)$ of functions on a compact,
convex domain in $\RR^d$. We derive two types of results:
\begin{enumerate}
\item[(a)] In Section~\ref{SectionParametricRate} we show that a $\sqrt
n$-rate is attainable by using a higher order influence function (of
order determined by $\g$) as long as
\begin{equation}
\label{EqEnoughRootn}
\frac{\a+\b}2\ge \frac d4.
\end{equation}
This condition is strictly weaker than the condition (\ref{EqEnoughSmoothnessForRootnByLinear})
under which the linear estimator attains a $\sqrt n$-rate.
Thus even in the $\sqrt n$-situation higher order estimating
equations may yield estimators that are applicable in a
wider range of models. For instance, in the case that $\a=\b$
the cut-off (\ref{EqEnoughSmoothnessForRootnByLinear}) arises 
for $\a=\b\ge d/2$, whereas (\ref{EqEnoughRootn}) reduces to 
$\a=\b\ge d/4$.
\item[(b)] We consider minimax estimation in the case $(\a+\b)/2< d/4$, when 
the rate becomes slower than $1/\sqrt n$. 
It is shown in \cite{RobinsetalMinimax}
that even if $g=f/a$ were known, then the minimax rate for
$a$ and $b$ ranging over balls in the H\"older classes $C^\a(\Z)$ and $C^\b(\Z)$ cannot be faster
than $n^{-(2\a+2\b)/(2\a+2\b+d)}$. In Section~\ref{SectionMinimaxRate}
we show that this rate is attainable if $g$ is known, and also if 
$g$ is unknown, but is a-priori known to belong
to a H\"older class $C^\g(\Z)$ for sufficiently large
$\g$, as given by (\ref{EqfSmoothEnough}). 
(Heuristic arguments, not discussed in this paper,
appear to indicate that for smaller $\g$ the minimax rate is slower
than $n^{-(2\a+2\b)/(2\a+2\b+d)}$.)
\end{enumerate}

We start by discussing the first and second 
order estimators in Sections~\ref{SectionMarFirstOrder} and~\ref{SectionMarSecondOrder},
where the first is merely a summary of well known facts, but the second already 
contains some key elements of the new approach of the present paper. The preceding results
(a) and (b) are next obtained in Sections~\ref{SectionParametricRate} ($\sqrt n$-rate
if $(\a+\b)/2\ge d/4$)
and~\ref{SectionMinimaxRate} (slower rate if $(\a+\b)/2<d/4$),
using the higher-order influence functions
of an approximate functional, which is defined 
in the intermediate Section~\ref{SectionMARApproximateFunctional}.
In the next section we discuss the general heuristics of our approach.

\begin{assumption}
\label{Assumption}
We assume throughout that the functions $1/a,b,g$ and their preliminary
estimators $1/\hat a,\hat b,\hat g$ are bounded away from their extremes:
0 and 1 for the first two, and 0 and $\infty$ for the third.
\end{assumption}

\begin{remark}
\label{RemarkNonMAR}
The assumption that the responses are ``missing at random (MAR)'' is used to identify the
mean response functional. Without this assumption the results of the paper are
still valid, but concern the functional $\int b_1(z)f(z)\,dz$, in which $b_1(z)=\E(Y\given A=1,Z=z)$ has taken
the place of $b(z)=\E(Y\given Z=z)$, two functions that are identical under MAR. This follows from
the fact that the likelihoods of $X=(YA, A, Z)$ without or with assuming MAR take exactly the 
same form, as given in (\ref{EqMarLikelihood}), but with $b$ replaced by $b_1$.
After this replacement all results  go through. However, the functional $\int b_1(z)f(z)\,dz$
has the interpretation of the mean response only when MAR holds.
\end{remark}

\section{General heuristics}
\label{SectionHeuristics}
Our basic estimator has the form (\ref{EqLinearEstimator})
except that we replace the linear term by  a general $U$-statistic.
Given measurable functions $\chi_{p}: \X^m\to\RR$, for a fixed 
\emph{order} $m$, we consider estimators $\hat \chi_n$ of $\chi(p)$ of the type
\begin{equation}
\label{EqEstimator}
\hat\chi_n= \chi(\hat p_n)+\UU_n\chi_{\hat p_n}.
\end{equation}
The initial estimators $\hat p_n$ are thought to have a certain
(optimal) convergence rate $d(\hat p_n,p)\ra 0$, but need not possess
(further) special properties. Throughout we
shall treat these estimators as being based on an
independent sample of observations, so that $\hat p_n$ and $\UU_n$ in (\ref{EqEstimator}) are independent.
This takes away technical complications, and allows us to focus on rates
of estimation in full generality. (A simple way to avoid the
resulting asymmetry would be to swap the two samples, calculate the estimator a
second time and take the average.)

\subsection{Influence functions}
\label{SubSectionInfluenceFunctions}
The key is to find suitable ``influence functions'' $\chi_{p}$. A decomposition of type
(\ref{EqBiasVarianceLinearEstimator}) for the estimator (\ref{EqEstimator})
 yields 
\begin{equation}
\label{EqBiasVarianceUEstimator}
\hat\chi_n-\chi(p)=\bigl[\chi(\hat p_n)-\chi(p)+
P^m \chi_{\hat p_n}\bigr]+(\UU_n-P^m) \chi_{\hat p_n}.
\end{equation}
This suggests to construct the influence functions such that $-P^m\chi_{\hat p_n}$ represents the
first $m$ terms of the Taylor expansion of  $\chi(\hat p_n)-\chi(p)$. We shall  translate this
requirement into a manageable form, and next work it out in detail for the missing data problem.

First the requirement implies that
the influence function used in (\ref{EqEstimator}) must be unbiased:
\begin{equation}
\label{EqUnbiasedIF}
P^m\chi_p=0.
\end{equation}
Next, to operationalize a ``Taylor
expansion'' on the (infinite-dimensional) ``manifold'' $\P$
we employ ``smooth'' submodels $t\mapsto p_t$. These are defined as maps from a
neighbourhood of $0\in\RR$ to $\P$ that pass through
$p$ at $t=0$ (i.e.\ $p_0=p$) such that the derivatives in the following exist. For a
large model there will be many such submodels, approaching $p$ from various ``directions''.
Given a collection of submodels we determine $\chi_p$ such that, for each submodel $t\mapsto p_t$,
$$\frac{d^j}{dt^j}_{|t=0}\chi(p_t) 
=-\frac{d^j}{dt^j}_{|t=0} P^m\chi_{p_t},\qquad j=1,\ldots,m.$$
The subscript $|t=0$ on the differential quotients means ``derivative evaluated at $t=0$'',
i.e.\ at $p=p_0$. A slight strengthening is to impose this condition
``everywhere'' on the path, i.e.\
the $j$th derivative of $t\mapsto\chi(p_t)$ at $t$ is the
$j$th derivative of $h\mapsto - P_t^m\chi_{p_{t+h}}$ at $h=0$, for every $t$.
(Here $P_t$ is the measure corresponding to the density $p_t$ and
$P_t^mf$ the expectation of a function $f$ under the $m$-fold product of these measures.)
If the map $(s,t)\mapsto P_s^m\chi_{p_t}$ is smooth, then
the latter implies (cf.\ Lemma~\ref{LemmaTaylorExpansion}
applied with $\chi=f$ and $g(s,t)=-P_t^m \chi_{p_s}$)
\begin{equation}
\label{EqHODerivatives}
\frac{d^j}{dt^j}_{|t=0}\chi(p_t) 
=\frac{d^j}{dt^j}_{|t=0} P_t^m\chi_{p},\qquad j=1,\ldots,m.
\end{equation}
Relative to the previous formula the subscript $t$ on the right hand side has changed places,
and the negative  sign has disappeared. This is similar to the ``Bartlett equalities'' familiar from manipulating
expectations of scores and their higher derivatives. We take this equation together with unbiasedness
as the defining property. Thus a measurable function $\chi_p: \X^m\to \RR$  is said to be
an $m$th order \emph{influence function} at $p$ of the functional $p\mapsto \chi(p)$ 
relative to a given collection of one-dimensional
submodels $t\mapsto p_t$ (with $p_0=p$) if it satisfies (\ref{EqUnbiasedIF}) and (\ref{EqHODerivatives}),
for every submodel under consideration. 

Equation (\ref{EqHODerivatives}) implies a Taylor expansion of $t\mapsto\chi(p_t)$ at $t=0$ of order $m$,
but in addition requires that the derivatives of this map can be represented as \emph{expectations}
involving a function $\chi_p$. The latter is made operational by requiring the derivatives to be identical
to those of the map $t\mapsto P_t^m\chi_p$, which automatically have the desired representation.
The representation as an expectation is essential for the construction of estimators. For exploiting
derivatives up to the $m$th order, groups of $m$ observations can be used to match
the expectation $P^m$; this leads to $U$-statistics of order $m$. 

It is also essential that the expectation is relative to the law of the observations $X_1,\ldots, X_n$.
In a structured model, such as the missing data problem, 
the  law $P_\eta$ of the observations depends on a parameter $\eta$ and the functional of interest
is a quantity $\psi(\eta)$ defined in term of $\eta$. Then the representation
requires to represent the derivative of the map $\eta\mapsto\psi(\eta)$
as an expectation relative to $P_\eta$. An expansion of just $\eta\mapsto \psi(\eta)$ without reference
to the data distribution is not sufficient.
Expressing the derivates in $P_\eta$ implicitly utilises the  inverse map $P_\eta\mapsto \eta$, but
by directly defining the influence function by (\ref{EqHODerivatives}) we sidestep 
an expansion of $\eta\mapsto\psi(\eta)$ and explicit inversion of the latter map.

We allow that there may be more than one influence function.
In particular, we do not require $\chi_p$ in (\ref{EqHODerivatives}) 
to be symmetric in its arguments, although a given influence
function can always be symmetrized without loss of generality. Furthermore, as the 
collection of paths $t\mapsto p_t$ is restricted by the model, which may be smaller
than the set of all possible densities on the sample space, certain projections of
an influence function may also be influence functions.

\begin{example}
[Classical $U$-statistic]\normalfont
The mean functional  $\chi(p)=\E_p\UU_nf=P^kf$ of a $k$th order $U$-statistic 
has $m$th order influence function given by $\chi_p(x_1,\ldots,x_m)=f(x_1,\ldots,x_k)-P^kf$, for every
$m\ge k$. Alternatively,  the symmetrized version $\UU_mf-P^kf$ of this function is also an influence function.
This example connects to classical $U$-statistic theory, and may serve to gain some insight
in the definition, but our interest in influence functions will go in a different direction.

In the preceding claim we did not specify the set of paths $t\mapsto p_t$. In fact the claim is true
for the nonparametric model and all reasonable paths. %For instance, for every bounded function $g: \X\to \RR$ with $Pg=0$, the definition $p_t=p(1+tg)$ gives a valid probability density, at least for small $|t|$. 
%The function $g$ is the score function of the one-dimensional model $t\mapsto p_t$. 
The claim follows trivially from the fact that $t\mapsto\chi(p_t)=P_t^kf$ has
the same derivatives as $t\mapsto P_t^m\chi_p=P_t^mf-P^kf=P_t^k f-P^kf$, where
in the last equality we use that $m\ge k$. (The $j$th derivative for $j>k$ vanishes.)

For $1\le m<k$ one can   verify, with more effort,  that the orthogonal projection in $L_2(P^k)$ 
of $f$ on the subspace of functions of $m$ variables is an influence function.
\end{example}

\begin{example}
[Missing data, paths]\normalfont
The missing data model introduced in Section~\ref{SectionMissingData} is parameterized
by the parameter triplet $(a,b,f)$. %or alternatively by $(a,b,g)$, with $g=f/a$. 
The likelihood of a typical observation $X=(YA,A,Z)$ can be seen to take the form
\begin{equation}
\label{EqMarLikelihood}
p_{a,b,f}(X)=f(Z) \Bigl(\frac 1{a(Z)}b(Z)^Y\bigl(1-b(Z)\bigr)^{1-Y}\Bigr)^A \Bigl(1-\frac 1{a(Z)}\Bigr)^{1-A}.
\end{equation}
Submodels are naturally constructed as $t\mapsto p_{a_t, b_t, f_t}$, for given curves 
$t\mapsto a_t$, $t\mapsto b_t$ and $t\mapsto f_t$ in the respective parameter spaces.

In view of Assumption~\ref{Assumption} paths of the form
$a_t=a+t\asc$ and $b_t=b+t\bsc$, for given bounded, measurable functions
$\asc, \bsc: \Z\to\RR$ are valid curves in the parameter space, at least for $t$
in a neighbourhood of 0. We may restrict the perturbations $\asc$ and $\bsc$ to
be sufficiently smooth to ensure that these paths also belong to the appropriate 
H\"older spaces. 

It is convenient to define the perturbation of the marginal density slightly
differently in the form $f_t=f(1+t\fsc)$.  For a given bounded function $\fsc: \Z\to\RR$ with 
$\int \fsc f\,d\n=0$, and sufficiently small $|t|$, each $f_t$ is indeed a probability density. The advantage
of defining the perturbation by $f\fsc$ instead of $\fsc$ is simply that in the present form $\fsc=d/dt_{|t=0}\log f_t$
can be interpreted as the score function of the model $t\mapsto f_t$.

These paths are usually enough to identify influence functions. By slightly
changing the definitions one might also allow non-bounded functions as ``directions''
of the perturbations.
\end{example}

\subsection{Relation to semiparametric theory and tangent spaces}
In semiparametric theory (e.g.\ \cite{BKRW, Pfanzagl82, vdV88,vanderVaart91})
influence functions are described through inner products with score functions. We do not
follow this route here, but make the connection in this section. Scores give a way of 
rewriting (\ref{EqHODerivatives}), which will be useful mainly
for first order influence functions.

For a sufficiently regular submodel $t\mapsto p_t$ equation (\ref{EqHODerivatives}) 
for $m=1$ can be written in the form
\begin{equation}
\label{EqFirstOrderIF}
\frac{d}{dt}_{|t=0}\chi(p_t)=\frac{d}{dt}_{|t=0} P_t\chi_p=P(\chi_p g),
\end{equation}
where $g=(d/dt)_{t=0}p_t/p$ is the \emph{score function} of the
model $t\mapsto p_t$ at $t=0$. 
A function $\chi_p$ satisfying (\ref{EqFirstOrderIF}) is exactly what
is called an \emph{influence function} in semiparametric
theory. The linear span of all scores attached some submodel $t\mapsto p_t$ is called
the \emph{tangent space} of the model at $p$ and an influence function is an element of
$L_2(p)$ whose inner products with the elements
of the tangent space represent the derivative of the functional in the sense  of (\ref{EqFirstOrderIF})
(\cite{vanderVaart98}, page~363, or \cite{BKRW, Pfanzagl82, vdV88,vanderVaart91}).

%\label{SectionTangentSpaceMissing}
\begin{example}
[Missing data, score functions]\normalfont
To obtain the score functions at $t=0$ of the one-dimensional submodels 
$t\mapsto p_t:=p_{a_t,b_t,f_t}$ induced by paths of the form  $a_t=a+t\asc$, $b_t=b+t\bsc$,  and $f_t=f(1+t\fsc)$,
for given measurable functions $\asc, \bsc, \fsc: \Z\to\RR$
(where $\int \fsc f\,d\n=0$), we substitute these paths in the right side of equation
(\ref{EqMarLikelihood}) for the likelihood, take the logarithm, and differentiate
at $t=0$. If we insert the perturbations for the three parameters separately, keeping the other
parameters fixed, we obtain what could be called ``partial score functions'' given by
\begin{align*}
B_p^a\asc(X)&=-\frac{A a(Z)-1}{a(Z)(a-1)(Z)}\asc(Z),
\qquad  a-\text{score},\\
B_p^b\bsc(X)&=\frac{A\bigl(Y-b(Z)\bigr)}{b(Z)(1-b)(Z)}\bsc(Z),
\qquad\  b-\text{score}, \\ 
B_p^f\fsc(X)&=\fsc(Z),
\qqqquad\qqqquad\ \ \ f-\text{score}.
\end{align*}
The scores are deliberately written in a form suggesting operators $B_p^a, B_p^b, B_p^f$ working
on the three directions $\asc, \bsc,\fsc$. These are called \emph{score operators} 
in semiparametric theory, and their direct sum is the overall score operator, which we write as $B_p$.
Thus $B_p(\asc,\bsc,\fsc)(X)$ is defined as the sum of the three left sides of the preceding equation.

We claim that the first-order influence function of the functional
$\chi: p_{a,b,f}\mapsto \int bf\,d\nu$ is given by 
\begin{equation}
\label{EqMARFOIF}
\chi_p^{(1)}(X)=A a(Z)\bigl(Y-b(Z)\bigr)+b(Z)-\chi(p).
\end{equation}
To prove this well-known fact, it suffices to verify that this function satisfies,
for every path $t\mapsto p_t$ as described previously,
$$\frac d{dt}_{|t=0}\chi(p_t)
=\E_p \bigl[\chi_p^{(1)}(X)\,B_p(\asc,\bsc,\fsc)(X)\bigr].$$
This follows by straightforward calculations, where it suffices to verify the equation
for each of the three perturbations separately. For instance, for a perturbation of only the parameter
$a$, the left side of the display is clearly zero, as the functional does not depend on $a$.
The right side with $\bsc=\fsc=0$ reduces to $\E_p \bigl[\chi_p^{(1)}(X) B_p^a\asc(X)\bigr]$,
which can be seen to be zero from the fact that $Aa(Z)-1$ and $Y-b(Z)$ are uncorrelated given $Z$.
The validity of the display for the two other types of scores can be verified similarly.

The advantage of choosing $a$ an inverse probability
is clear from the form of the (random part of the) influence function (\ref{EqMARFOIF}),
which is bilinear in $(a,b)$.

Computing (approximate) higher order influence functions for this model is a main achievement of
this paper. Expressions are given later on.
\end{example}

For  $m>1$  equation (\ref{EqHODerivatives})
can be expanded similarly in terms of inner products
of the influence function with score functions, but
``higher-order score functions'' arise next to ordinary score
functions. Here we do not follow this route, but have defined an higher
order influence function through (\ref{EqHODerivatives}), and
leave the alternative route to other papers. Suitable higher-order tangent spaces are discussed 
in \cite{RobinsetalFreedman} (also see \cite{vdVStatScience}), using score functions as defined in
\cite{watermanlindsay}. A discussion of \emph{second order} scores and tangent spaces can be
found in \cite{RobinsetalMetrika}. Second order tangent spaces
are also discussed in \cite{Pfanzagl85}, from a different point
of view of and with the  purpose of defining \emph{higher order efficiency}
of estimators. Higher-order efficient estimators attain the first order efficiency bound
(the ``asymptotic Cram\'er-Rao bound'') and also optimize certain lower order terms
in their distribution or risk. In the present paper we are interested in \emph{first order efficiency},
measured mostly by the convergence rate, which in the most interesting cases is slower than 
$\sqrt n$, and not in refinements of the first order behaviour.

\subsection{Computing the influence function}
\label{SectionComputingIF}
Equation (\ref{EqHODerivatives}) involves multiple derivatives and many paths
and is not easy to solve for $\chi_p$. For actual computation of an influence function
it is usually easier to derive higher order influence functions as influence functions
of lower order ones. 

To describe this operation, we need to decompose the influence function $\chi_p$,
or rather its symmetrized version in degenerate functions.
Any $m$th order, zero-mean $U$-statistic can be decomposed as the
sum of $m$ degenerate $U$-statistics of orders $1, 2,\ldots, m$,
by way of its Hoeffding decomposition. In the present situation
we can write
$$\UU_n\chi_{p}=\UU_n \chi^{(1)}_p
+\thalf \UU_n \chi^{(2)}_p
+\cdots+\frac1{m!}\UU_n \chi^{(m)}_p,$$
where $\chi^{(j)}_p: \X^j\to\RR$ is a degenerate kernel of
$j$ arguments, defined uniquely as a projection of $\chi_p$ (cf.\ \cite{vanZwet} and (\ref{EqMakeDegenerate})). 
Since $\chi_p$ is a function of $m$ arguments, for $m=n$ the left side evaluates to the symmetrization of the 
function $\chi_p$, and it is equal to $\chi_p$ if $\chi_p$ is already permutation symmetric in its arguments.
The functions on the right side are similarly symmetric, and the equation can be read
as a decomposition of the symmetrized version of $\chi_p$ into symmetrizations of certain degenerate
functions $\chi_p^{(j)}$. 
Suitable (symmetric) functions $\chi_p^{(j)}$ in this decomposition can be found by the following algorithm:
\begin{enumerate}
\item[{[1]}] Let $x_1\mapsto\bar\chi_p^{(1)}(x_1)$ be a  first order influence function
of the functional $p\mapsto \chi(p)$.
\item[{[2]}] Let $x_j\mapsto\bar \chi_p^{(j)}(x_1,\ldots, x_j)$ 
be a  first order influence function
of the functional $p\mapsto\bar\chi_p^{(j-1)}(x_1,\ldots, x_{j-1})$,
for each $x_1,\ldots, x_{j-1}$, and $j=2,\ldots, m$.
\item[{[3]}] Let $\chi_p^{(j)}=D_P\bar\chi_p^{(j)}$ be 
the degenerate part of $\bar\chi_p^{(j)}$ relative to $P$, as
defined in (\ref{EqMakeDegenerate}).
\end{enumerate}
See Lemma~\ref{LemmaRecursiveIF} for a proof.
Thus higher order influence functions are constructed
as first order influence functions of influence functions. 
Somewhat abusing language we shall refer to the function
$\chi_p^{(j)}$ also as a ``$j$th order influence function''.
The overall order $m$ will be fixed at a suitable value; for simplicity
we do not let this show up in the notation $\chi_p$.

% Because only inner products with scores matter, an ``influence function'' is unique only up to projections onto the tangent space and (averaging over) permutations of its arguments. 
%With any choice of influence functions the algorithm produces some influence function. In particular,
The starting influence function $\bar\chi_p^{(1)}$ in step [1]
may be any first order influence function (thus satisfying (\ref{EqHODerivatives}) for $m=1$,
or alternatively a function $\chi_p$ that satisfies (\ref{EqFirstOrderIF}) for every score $g$); it does not
have to possess mean zero, or be an element of the first order tangent space. 
A similar remark applies to the (first order) influence
functions found in step [2]. It is only in step [3] that we
make the influence functions degenerate. 

\begin{example}
[Missing data, higher order scores]\normalfont
The second order score function for the missing data problem is computed 
as a derivative of the first order score function (\ref{EqMARFOIF}) in Section~\ref{SectionMarSecondOrder}.
As will be explained momentarily the result (\ref{EqMARAdhocSecondOrderIF}) is actually
only a partial second order score function.

Higher order score functions are computed in Sections~\ref{SectionParametricRate}
and~\ref{SectionMinimaxRate}.
\end{example}

\subsection{Bias-variance trade-off}
\label{SectionBiasVarianceTradeOff}
Because it is centered, the ``variance part'' in (\ref{EqBiasVarianceUEstimator}),
the variable $(\UU_n-P^m)\chi_{\hat p_n}$,
should not change noticeably if we replace $\hat p_n$ by $p$, and be of the
same order as $(\UU_n-P^m)\chi_p$.  For a fixed square-integrable
function $\chi_p$ the latter centered $U$-statistic is well known to
be of order $O_P(n^{-1/2})$, and asymptotically normal if suitably
scaled.  A completely successful representation of the ``bias''
$R_n=\chi(\hat p_n)-\chi(p)+P^m\chi_{\hat p_n}$ in
(\ref{EqBiasVarianceUEstimator}) would lead to an error
$R_n=O_P\bigl(d(\hat p_n,p)^{m+1}\bigr)$, which becomes smaller with
increasing order $m$. Were this achievable for any $m$, then a $\sqrt
n$-estimator would exist no matter how slow the convergence rate
$d(\hat p_n,p)$ of the initial estimator. Not surprisingly, in many
cases of interest this ideal situation is not real. This is due to the
non-existence of influence functions that can exactly represent the Taylor
expansion of $\chi(\hat p_n)-\chi(p)$. 
%The technical reason is that multi-linear maps (even smooth ones) may not be
%representable through kernels. (See the discussion following (\ref{EqSecondOrderIP}) below.)

In general, we have to content ourselves with a partial representation.
Next to a first bias in the form of the remainder term $R_n$ of order $O_P\bigl(d(\hat p_n,p)^{m+1}\bigr)$,
we then also incur a ``representation bias''. The latter
bias can be made arbitrarily small by choice of the 
influence function, but only at the cost of increasing
its variance. We thus obtain
a trade-off between a variance and two biases. This typically 
results in a variance that is larger than $1/n$, and
a rate of convergence that is slower than $1/\sqrt n$, although sometimes
a nontrivial bias correction is possible without increasing the variance.

\begin{example}
[Missing data, variance and bias terms]\normalfont
The missing data problem is parameterized by the triple $(a,b,g)$ and hence 
the preliminary estimator $\hat p$ is constructed from estimates $\hat a$ and $\hat b$ and $\hat g$
of these parameters.

The remainder bias $R_n$ of the estimator for $m=1$ is given in (\ref{EqMARBias}). It is bounded by
$\|\hat a-a\|_2\,\|\hat b-b\|_2$ and hence is quadratic in the preliminary estimator, as expected.
There is no representation bias at this order. The variance of the linear estimator is of order $1/n$.
If the preliminary estimators can be constructed so that the product $\|\hat a-a\|_2\,\|\hat b-b\|_2$
is of lower or equal order than $1/n$, then the estimator is rate-optimal. Otherwise a higher order
estimator is preferable.

The bias and variance terms of the estimator for $m=2$ are given in Theorem~\ref{TheoremSecondOrderEstimator}.
The remainder bias $R_n$ is of the order $\|\hat a-a\|_r\,\|\hat b -b\|_r\,\|\hat g-g\|_r$, cubic in the preliminary estimator,
while the representation bias is of the order the product of the remainders after
projecting $\hat a-a$ and $\hat b-b$ onto a linear space chosen by the statistician.
The dimension $k$ of this space determines the variance of the estimator, adding a contribution of the order $k/n^2$.
Following the statement of the theorem it is shown how the variance can be traded off versus
the two biases. It turns concluded that in case the remainder bias of order 
$\|\hat a-a\|_r\,\|\hat b -b\|_r\,\|\hat g-g\|_r$ actively determines the outcome of this trade-off,
then an estimator of higher order is preferable.

For higher orders $m>2$ the remainder bias decreases to $\|\hat a-a\|_r\,\|\hat b -b\|_r\,\|\hat g-g\|_r^{m-1}$,
but the representation bias becomes increasingly complex. A discussion is deferred to 
Sections~\ref{SectionParametricRate} and~\ref{SectionMinimaxRate}.
\end{example}

\subsection{Approximate functionals}
\label{SectionApproximateFunctionals}
An attractive method to find
approximating influence functions is to compute exact
influence functions for an approximate functional. 
Because smooth functionals on finite-dimensional models
typically possess influence functions to any order,
projections on finite-dimensional models may deliver such
approximations.

A simple approximation would be 
$\chi(\tilde p)$ for a given
map $p\mapsto \tilde p$ mapping the model $\P$ onto
a suitable ``smaller'' model $\tilde\P$ (typically a submodel 
$\tilde\P\subset \P$).
A closer approximation can be obtained by also including 
a derivative term. Consider the functional
$\tilde\chi: \P\to \RR$ defined by, for a given map $p\mapsto \tilde p$,
\begin{equation}
\label{EqApproximatingFunctional}
\tilde\chi(p)=\chi(\tilde p)+P \chi_{\tilde p}^{(1)}.
\end{equation}
(A complete notation would be $\tilde p(p)$; the right hand side depends on 
$p$ at three places.)
By the definition of an influence function 
the term $-P \chi_{\tilde p}^{(1)}$
acts as the first order Taylor expansion of $\chi(\tilde p)-\chi(p)$.
Consequently, we may expect that
\begin{equation}
\label{EqErrorApproximateFunctional}
\bigl|\tilde\chi(p)-\chi(p)\bigr|=O\bigl(d(\tilde p,p)^2\bigr).
\end{equation}
This ought to be true for any ``projection'' $p\mapsto \tilde p$. 
If we choose the projection such that, for any path $t\mapsto p_t$,
\begin{equation}
\label{EqCleverApproximatingFunctional}
\frac d{dt}_{|t=0} \Bigl(\chi(\tilde p_t)
+P_0\chi_{\tilde p_t}^{(1)} \Bigr)=0,
\end{equation}
then  the functional $p\mapsto \tilde\chi(p)$ will be locally
(around $p_0$) equivalent to the functional 
$p\mapsto \chi\bigl(\tilde p_0\bigr)+P \chi_{\tilde p_0}^{(1)}$
(which depends on $p$ in only one place, $p_0$ being fixed) in the sense that
the first order influence functions are the same.
The first order influence function of the latter (linear) functional
at $p_0$ is equal to $\chi_{\tilde p_0}^{(1)}$, and hence
for a projection satisfying (\ref{EqCleverApproximatingFunctional})
the first order influence function
of the functional $p\mapsto \tilde\chi(p)$ will be
\begin{equation}
\label{EqIFApproximateFunctional}
\tilde\chi_p^{(1)}=\chi_{\tilde p}^{(1)}.
\end{equation}
In words, this means that the influence function of the approximating
functional $\tilde\chi$ satisfying (\ref{EqApproximatingFunctional})
and (\ref{EqCleverApproximatingFunctional}) at $p$ is obtained by 
substituting $\tilde p$ for $p$ in the influence function of the
original functional. 

This is relevant when obtaining  higher order influence functions. 
As these are recursive derivatives of the first order influence
function (see [1]--[3] in Section~\ref{SubSectionInfluenceFunctions}), 
the preceding display shows that we must compute influence functions of 
$$p\mapsto \chi_{\tilde p}^{(1)}(x),$$
i.e.\ we ``differentiate on the  model $\tilde\P$''. 
If the latter model is sufficiently
simple, for instance finite-dimensional, then exact
higher order influence functions of the functional
$p\mapsto \tilde\chi(p)$ ought to exist. We can use these as
approximate influence functions of $p\mapsto \chi(p)$.

%For instance, what if in the MAR example  we use different
%spaces??

%??Because $\chi(\tilde p)+P \chi_{\tilde p}^{(1)}$
%is approximately a quadratic, equation 
%(\ref{EqCleverApproximatingFunctional}) will
%typically define a (weighted) projection.??

\begin{example}
[Missing data, approximate functional]\normalfont
In the missing data problem the density $p$ corresponds one-to-one to
a triplet of parameters $(a,b,g)$ and hence the projection $p\mapsto \tilde p$
can be described as projections of the parameters. We leave $g$ invariant,
and map $a$ and $b$ onto a finite-dimensional affine space, as follows.

We fix a given finite-dimensional subspace $L$ of $L_2(\n)$ that has
good approximation properties for our model classes, the H\"older spaces
$C^\a(\Z)$ and $C^\b(\Z)$, for instance constructed from a  wavelet basis.
For fixed functions $\hat a,\aa, \hat b, \bb: \Z\to \RR^+$ 
we now let $\tilde a$ and $\tilde b$ be the functions such that 
$(\tilde a-\hat a)/\aa$ and $(\tilde b-\hat b)/\bb$ are the orthogonal projections 
of the functions $(a-\hat a)/\aa$ and $(b-\hat b)/\bb$ onto $L$ in $L_2(\aa\bb g)$.
Finally we define the map $p\mapsto \tilde p$ by  correspondence to 
$(a,b,g)\mapsto (\tilde a,\tilde b, g)$.

In Section~\ref{SectionMARApproximateFunctional} we shall see that the
orthogonal projections follow (\ref{EqCleverApproximatingFunctional}), while 
the concrete form of (\ref{EqErrorApproximateFunctional}) is valid in that
$$\Bigl|\int abg\,d\nu-\int \tilde a\tilde b g\,d\nu\Bigr|^2
\le\int \Bigl|\frac{a-\tilde a}\aa\Bigr|^2\,\aa\bb g\,d\nu\,
\int\Bigl|\frac{b-\tilde b}\bb\Bigr|^2\,\aa\bb g\,d\nu.$$
%\le\int \Bigl|\frac{a-\hat a}\aa-\frac{\tilde a-\hat a}\aa\Bigr|^2\,\aa\bb g\,d\nu
%\int\Bigl|\frac{b-\hat b}\bb-\frac{\tilde b-\hat b}\bb\Bigr|^2\,\aa\bb g\,d\nu.$$
This approximation error can be made arbitrarily small by making the space $L$ large enough.
In that case the approximate functional $p\mapsto \int \tilde a\tilde b g\,d\nu$ 
is close to the parameter of interest, and we may focus instead on estimating
this functional. The advantage is that by construction this depends only on finitely many unknowns,
e.g.\ the coefficients of $(\tilde a-\hat a)/\aa$ and $(\tilde b-\hat b)/\bb$ in a basis
of $L$. Higher order influence functions exist to any order.

The bias-variance trade-off of Section~\ref{SectionBiasVarianceTradeOff} arises
as the approximation error must be traded off against the ``variance of estimating the coefficients''
as well as against the remainder of using an $m$th order estimator.
\end{example}

\section{First order estimator}
\label{SectionMarFirstOrder}
The first order estimator (\ref{EqLinearEstimator}) is well studied for the missing data problem.
The first order influence function is given in (\ref{EqMARFOIF}), where $\chi_p=\chi_p^{(1)}$.
As it depends on the parameter $(a,b,f)$ only through $a$ and $b$, preliminary
estimators $\hat a $ and $\hat b$ suffice.

The ``first order bias'' of this estimator, the first term in (\ref{EqBiasVarianceLinearEstimator}),
can explicitly be computed as
\begin{align}
\label{EqMARBias}
\chi(\hat p)-\chi(p) + P\chi_{\hat p}^{(1)}
&=\E_p\bigl[(A\hat a(Z)-1)(Y-\hat b(Z))+\hat b(Z)\bigr]-\int bf\,d\nu\nonumber\\
&=-\int (\hat a-a)(\hat b-b)\,g\,d\n.
\end{align}
In agreement with the heuristics given in 
Sections~\ref{SectionIntroduction} and~\ref{SectionHeuristics} 
this bias is quadratic in the errors of the initial estimator. 

Actually, the form of the bias term is special in that square estimation errors $(\hat a-a)^2$ and $(\hat b-b)^2$ of
the two initial estimators $\hat a$ and $\hat b$ do not arise, but
only the product $(\hat a-a)(\hat b-b)$ of their errors. This property, termed
``double robustness'' in \cite{vdLaanRobins}, makes that for first order inference it suffices
that one of the two parameters be estimated well. A
prior assumption that the parameters $a$ and $b$ are $\a$ and $\b$ regular,
respectively, would allow estimation errors of the orders
 $n^{-\a/(2\a+d)}$ and 
$n^{-\b/(2\b+d)}$. If the product of these rates is $O(n^{-1/2})$,
then the bias term matches the variance. This leads
to the (unnecessarily restrictive) condition (\ref{EqEnoughSmoothnessForRootnByLinear}).

If the preliminary estimators $\hat a$ and $\hat b$ are 
solely selected for having small errors $\|\hat a-a\|$ and $\|\hat b-b\|$
(e.g.\ minimax in the $L_2$-norm), 
then it is hard to see why (\ref{EqMARBias})
would be small unless the product 
$\|\hat a-a\|\|\hat b-b\|$ of the errors is small. Special estimators
might exploit that the bias is an integral, in which cancellation of errors could occur.
As we do not wish to use special estimators,
our approach will be to replace the linear estimating
equation by a higher order one, leading to an analogue of (\ref{EqMARBias})
that is a cubic or higher order polynomial of the estimation errors.

As noted the marginal density $f$ (or $g$) does not enter into
the first order influence function (\ref{EqMARFOIF}). Even though the functional
depends on $f$ (or $g$), a rate on the initial estimator
of this function is not needed for the construction of the
first order estimator. This will be different
at higher orders. 

\section{Second order estimator}
\label{SectionMarSecondOrder}
In this section we derive a second order influence function for the
missing data problem,
and analyze the risk of the corresponding estimator. This estimator
is minimax if $(\a+\b)/2\ge d/4$ and
\begin{equation}
\label{EqCutoffSecondOrder}
\frac{\g}{2\g+d}\ge \frac12\wedge \frac{2\a+2\b}{d+2\a+2\b}  -\frac{\a}{2\a+d}-\frac{\b}{2\b+d}.
\end{equation}
In the other case, higher order estimators have smaller risk, as
shown in Sections~\ref{SectionParametricRate}-\ref{SectionMinimaxRate}.
However, it is worth while to treat the second order estimator
separately, as its construction exemplifies essential elements, 
without involving technicalities attached to the higher order estimators.

To find a second order influence
function, we follow the strategy [1]--[3] of 
Section~\ref{SubSectionInfluenceFunctions}, 
and try and find a function $\chi_p^{(2)}:\X^2\to\RR$ such that,
for every $x_1=(y_1a_1,a_1,z_1)$, and all directions $\asc, \bsc,\fsc$,
$$\frac d{dt}_{|t=0}\Bigl[\chi_{p_t}^{(1)}(x_1)+\chi(p_t)\Bigr]
=\E_p \chi_p^{(2)}(x_1,X_2)\, B_p(\asc,\bsc,\fsc)(X_2).$$
Here the expectation $\E_p$ on the right side
 is relative to the variable $X_2$ only, with $x_1$ fixed. This
equation expresses that $x_2\mapsto \chi_p^{(2)}(x_1,x_2)$ is a first order
influence function of $p\mapsto \chi_p^{(1)}(x_1)+\chi(p)$, for fixed $x_1$.
On the left side we added the ``constant''
$\chi(p_t)$ to the first order influence function
(giving another first order influence function) to facilitate 
the computations. This is justified as the strategy [1]--[3] works with
any influence function. 
In view of (\ref{EqMARFOIF}) and the definitions of the
paths $t\mapsto a+t\asc$, $t\mapsto b+t\bsc$ and $t\mapsto f(1+t\fsc)$,
this leads to the equation
\begin{align}
&a_1\bigl(y_1-b(z_1)\bigr)\asc(z_1)-\bigl(a_1a(z_1)-1\bigr)\bsc(z_1)\nonumber\\
&\hskip5cm=\E_p \chi_p^{(2)}(x_1,X_2)
\, B_p(\asc,\bsc,\fsc)(X_2).\label{EqSecondOrderIP}
\end{align}
Unfortunately, no function $\chi_p^{(2)}$ that solves
this equation for every $(\asc,\bsc,\fsc)$ exists. To see this note that
for the special triplets with $\bsc=\fsc=0$ the requirement can be written in
the form
$$\asc(z_1)=\E_p\left[\frac{\chi_p^{(2)}(x_1,X_2)}{a_1\bigl(y_1-b(z_1)\bigr)}
\frac{1-A_2 a(Z_2)}{a(Z_2)(a-1)(Z_2)}\right]\asc(Z_2).$$
The right side of the equation can be written as 
$\int K(z_1,z_2)\asc(z_2)\,dF(z_2)$, for $K(z_1,Z_2)$ 
the conditional expectation of the function in square brackets given $Z_2$.
Thus it is the image of $\asc$ under the kernel operator with kernel $K$.
If the equation were true for any $\asc$, then this kernel operator would work 
as the identity operator. However, on infinite-dimensional domains the identity operator
is not given by a kernel. (Its kernel would be
a ``Dirac function on the diagonal''.) 

Therefore, we have to be satisfied with an influence function that gives
a partial representation only. In particular, a projection
onto a finite-dimensional linear space possesses a kernel,
and acts as the identity on this linear space. A ``large'' linear
space gives representation in ``many'' directions. By reducing
the expectation in (\ref{EqSecondOrderIP}) to an integral
relative to the marginal distribution of $Z_2$, we can use
an orthogonal projection $\Pi_p: L_2(g)\to L_2(g)$ onto a subspace 
$L$ of $L_2(g)$. Writing also $\Pi_p$ for its kernel,
and letting $S_2h$ denote the symmetrization 
$\bigl(h(X_1,X_2)+h(X_2,X_1)\bigr)/2$ of a function $h: \X^2\to \RR$,
we define
\begin{equation}
\chi_p^{(2)}(X_1,X_2)
=-2S_2\Bigl[A_1\bigl(Y_1-b(Z_1)\bigr)\Pi_p(Z_1,Z_2)
\bigl(A_2a(Z_2)-1\bigr)\Bigr].
\label{EqMARAdhocSecondOrderIF}
\end{equation}

\begin{lemma}
\label{LemmaSecondOrderIF}
For $\chi_p^{(2)}$ defined by (\ref{EqMARAdhocSecondOrderIF})
with $\Pi_p$ the kernel of an
orthogonal projection $\Pi_p: L_2(g)\to L_2(g)$ onto a subspace 
$L\subset L_2(g)$, equation (\ref{EqSecondOrderIP}) 
is satisfied for every path
$t\mapsto p_t$ corresponding to directions $(\asc,\bsc,\fsc)$ such
that $\asc\in L$ and $\bsc\in L$.
\end{lemma}

\begin{proof}
By definition $\E(A\given Z)=(1/a)(Z)$ and $\E(Y\given Z)=b(Z)$.
Also  $\var\bigl(Aa(Z)\given Z\bigr)=a(Z)-1$ and
$\var(Y\given Z)=b(Z)(1-b)(Z)$.
By direct computation using these identities, we find that for the
influence function (\ref{EqMARAdhocSecondOrderIF}) the right side
of (\ref{EqSecondOrderIP}) reduces to 
$$a_1\bigl(y_1-b(z_1)\bigr)\Pi_p\asc(z_1)
-\bigl(a_1a(z_1)-1\bigr) \Pi_p\b(z_1).$$
Thus (\ref{EqSecondOrderIP}) holds for every $(\asc,\bsc,\fsc)$ 
such that $\Pi_p\asc=\asc$ and $\Pi_p\bsc=\bsc$. 
\end{proof}

Together with the first order influence function 
(\ref{EqMARFOIF}) the influence function (\ref{EqMARAdhocSecondOrderIF})
defines the (approximate) influence function
$\chi_p=\chi_p^{(1)}+\thalf \chi_p^{(2)}$. For an
initial estimator $\hat p$ based on independent 
observations we now construct the estimator (\ref{EqEstimator}), i.e.
\begin{equation}
\label{EqEstimatorSecondOrder}
\hat\chi_n=\chi(\hat p)+\PP_n\chi_{\hat p}^{(1)}+
\thalf \UU_n\chi_{\hat p}^{(2)}.
\end{equation}
Unlike the first order influence function, 
the second order influence function does depend on the
density $f$ of the covariates, or rather the
function $g=f/a$ (through the kernel $\Pi_p$, which is defined relative to $L_2(g)$), and
hence the estimator (\ref{EqEstimatorSecondOrder}) involves a preliminary estimator of $g$. 
As a consequence, the quality of the estimator of the functional $\chi$
depends on the precision by which $g$ (as part of the plug-in $\hat p=(\hat a,\hat b,\hat g)$) 
can be estimated. The intuitive reason is that the bias (\ref{EqMARBias}) depends
on $g$, and it can only be made smaller by estimating it.

Let $\hat\E_p$ and $\hat\var_p$ denote conditional expectations given
the observations used to construct $\hat p$, let $\|\cdot\|_r$ be
the norm of $L_r(g)$, and let $\|\Pi\|_r$ denote
the norm of an operator $\Pi: L_r(g)\to L_r(g)$.

\begin{theorem}
\label{TheoremSecondOrderEstimator}
The estimator $\hat\chi_n$ given in (\ref{EqEstimatorSecondOrder}) 
with influence functions
$\chi_p^{(1)}$ and $\chi_p^{(2)}$ defined by 
(\ref{EqMARFOIF}) and (\ref{EqMARAdhocSecondOrderIF}), for
$\Pi_p$ the kernel of an orthogonal projection in $L_2(g)$ onto 
a $k$-dimensional linear subspace, satisfies, for $r\ge 2$ (with $r/(r-2)=\infty$ if $r=2$),
\begin{align*}
\hat\E_p\hat \chi_n-\chi(p)&=O\Bigl(\|\Pi_p\|_r
\|\Pi_{\hat p}\|_r\|\hat a-a\|_r\|\hat b-b\|_r \|\hat g-g\|_{r/(r-2)}\Bigr)\\
&\qqqquad+O\Bigl(\bigl\|(I- \Pi_p) (a-\hat a)\bigr\|_2
\bigr\|(I-\Pi_p) (b-\hat b)\bigr\|_2\Bigr),\\
\hat\var_p \hat\chi_n&= O\Bigl( \frac 1n+\frac k{n^2}\Bigr).
\end{align*}
\end{theorem}

The two terms in the bias result from having to estimate
$p$ in the second order influence function (giving ``third order
bias'') and using an approximate influence function (leaving the remainders $I-\Pi_p$
after projection), respectively.
The terms $1/n$ and $k/n^2$ in the variance appear as the variances
of $\UU_n\chi_p^{(1)}$ and $\UU_n\chi_p^{(2)}$, the second being 
a degenerate second order $U$-statistic (giving $1/n^2$, see (\ref{EqVarDegU})) with a kernel of variance $k$.

The proof of the theorem is deferred to 
Section~\ref{SectionProofofTheoremSecondOrderEstimator}.

Assume now that the range space of the projections $\Pi_p$ can be chosen
such that, for some constant $C$,
\begin{equation}
\label{EqApproxByPi}
\|a-\Pi_p a\|_2\le C \Bigl(\frac 1 k\Bigr)^{\a/d},
\qquad 
\|b-\Pi_p b\|_2\le C \Bigl(\frac 1 k\Bigr)^{\b/d}.
\end{equation}
Furthermore, assume that there exist estimators $\hat a$ and $\hat b$
and $\hat g$ that achieve convergence rates $n^{-\a/(2\a+d)}$,
$n^{-\b/(2\b+d)}$ and $n^{-\g/(2\g+d)}$, respectively, in $L_r(g)$ and $L_{r/(r-2)}(g)$, uniformly
over these a-priori models and a model for $g$ (e.g.\ for $r=3$), and that the preceding
displays also hold for $\hat a$ and $\hat b$. These assumptions
are satisfied if the unknown functions $a$ and $b$ are 
``regular'' of orders $\a$ and $\b$ on a compact subset of $\RR^d$
(see e.g.\ \cite{tsybakov}).
Then the estimator $\hat\chi_n$ of Theorem~\ref{TheoremSecondOrderEstimator}
attains the square rate of convergence 
\begin{equation}
\label{EqRate2ndOrder}
\Bigl(\frac 1n\Bigr)^{2\a/(2\a+d)
+2\b/(2\b+d)+2\g/(2\g+d)}\vee \Bigl(\frac 1k\Bigr)^{(2\a+2\b)/d}
\vee \frac 1 n\vee \frac k{n^2}.
\end{equation}
We shall see in the next section that the first of the four
terms in this maximum can be made smaller by choosing an
influence function of order higher than 2, 
while the other three terms arise at any order.
This motivates to determine a ``second order `optimal'' value of $k$ by
balancing the second, third and fourth terms. We next would
use the second order estimator if $\g$ is large enough so that
the first term is negligible relative to the other terms.

For $(\a+\b)/2\ge d/4$ we can choose $k=n$ and the resulting rate
(the square root of (\ref{EqRate2ndOrder}))
is $n^{-1/2}$ provided that (\ref{EqCutoffSecondOrder}) holds.
The latter condition is certainly satisfied under the sufficient condition
(\ref{EqEnoughSmoothnessForRootnByLinear}) for the linear estimator to
yield rate $n^{-1/2}$.   

More interestingly, for $(\a+\b)/2< d/4$ we choose
$k\sim n^{2d/(d+2\a+2\b)}$ and obtain the rate, provided that (\ref{EqCutoffSecondOrder}) holds,
$$n^{-(2\a+2\b)/(d+2\a+2\b)}.$$
This rate is slower than $n^{-1/2}$, but better than the rate 
$n^{-\a/(2\a+d)-\b/(2\b+d)}$ obtained by the linear estimator.
In \cite{RobinsetalMinimax} this rate is shown to be the fastest
possible in the minimax sense, for the model in which $a$ and $b$
range over balls in $C^\a(\Z)$ and $C^\b(\Z)$, and $g$ being known.

In both cases the second order estimator is better than the linear
estimator, but minimax only for sufficiently large $\g$. This motivates
to consider higher order estimators.

\section{Approximate functional}
\label{SectionMARApproximateFunctional}
Even though the functional of interest does not possess an exact
second-order influence function, we might proceed to higher orders by
differentiating the approximate second-order influence function
$\chi_p^{(2)}$ given in (\ref{EqMARAdhocSecondOrderIF}), and balancing
the various terms obtained. However, the formulas are much more
transparent if we compute \emph{exact} higer-order influence functions of an
approximating functional instead. In this section we first define a suitable functional
and next compute its influence functions.

Following the heuristics of Section~\ref{SectionApproximateFunctionals},
we define an approximate functional by equation 
(\ref{EqApproximatingFunctional}), using a
particular projection $p\mapsto \tilde p$ of the parameters.
We choose this projection to map the parameters $a$ and $b$ onto finite-dimensional
models and leave the parameter $g$ unaltered: 
$p$ is mapped into an element $\tilde p$ of
the approximating model, or equivalently a triplet $(a,b,g)$
into a triplet $(\tilde a,\tilde b,g)$ in the approximating model
for the three parameters (where $g$ is unaltered). (Even though
this is not evident in the notation, the projection is joint
in the three parameters: the induced maps $(a,b,g)\mapsto\tilde a$ and 
$(a,b,g)\mapsto \tilde b$ do not reduce to maps
$a\mapsto \tilde a$ and $b\mapsto \tilde b$, but $\tilde a$ and
$\tilde b$ depend on the full triplet $(a,b,g)$.)

As ``model'' for $(a,b)$ we consider the
product of two affine linear spaces
\begin{equation}
\label{EqAffineModels}
\bigl(\hat a+\aa L\bigr)\times\bigl(\hat b+\bb L\bigr),
\end{equation}
for a given finite-dimensional subspace $L$ of $L_2(\n)$
and fixed functions $\hat a,\aa, \hat b, \bb: \Z\to \RR$ that are bounded away from zero and infinity.
(Later the functions $\hat a$ and $\hat b$ are taken equal to the
preliminary estimators; one choice for the other functions is $\aa=\bb=1$.) 
The pair $(\tilde a, \tilde b)$ of
projections are defined as elements of the model (\ref{EqAffineModels}) satisfying
equation (\ref{EqCleverApproximatingFunctional}). In view of 
(\ref{EqMARBias}), for any path $\tilde p_t\leftrightarrow (\tilde a_t,\tilde b_t, g)
=(\tilde a+t\aa\, l, \tilde b+t\bb\, l',g)$, for given $l,l'\in L$, 
\begin{equation}
\label{EqChiTilde2}
\chi(\tilde p_t)+P\chi_{\tilde p_t}^{(1)}=\chi(p)-\int \bigl(\tilde a+t\aa\, l-a\bigr) \bigl(\tilde b+t\bb\, l'-b\bigr)\,g\,d\n.
\end{equation}
Equation (\ref{EqCleverApproximatingFunctional}) requires that the derivative
of this expression with respect to $t$ at $t=0$ vanishes. Thus 
the functions $\tilde a$ and $\tilde b$ must be chosen to satisfy 
the set of stationary equations, for every $l,l'\in L$,
\begin{align}
\label{EqApproxa}
0&=\int (\tilde a-a)\bb\, l'\,g\,d\n=\int \Bigl(\frac{\tilde a-\hat a}{\aa}-
\frac{a-\hat a}{\aa}\Bigr) \,l'\,\aa\bb g\,d\n,\qquad l'\in L,\\
0&=\int \aa\, l(\tilde b-b)\,g\,d\n=\int \Bigl(\frac{\tilde b-\hat b}{\bb}-
\frac{b-\hat b}{\bb}\Bigr)\, l\, \aa\bb g\,d\n,\qquad l\in L.
\label{EqApproxb}
\end{align}
Because the functions $(\tilde a-\hat a)/\aa$ and $(\tilde b-\hat b)/\bb$ are required to be
in $L$, the second way of writing these equations shows that the latter two functions are
the orthogonal projections 
of the functions $(a-\hat a)/\aa$ and $(b-\hat b)/\bb$ onto $L$ in $L_2(\aa\bb g)$.
%\footnote{Equivalently, for $\Pi_p$ the weighted projection operator,
%\begin{align*}
%\tilde a&=\aa \Pi_p\Bigl(\frac a{\aa}\Bigr)+\aa (I-\Pi_p)\Bigl(\frac {\hat a}{\aa}\Bigr),\\
%\tilde b&=\bb \Pi_p\Bigl(\frac b{\bb}\Bigr)+\bb (I-\Pi_p)\Bigl(\frac {\hat b}{\bb}\Bigr).
%\end{align*}
%If the estimators $\hat a$ and $\hat b$ are chosen in the spaces $\aa L$ and $\bb L$, respectively,
%then the second terms on the right are zero, and the expression simplifies.}

As explained in Section~\ref{SectionApproximateFunctionals},
as it satisfies (\ref{EqCleverApproximatingFunctional})
the projection $(a,b,g)\mapsto (\tilde a, \tilde b,g)$
renders the first order influence function
of the approximate functional $\tilde\chi$ equal to the 
first order influence function of $\chi$ evaluated at the
projection. Furthermore, the difference between
$\chi$ and $\tilde \chi$ is quadratic in the distance
between $\tilde p$ and $p$ (see (\ref{EqErrorApproximateFunctional})). The 
following theorem summarizes the preceding and
verifies these properties in the present concrete situation.

\begin{theorem}
\label{TheoremCharacTildes}
For given measurable functions $\hat a, \aa, \hat b,\bb: \Z\to \RR$
with $\aa$ and $\bb$ bounded away from zero and infinity, define a map
$(a,b,g)\mapsto (\tilde a,\tilde b, g)$ by letting
$(\tilde a -\hat a)/\aa$ and $(\tilde b -\hat b)/\bb$ be
the orthogonal projections of $(a -\hat a)/\aa$ and $(b-\hat b)/\bb$ 
in $L_2(\aa\bb g)$ onto a closed subspace $L$. 
Let $\tilde p$ correspond to $(\tilde a,\tilde b,g)$
and define $\tilde\chi(p)=\chi(\tilde p)+P\chi_{\tilde p}^{(1)}$. Then 
$\tilde \chi$ has influence function
\begin{align}
\tilde\chi_p^{(1)}(X)
&=A \tilde a(Z)\bigl(Y-\tilde b(Z)\bigr)+\tilde b(Z)
-\chi\bigl(\tilde p\bigr).
%&=A \tilde a(Z)Y-\tilde b(Z)\bigl(A\tilde a(Z)-1\bigr)
%-\chi\bigl(\tilde p\bigr).
\label{EqMARIFApprox}
\end{align}
Furthermore, for $\underline g=\aa\bb g$,
$$\bigl|\tilde\chi(p)-\chi(p)\bigr|\le 
\Bigl\|(I-\Pi_p)\frac{\hat a- a}{\aa}\Bigr\|_{2,\underline g}
\Bigl\|(I-\Pi_p)\frac{\hat b- b}{\bb}\Bigr\|_{2,\underline g}.$$
\end{theorem}

\begin{proof}
The formula for the influence function agrees with the
combination of equations (\ref{EqIFApproximateFunctional}) 
and  (\ref{EqMARFOIF}), and can also be verified directly.
In view of (\ref{EqApproximatingFunctional}) and (\ref{EqMARBias}),
$$\tilde\chi(p)-\chi(p)
=-\int \bigl(\tilde a-a\bigr)\bigl(\tilde b-b\bigr)\, g\,d\n.$$
We rewrite the right side as an integral relative
to $\underline g\,d\n$, and next apply the Cauchy-Schwarz inequality.
Finally we note that $(\tilde a-a)/\aa=
(\tilde a-\hat a)/\aa -(a-\hat a)/\aa=(I-\Pi_p)\bigl((\hat a-a)/\aa\bigr)$,
and similarly for $b$.
\end{proof}

The approximation error $\tilde \chi(p)-\chi(p)$
can be rendered arbitrarily small by choosing the 
space $L$ large enough. Of course, we choose $L$ to
be appropriate relative to a-priori assumptions on the
functions $a$ and $b$. If these functions
are known to belong to H\"older classes, then  $L$ can for instance
be chosen as the linear span of the first $k$ basis elements
of a suitable orthonormal wavelet basis of $L_2(\n)$.

To compute higher order influence functions of $\tilde\chi$ we
recursively determine influence functions of influence functions,
according to the algorithm [1]--[3] in Section~\ref{SectionComputingIF},
starting with the influence function of $p\mapsto \tilde\chi_p^{(1)}(x_1)
+\chi(\tilde p)$, for a fixed $x_1$. We defer the details
of this derivation to Section~\ref{SectionDerivationHigherOrderIFApproximate},
and summarize the result in the following theorem.

To simplify notation, define
\begin{align}
\tilde Y&=A\bigl(Y-\tilde b(Z)\bigr)\aa(Z),\nonumber\\
\tilde A&=\bigl(A\tilde a(Z)-1\bigr)\bb(Z),\label{EqDeftildeYtildeA}\\
\AA  &=A \aa(Z)\bb(Z).\nonumber
\end{align}
These are the generic variables; indexed versions $\tilde Y_i, \tilde A_i,\AA_i,\ldots$ are defined
by adding an index to every variable in the equalities.
With this notation and with $\aa=\bb=1$ the second order influence 
function (\ref{EqMARAdhocSecondOrderIF}) at $p=\tilde p$
can be written as the symmetrization of 
$-2\tilde Y_1 \Pi_p(Z_1,Z_2)\tilde A_2$. This function was derived in an ad-hoc manner as
an approximate or partial influence function of $\chi$, but it is also the exact
influence function of $\tilde\chi$. The higher
order influence functions of $\tilde\chi$ possess an equally attractive form.

\begin{theorem}
\label{TheoremMARHigherOrderIF}
An $m$th order influence function $\tilde\chi^{(m)}_p$
evaluated at $(X_1,\ldots, X_m)$ of the functional $\tilde\chi$ 
defined in Theorem~\ref{TheoremCharacTildes}
is the degenerate (in $L_2(p)$) part of the variable
$$(-1)^{m-1}m!\,\tilde A_1 \Pi_{1,2}\AA_2 \Pi_{2,3}\AA_3\Pi_{3,4}\AA_4
\times\cdots\times \AA_{m-1}\Pi_{m-1,m}\tilde Y_m.$$
Here $\Pi_{i,j}$ is the kernel of the orthogonal projection
in $L_2(\aa\bb g)$ onto $L$, evaluated at $(Z_i,Z_j)$.
\end{theorem}

To obtain the degenerate part of the variable
in the preceding lemma, we apply the general formula
(\ref{EqMakeDegenerate}) together with 
Lemma~\ref{LemmaMARMakingDegenerate}. Assertions
(i) and (ii) of the latter lemma show that the
variable is already degenerate relative to $X_1$ and $X_m$,
while assertion (iii) shows 
that integrating out the variable $X_i$ for $1<i<m$ simply
collapses $\Pi_{i-1,i}\AA_i\Pi_{i,i+1}$ into $\Pi_{i-1,i+1}$.
For instance, with $S_m$ denoting symmetrization of a 
function of $m$ variables,
\beginskip
\footnote{Alternative expressions are obtained by expressing
the kernel $\Pi_p$ into a basis $e_1,\ldots, e_k$ of $L$.
By Lemma~\ref{LemmaProjectionKernelExpressedInBasis}
we have $\Pi_p(z_1,z_2)= \vec e_k(z_1)^T C_p^{-1}\vec e_k(z_2)$, for
$C_p$ the matrix with $(i,j)$th element 
$(C_p)_{ij}=\E_p e_i(Z)e_j(Z)\AA$ and 
$\vec e_k(z)=\bigl(e_1(z),\ldots, e_k(z)\bigr)$.
It follows that the variable
$\tilde A_1 \Pi_{1,2}\AA_2 \Pi_{2,3}\AA_3\Pi_{3,4}\AA_4
\times\cdots\times \AA_{m-1}\Pi_{m-1,m}\tilde Y_m$ can be
written as 
$$\tilde A_1 \vec e_k(Z_1)^T
\biggl[\prod_{j=2}^{m-1} C_p^{-1}\AA_j \vec e_k(Z_j)\vec e_k(Z_j)^T\biggr]
C_p^{-1}\vec e_k(Z_m)\tilde Y_m.$$
Because the degenerate part of
a product of independent variables 
is simply the product of the centered variables,
the degenerate part of this variable is
$$\tilde A_1 \vec e_k(Z_1)^T
\biggl[\prod_{j=2}^{m-1} C_p^{-1}\bigl(\AA_j \vec e_k(Z_j)\vec e_k(Z_j)^T
-C_p\bigr)\biggr] \,C_p^{-1}\vec e_k(Z_m)\tilde Y_m.$$
The $m$th order influence function, evaluated
at the observations, is the symmetrization of this variable.
Because the expression depends on the space $L$ only,
and not on the particular basis $e_1,\ldots, e_k$,
we may choose a basis that is orthogonal relative to the
weights $\AA$. Then $e_1,\ldots, e_k$ will depend on $p$
(which will be replaced by the initial estimator $\hat p$ in the
construction of the estimator), but $C_p$ reduces to the identity.}
\endskip
\begin{align}
\tilde\chi_p^{(2)}(X_1,X_2)
&=-2S_2[\tilde A_1\Pi_{1,2}\tilde Y_2],\nonumber\\
\tilde\chi^{(3)}_p(X_1,X_2, X_3)
&=6S_3\Bigl[\tilde A_1 \Pi_{1,2}\AA_2\Pi_{2,3}\tilde Y_3
-\tilde A_1 \Pi_{1,3}\tilde Y_3\Bigr],
\label{EqIFsConcrete}\\
\tilde\chi^{(4)}_p(X_1,X_2, X_3,X_4)
&=-24S_4\Bigl[\tilde A_1 \Pi_{1,2}\AA_2\Pi_{2,3}\AA_3\Pi_{3,4}\tilde Y_4
\nonumber\\
&\hskip-1cm-\tilde A_1 \Pi_{1,3}\AA_3\Pi_{3,4}\tilde Y_4
-\tilde A_1 \Pi_{1,2}\AA_2\Pi_{2,4}\tilde Y_4
+\tilde A_1 \Pi_{1,4}\tilde Y_4\Bigr].\nonumber
\end{align}
As shown on the left, but not on the right of the equations, these quantities depend on the unknown
parameter $p=(a,b,g)$. In the right sides, the variables $\tilde Y_i$ and $\tilde A_i$
depend on $p$ through $\tilde b$ and $\tilde a$, and hence are not observable variables.
Furthermore, the kernels $\Pi_{i,j}$  depend on $g$ as they are orthogonal projections
in $L_2(\aa\bb g)$.

\section{Parametric rate ($(\a+\b)/2\ge d/4$)}
\label{SectionParametricRate}
In this section we show that the parameter $\chi(p)$ is estimable
at $1/\sqrt n$-rate provided the average smoothness 
$(\a+\b)/2$ is at least $d/4$. We achieve this using the
estimator 
\begin{equation}
\label{EqEstimatorParametricRate}
\hat\chi_n=\chi(\hat p)+\UU_n\Bigl(\tilde\chi_{\hat p}^{(1)}
+\thalf\tilde\chi^{(2)}_{\hat p}+\cdots
+\frac1{m!}\tilde\chi_{\hat p}^{(m)}\Bigr),
\end{equation}
with the influence functions $\tilde\chi_p^{(j)}$ those of the approximate functional $\tilde\chi$
in Section~\ref{SectionMARApproximateFunctional}: they are given in
Theorems~\ref{TheoremCharacTildes} and~\ref{TheoremMARHigherOrderIF} for $j=1$, and
$j=2,\ldots, m$, respectively. (Because the map $p\mapsto \tilde p$ maps $\hat p$ into itself,
the influence function for $j=1$ in the display is also the first order influence function (\ref{EqMARIFApprox}) of
of $\chi$, when evaluated at $p=\hat p$.)

We assume that the projections $\Pi_p$ and $\Pi_{\hat p}$ 
map $L_s\bigl(\aa\bb g)$ to $L_s\bigl(\aa\bb g)$, for every $s\in \bigl[r/(r-1),r\bigr]\cup\{4\}$,
with uniformly bounded norms. 
(For $r=2$ this entails only $s\in\{2,4\}$; in this case we define $r/(r-2)=\infty$.)
%$s\in \bigl[r/(r-1),r(m-1)/(m+r-3)\bigr]$ is sufficient

\begin{theorem}
\label{TheoremParametricRate}
The estimator (\ref{EqEstimatorParametricRate}), with $\Pi_p$
a kernel of an orthogonal projection in $L_2\bigl(\aa\bb g)$
satisfying (\ref{EqProjectionKernel}) with $\sup_x \Pi_p(x,x)\lesssim k$,
satisfies, for a constant $c$ that depends on the supremum norms
of $\aa,\bb, 1/a, b, p/\hat p,g/\hat g$, the norms of the operators 
$\Pi_{\hat p}^{(0,l]}: L_s(\aa\bb \hat g)\to L_s(\aa\bb \hat g)$, for $l=1,\ldots, k$ only, and $r\ge 2$,
\begin{align*}
\hat\E_p\hat \chi_n-\chi(p)
&=O\Bigl(\|\hat a-a\|_{r}\|\hat b-b\|_{r}
\|\hat g-g\|_{(m-1)r/(r-2)}^{m-1}\Bigr)\\
&\qqqquad +O\Bigl(\Bigl\|(I-\Pi_p)\frac{\hat a- a}{\aa}\Bigr\|_2
\Bigl\|(I-\Pi_p)\frac{\hat b- b}{\bb}\Bigr\|_2\Bigr),\\
\hat\var_p \hat\chi_n
%&\lesssim \frac1n+\sum_{j=2}^m\sum_{l=1}^j\frac{c^jj^{2l}k^{l-1}}{n^l}\e_n^{2(j-l)}.
&\lesssim \sum_{l=1}^m\Bigl(\sum_{j=l}^mc^jj^{2l}\e_n^{2(j-l)}\Bigr)\frac{k^{l-1}}{n^l}.
\end{align*}
Here $\e_n$ is the maximum of the three rates $\|a-\hat a\|_4$, $\|b-\hat b\|_4$ and 
$\|g-\hat g\|_\infty$.
\end{theorem} 

The first term in the bias is of the order $1+1+(m-1)=m+1$, as to be expected 
for an estimator based on an $m$th order influence function;
the second term  is due to estimating $\tilde \chi$ rather
than $\chi$; it is independent of $m$, and the same as in
Theorem~\ref{TheoremSecondOrderEstimator} if $\aa=\bb=1$. The bound on the variance
is a sum of terms of the order $k^{j-1}/n^j$, which can roughly be understood in that each of the degenerate $U$-statistics
$\UU_n\tilde\chi_{\hat p}^{(j)}$ in (\ref{EqEstimatorParametricRate}) contributes a term of order $k^{j-1}/n^j$.
(The inner sums will typically be dominated by the terms with $j=l$, as the terms with  $l<j$ include a positive power of
the estimation error $\e_n$; the latter are lower order terms resulting from higher order $U$-statistics.)

For $\a$-, $\b$- and $\g$-regular parameters $a,b,g$ on a $d$-dimensional domain
the range space of the projections $\Pi_p$ can be chosen
so that (\ref{EqApproxByPi}) holds and such that there exist estimators $\hat a, \hat b,\hat g$ of $a, b,g$,
with the first two taking values in this range space, with convergence rates 
$n^{-\a/(2\a+d)}$, $n^{-\b/(2\b+d)}$ and $n^{-\g/(2\g+d)}$.
Then the second term in the bias (with $\aa=\bb=1$) is of order $(1/k)^{\a/d+\b/d}$.
If  $(\a+\b)/2\ge d/4$ and we choose $k=n$, then this is of order 
$1/\sqrt n$. For $k=n$ the standard deviation of the 
resulting estimator is also of the
order $1/\sqrt n$, while the first
term in the bias can be made arbitrarily small by choosing
a sufficiently large order $m$. Specifically, the
estimator $\hat\chi_n$ attains a $\sqrt n$-rate of convergence as soon as
\begin{equation}
\label{EqParametricOrder}
m-1\ge\Bigl(\frac12-\frac{\a}{2\a+d}-\frac{\b}{2\b+d}\Bigr)
\Bigl(\frac{2\g+d}\g\Bigr).
\end{equation}
For any $\g>0$ there exists an order $m$
that satisfies this, and hence
the parameter is $\sqrt n$-estimable as soon as
$(\a+\b)/2\ge d/4$.

More ambitiously, we may aim at attaining the parametric rate for every $\g>0$,
without a-priori knowledge of $\g$. This can be achieved if $(\a+\b)/2>d/4$
by using orders $m=m_n$ that increase to infinity with the sample size.
In this case the estimator can also be shown to be asymptotically efficient
in the semiparametric sense.

\begin{theorem}
\label{TheoremParametricRateEfficient}
If $(\a+\b)/2>d/4$, then 
the estimator (\ref{EqEstimatorParametricRate}), with $m=\log n$  and $\Pi_p$
a kernel of an orthogonal projection in $L_2\bigl(\aa\bb g)$ on a $k=n/(\log n)^3$-dimensional
space satisfying (\ref{EqApproxByPi})  and (\ref{EqProjectionKernel}) with $\sup_x \Pi_p(x,x)\lesssim k$,
based on preliminary estimators $\hat a, \hat b, \hat g$ that attain
rates $(\log n/n)^{-\d/(2\d+d)}$ relative to the uniform norm, satisfies
$$\sqrt n\bigl(\hat\chi_n-\chi(p)-\PP_n \tilde\chi_p^{(1)}\bigr)\prob 0.$$
\end{theorem}

An estimator that is asymptotically linear in the first order efficient influence function, as in the
theorem, is asymptotically optimal in terms of the local asymptotic minimax and
convolution theorems (see e.g.\ \cite{vanderVaart98}, Chapter~25). The present
estimator $\hat\chi_n$ actually looses its efficiency  by splitting the sample
in a part used to construct the preliminary estimators and a part to form $\PP_n$.
This can be easily remedied by crossing over the two parts of the split, and
taking the average of the two estimators so obtained. By the theorem these are
both asymptotically linear in their sample, and hence their average is asymptotically
linear in the full sample and asymptotically efficient. 

The proofs of the theorems are deferred to 
Section~\ref{SectionProofofTheoremParametricRate}.

\section{Minimax rate at lower smoothness ($(\a+\b)/2< d/4$)}
\label{SectionMinimaxRate}
If the average a-priori smoothness $(\a+\b)/2$ 
of the functions $a$ and $b$ falls below $d/4$, then
the functional $\chi$ cannot be estimated any more at
the parametric rate (\cite{RobinsetalMinimax}). The estimator (\ref{EqEstimatorParametricRate})
of Theorem~\ref{TheoremParametricRate} can still be used and, with its bias and
variance as given in the theorem properly balanced, 
attains a certain rate of convergence, faster than the current state-of-the-art linear estimators.
However, in this section we present an estimator that is always better, and attains the
minimax rate of convergence $n^{-(2\a+2\b)/(2\a+2\b+d)}$
provided that the parameter $g$ is sufficiently regular. 

This estimator takes the same general form 
\begin{equation}
\label{EqEstimatorGeneral}
\hat\chi_n=\chi(\hat p)+\UU_n\Bigl(\tilde\chi_{\hat p}^{(1)}
+\thalf\tilde\chi^{(2)}_{\hat p}+\cdots
+\frac1{m!}\chi_{\hat p}^{(m)}\Bigr),
\end{equation}
as the estimator (\ref{EqEstimatorParametricRate}), but
the influence functions $\chi_p^{(j)}$ for $j\ge 3$ will be different.
The idea is to ``cut out'' certain terms from the influence functions in (\ref{EqEstimatorParametricRate})
in order to decrease the variance, but without increasing the bias.
For clarity we first consider the third order estimator, and next extend to
the general $m$th order. To attain the minimax rate the order $m$ must be fixed 
to a large enough value so that the first term in the bias given in Theorem~\ref{TheoremParametricRate}
is no larger than $n^{-(2\a+2\b)/(2\a+2\b+d)}$. (Apart from added complexity there is no loss
in choosing $m$ \emph{larger} than needed.)

The third order kernel $\tilde\chi_p^{(3)}$ in 
(\ref{EqIFsConcrete}) is the symmetrization of the variable
$$6\tilde A_1\Bigl(\Pi_p(Z_1,Z_2)\AA_2\Pi_p(Z_2,Z_3)-
\Pi_p(Z_1,Z_3)\Bigr) \tilde Y_3.$$
Here  $\Pi_p$ is the kernel of an orthogonal projection in $L_2(\aa\bb g)$
onto a $k$-dimensional linear space, which we may view as the sum
of $k$ projections on one-dimensional spaces. The quantity $k^2$ in the
order $O(k^2/n^3)$ of the variance in Theorem~\ref{TheoremParametricRate}
for $m=3$ arises as the number of terms in the product
$\Pi_p(Z_1,Z_2)\AA_2\Pi_p(Z_2,Z_3)$ of the two $k$-dimensional projection kernels.
It turns out that this order can be reduced without
increasing the bias by cutting out ``products of projections on
higher base elements''.

To make this precise, we partition the projection space in blocks, 
and decompose the two projections in the influence function over the blocks: 
\begin{equation}\label{EqDecomposeProjections}
\Pi_p=\sum_{r=0}^R \Pi_p^{(k_{r-1},k_r]},
\qquad
\Pi_p=\sum_{s=0}^S \Pi_p^{(l_{s-1},l_s]}.
\end{equation}
Here $\Pi_p^{(m,n]}$ is the projection onto the subspace spanned by base elements 
with index in intervals $(m,n]$, and $1=k_{-1}<k_0<k_1<\cdots<k_R=k$ 
and $1=l_{-1}<l_0<l_1<\cdots<l_S=k$ are suitable partitions of the set $\{1,\ldots, k\}$.
(``Full'' partitions in singleton sets would make the construction conceptual simpler, but
a small number of blocks will be needed in our proofs.)
The product of the two kernels now becomes a double sum, from
which we retain only terms with small values of $(r,s)$. The improved
 third order
influence function is, with as before $S_3$ denoting symmetrization,
\begin{align}
\nonumber
\chi_p^{(3)}(X_1,X_2,X_3)
&=6S_3\biggl[\dsum_{{(r,s):r+s\le D}\atop{ \vee r=0 \vee s=0}}
\!\!\tilde A_1\Bigl(\Pi_p^{(k_{r-1},k_r]}(Z_1,Z_2)\AA_2\Pi_p^{(l_{s-1},l_s]}(Z_2,Z_3)\\
&\qqqqqquad-\Pi_p^{(k_{r-1}\vee l_{s-1},k_r\wedge l_s]}(Z_1,Z_3)\Bigr)\tilde Y_3\biggr].
\label{EqDoubleThirdOrderIF}
\end{align}
The negative term in the display is the conditional expectation given
$Z_1, Z_3$ of the leading term, and maintains the degeneracy of the kernel.

For the decomposition (\ref{EqDecomposeProjections}) to be valid, the subspaces corresponding
to the blocks must be orthogonal in $L_2(\aa\bb g)$. We may achieve this by starting with a
standard basis $e_1,e_2,\ldots$, with good approximation properties for a target model,
and next replacing this by an orthonormal basis in $L_2(\aa \bb g)$ by the Gram-Schmidt procedure.
For a bounded $g$ the approximation properties will be preserved.

The grids are defined by 
\begin{align}
\label{EqDefinitionGridk}
k_{-1}&=1,& k_r&\sim n 2^{r/\a},&\qquad r&=0,\ldots, R,\\
\label{EqDefinitionGridl}
l_{-1}&=1,& l_s&\sim n 2^{s/\b},&\qquad  s&=0,\ldots, S,
\end{align}
where $R$ and $S$ are chosen such that $k_R\sim l_S\sim k$ (note that $k_0=l_0=n$).
In these definitions the notation $\sim$ means 
``equal up to a fixed multiple'' (needed to allow that $k_r$ and $l_s$ are (dyadic) integers). 
For ease of notation let $l_s=l_{-1}$ for $s\le -1$, and $l_s=l_S$ for $s\ge S$.

The grids $k_0<k_1<\cdots<k_R$ and $l_0<l_1<\cdots<l_S$ partition the
integers $n,n+1,\ldots,k$ in $R$ and $S$ groups. As $k_r^\a
l_s^\b=2^{r+s}n^{\a+\b}$, for every $r,s\ge0$, the cut-off $r+s\le D$
in (\ref{EqDoubleThirdOrderIF}) is delimited by the ``hyperbola'' $i^\a
j^\b\sim 2^Dn^{\a+\b}$ in the space of indices 
$(i,j)\in\{1,\ldots,k\}^2$ of base elements used in the two kernels, 
with only the pairs below the hyperbola retained (see Figure~\ref{Figure}). The intuition 
behind this hyperbolic cut-off is the product form of the bias (\ref{EqMARBias}):
a higher order correction on the estimator of $a$ may combine with a lower
order correction on $b$, and vice versa, to give an overall correction of the
desired order. The overall bias is smaller if the cut-off $D$ is chosen
larger, but then more terms are included in the estimator and the variance will be bigger.

\begin{figure}
\centerline{\resizebox{1.6in}{!}{\includegraphics{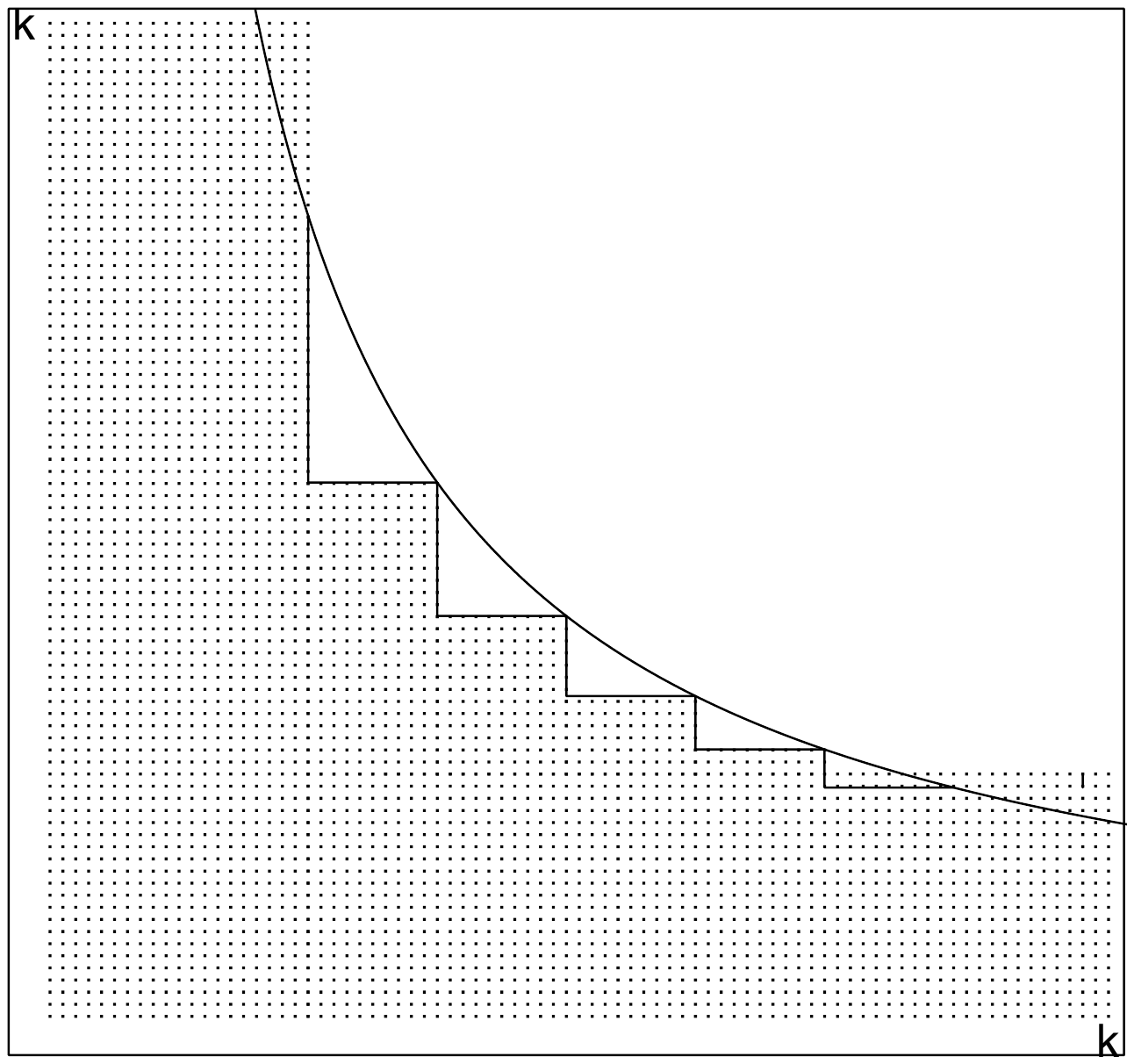}}}
\caption{Both axis carry the indices of the basis functions spanning the projection space $L$, and 
point in the plane refers to a product of two projections. 
Products of projections on pairs of basis functions in the shaded area are included in the third
order influence function. The step function refers to the partitions of the indices as in (\ref{EqDecomposeProjections}).}
\label{Figure}
\end{figure}

Before deriving an optimal value of $D$, we introduce 
the  $m$th order estimator for general $m\ge 3$.
Again we take the estimator of Theorem~\ref{TheoremParametricRate} as starting point,
but modify the higher order influence functions $\tilde\chi^{(j)}_p$, for $j=4,\ldots,m$,
similar and in addition to the modification of the third order influence function.
For given $j$ the former influence function is given in Theorem~\ref{TheoremMARHigherOrderIF} 
(with $m$ of the theorem taken equal to $j$), and is based on a product of $j-1$ projection kernels.
We modify this in two steps.
For each of the $j-2$ contiguous pairs of kernels ($(1st,2nd), (2nd,3rd),\ldots, ((j-2)th,(j-1)th)$) 
we form a new kernel by truncating the pair at the hyperbola as described previously for the third
order kernel,  and truncating all other kernels at $n$. Next the modified $j$th order kernel
is the sum of the resulting $j-2$ kernels. More formally, the modified $j$th order kernel is equal to 
\begin{equation}
\label{EqHigherOrderCutIF}
\chi_p^{(j)}(X_1,\ldots,X_j)=\sum_{i=1}^{j-2}\chi_p^{(j,i)}(X_1,\ldots,X_j),
\end{equation}
where $\chi_p^{(j,i)}(X_1,\ldots,X_j)$ is the symmetrized, degenerate (relative to $L_2(p)$) part of the variable,
for $i=1,\ldots, j-2$, written in the notation of Theorem~\ref{TheoremParametricRate},
\begin{align*}
&j!(-1)^{j-1}\, \tilde Y_1 \Pi_{1,2}^{(0,n]}\AA_2\times\cdots\times\AA_{i-1}\Pi_{i-1,i}^{(0,n]}\AA_i\times\\
&\times\Bigl[\dsum_{{(r,s):r+s\le D}\atop{ \vee r=0 \vee s=0}}\Pi_{i,i+1}^{(k_{r-1},k_r]}\AA_{i+1}\Pi_{i+1,i+2}^{{(l_{s-1},l_s]}}\Bigr]
\AA_{i+2}\Pi_{i+2,i+3}^{(0,n]}\times\cdots\times\AA_{j-1}\Pi_{j-1,j}^{(0,n]}\tilde A_j.
\end{align*}
For $j=3$ there is only one pair of kernels, and the construction reduces
to the modification (\ref{EqDoubleThirdOrderIF}) as discussed previously. 

We assume that the projections $\Pi_p^{(0,l]}$ and
$\Pi_{\hat p}^{(0,l]}$ map $L_s\bigl(\aa\bb g)$
to $L_s\bigl(\aa\bb g)$, for every $s\in \bigl[r/(r-1),r\bigr]$,
with uniformly bounded norms.

\begin{theorem}
\label{TheoremMinimaxRate}
The estimator (\ref{EqEstimatorGeneral}) for $m\ge 3$
with the influence functions $\tilde \chi_p^{(j)}$ and $\chi_p^{(j)}$ given in
(\ref{EqMARIFApprox}) and (\ref{EqIFsConcrete}) for $j=1,2$, respectively, and
in  (\ref{EqHigherOrderCutIF}) for $j\ge 3$, and with $\Pi_p^{(0,l]}$ kernels of orthogonal projections in $L_2\bigl(\aa\bb g)$
satisfying (\ref{EqProjectionKernel}) with 
$\sup_x \Pi_{\hat p}^{(0,l]}(x,x)\lesssim l$, satisfies, for $r\ge 2$ (and $r/(r-2)=\infty$ if $r=2$),
\begin{align*}
\hat\E_p&\hat \chi_n-\chi(p)
=O\Bigl(\|\hat a-a\|_{r}\|\hat b-b\|_{r}
\|\hat g-g\|_{\frac{mr}{r-2}}^{m-1}\Bigr)\\
& +O\Bigl(\Bigl\|(I-\Pi_p^{(0,k]})\frac{\hat a- a}{\aa}\Bigr\|_2
\Bigl\|(I-\Pi_p^{(0,k]})\frac{\hat b- b}{\bb}\Bigr\|_2\Bigr),\\
&+O\Bigl(\sum_{r=1}^R 
\Bigl\|(I-\Pi_{\hat p}^{(0,k_{r-1}]})\Bigl(\frac{\hat a-a}\aa\Bigr)\Bigr\|_{r}
\Bigl\|(I-\Pi_{\hat p}^{(0,l_{D-r}]})\Bigl(\frac{\hat b-b}\bb\Bigr)\Bigr\|_{r}
\|\hat g-g\|_{\frac r{r-2}} \Bigr)\\
&+O\Bigl(R\Bigl\|(I-\Pi_{\hat p}^{(0,n]})\frac{\hat a- a}{\aa}\Bigr\|_r
\Bigl\|(I-\Pi_{\hat p}^{(0,n]})\frac{\hat b- b}{\bb}\Bigr\|_r\|\hat g-g\|_{\frac{mr}{r-2}}^2\Bigr),\\
&\hat\var_p \hat\chi_n\lesssim \frac{1+R^2\e_n^4}n+\frac{k(1+R\e_n^2)}{n^2}
+\frac{D 2^{(\frac1\a\vee\frac1\b)D}}{n}.
\end{align*}
Here $\e_n$ is the maximum of the three rates $\|a-\hat a\|_4$, $\|b-\hat b\|_4$ and 
$\|g-\hat g\|_\infty$ and the constant $c$ depends on the supremum norms
of $\aa,\bb, 1/a, b, p/\hat p,g/\hat g$, the norms of the operators 
$\Pi_{\hat p}^{(0,l]}: L_s(\aa\bb \hat g)\to L_s(\aa\bb \hat g)$, for $l=1,\ldots, k$ only, and $r\ge 2$,
\end{theorem}

A proof of the theorem is presented in
Sections~\ref{SectionProofofTheoremMinimaxRate3} and~\ref{SectionProofofTheoremMinimaxRate}.

The first two terms in the bias are the same as in Theorem~\ref{TheoremParametricRate};
the third and fourth terms are the price paid for cutting out
terms from the influence function. The benefit is a reduced
variance. We shall show that the boundary parameter $D$ can be
chosen such that the third term in the variance (resulting from
the third and higher order parabolic parts of the influence function) is not bigger
than the second term, while the increase in bias is negligible.
The number $R$ will be logarithmic and $\e_n$ a negative power of $n$, the product
$R\e_n^2$ will tend to zero and the first two terms of the variance will be of 
order $1/n$ and $k/n^2$.

Assume that the functions $a$ and $b$ and their estimates are known to
belong to models that are well approximated by the base functions
$e_1,e_2,\ldots$ in the sense that, for $p\in\{p,\hat p\}$, and every value $l$ in
one of the two grids (\ref{EqDefinitionGridk})-(\ref{EqDefinitionGridl}),
\begin{align}
\label{EqApproxByPiRefined}
\Bigl\|(I-\Pi_p^{(0,l]})\Bigl(\frac{\hat a-a}\aa\Bigr)\Bigr\|_r
&\lesssim  \Bigl(\frac 1 l\Bigr)^{\a/d},\\
\Bigl\|(I-\Pi_p^{(0,l]})\Bigl(\frac{\hat b-b}\bb\Bigr)\Bigr\|_r
&\lesssim  \Bigl(\frac 1 l\Bigr)^{\b/d}.
\label{EqApproxByPiRefinedTwo}
\end{align}
Then the second term in the bias is of the
order $(1/k)^{\a/d+\b/d}$, as in Theorem~\ref{TheoremParametricRate},
which is smaller than the minimax rate $n^{-(2\a+2\b)/(2\a+2\b+d)}$ for
\begin{equation}
\label{EqDefinitionk}
k\sim n^{2d/(2\a+2\b+d)}.
\end{equation}
With this choice of $k$, the upper bound on the variance
is of the square minimax rate
$n^{-(4\a+4\b)/(2\a+2\b+d)}$ if $D$ is chosen to satisfy
\begin{equation}
\label{EqDefinitionD}
2^{(\frac1\a\vee\frac1\b)D}\sim \frac{1}{\log n}n^{(d-2\a-2\b)/(d+2\a+2\b)}.
\end{equation}
Furthermore, under (\ref{EqDefinitionk})
the numbers $R,S$  of grid points are of the order $\log n$. 

In the third term of the bias we  apply assumptions 
(\ref{EqApproxByPiRefined})-(\ref{EqApproxByPiRefinedTwo}) and the
identity $k_{r-1}^\a l_{D-r}^\b \sim n^{\a+\b}2^{D}$, which results from
(\ref{EqDefinitionGridk})-(\ref{EqDefinitionGridl}), 
to see that the third term of the bias is of order
$$\sum_{r=1}^R \Bigl(\frac 1{k_{r-1}}\Bigr)^{\a/d}\Bigl(\frac 1{l_{D-r}}\Bigr)^{\b/d}\|\hat g-g\|_{r/(r-2)}
\le R \Bigl(\frac 1{n^{\a+\b}2^D}\Bigr)^{1/d}\|\hat g-g\|_{r/(r-2)}.$$
If the convergence rate of $\hat g$ is $n^{-\g/(2\g+d)}$,
then, for the choice of $D$ given in (\ref{EqDefinitionD}),
this can (by a calculation) seen to be of 
smaller order than the minimax rate $n^{-(2\a+2\b)/(2\a+2\b+d)}$
if $\g$ is large enough that 
\begin{equation}
\label{EqfSmoothEnough}
\frac{\g}{2\g+d}> \Bigl(\frac{\a\vee\b}{d}\Bigr)
\Bigl(\frac{d-2\a-2\b}{d+2\a+2\b}\Bigr).
\end{equation}
The fourth term in the bias can by a similar analysis be seen to be of the order
$$R\Bigl(\frac 1{n}\Bigr)^{\a/d}\Bigl(\frac 1{n}\Bigr)^{\b/d}\|\hat g-g\|_{(m-2)r/(r-2)}^2.$$
Again this is smaller than the minimax rate if $\g$ satisfies assumption (\ref{EqfSmoothEnough}).

Finally, if the convergence rates of $\hat a$ and $\hat b$
are $n^{-\a/(2\a+d)}$ and $n^{-\b/(2\b+d)}$, then the first term
in the upper bound of the bias is of the order
$$\Bigl(\frac 1n\Bigr)^{\a/(2\a+d)+\b/(2\b+d)+(m-1)\g/(2\g+d)}.$$
We choose $m$ large enough so that this is of smaller order than the preceding terms.
In particular, we can choose it so that this is smaller than the minimax rate.

We summarize this in the following corollary, which is the most advanced result of the paper.

\begin{corollary}
If (\ref{EqApproxByPiRefined})--(\ref{EqfSmoothEnough}) hold, 
and $\Pi_p^{(0,l]}$ are kernels of orthogonal projections in $L_2\bigl(\aa\bb g)$ satisfying (\ref{EqProjectionKernel}) with 
$\sup_x \Pi_{\hat p}^{(0,l]}(x,x)\le l$, then the $m$th order estimator with the kernels
(\ref{EqHigherOrderCutIF}) for $j\ge 3$ and sufficiently large $m$ and suitable initial estimators,  
attains the rate $n^{-(2\a+2\b)/(2\a+2\b+d)}$ for estimating $\chi(p)$.
\end{corollary}

\section{Proofs}
\label{SectionProofs}

\subsection{Proof of Theorem~\ref{TheoremSecondOrderEstimator}}
\label{SectionProofofTheoremSecondOrderEstimator}
Write $\hat\Pi$ and $\Pi$ for $\Pi_{\hat p}$ and $\Pi_p$, respectively,
for both the kernels and the corresponding projection operators,
 and drop $p$ also in $\hat \E_p$ and $\hat\var_p$.
From (\ref{EqMARBias}) and (\ref{EqMARAdhocSecondOrderIF}) we have
\begin{align*}
&\hat\E \hat\chi_n-\chi(p)\\
&\quad=-\int (\hat a-a)(\hat b-b)\,g\,d\n
-\hat\E A_1\bigl(Y_1-\hat b(Z_1)\bigr)\bigl(A_2\hat a(Z_2)-1\bigr) 
\hat\Pi(Z_1,Z_2)\\
&\quad=-\int (\hat a-a)(\hat b-b)\,g\,d\n
+\int\!\!\int\bigl[(\hat a-a)\times(\hat b-b)\bigr]\,\bigl(g\times g\bigr)\,
\hat\Pi\,d\n\times\n.
\end{align*}
The double integral on the far right with $\hat\Pi$ replaced by $\Pi$
can be written as the single integral 
$\int (\hat a-a)\Pi(\hat b-b)\,g\,d\n$, for $\Pi(\hat b-b)$ the
image of $\hat b-b$ under the projection $\Pi$.
Added to the first integral on the right this
gives $-\int (\hat a-a)(I-\Pi)(\hat b-b)\,g\,d\n$, which is bounded
in absolute value by the second term in the upper bound for the bias.

Replacement of $\hat\Pi$ by $\Pi$ in the double integral 
gives a difference
\begin{align*}
&\int\!\!\int\bigl[(\hat a-a)\times(\hat b-b)\bigr]\,g\times g\,
(\hat\Pi-\Pi)\,d\n\times\n\\
&\qqqquad=\int(\hat a-a)\left(\hat\Pi\Bigl((\hat b-b)\frac g{\hat g}
\Bigr)-\Pi(\hat b-b)\right)\,g\,d\n\\
&\qqqquad\le\|\hat a-a\|_s
\left\|\hat\Pi\Bigl((\hat b-b)\frac g{\hat g}
\Bigr)-\Pi(\hat b-b)\right\|_{r,\hat g} \|g/\hat g\|_\infty^{1/r},
\end{align*}
by H\"older's inequality, for a conjugate pair $(r,s)$. 
Considering $\hat \Pi$ as the projection in $L_2(\hat g)$ with weight 1,
and $\Pi$ as the weighted projection 
in $L_2(\hat g)$ with weight function $\hat w=g/\hat g$,
we can apply Lemma~\ref{LemmaDifferenceOfProjections}(i) (with $q=s/r$ and $rp=s/(s-2)$)
to see that this is bounded in absolute value by
$$\|\hat a-a\|_s
\|\hat\Pi\|_{s,\hat g}\|\Pi\|_{s,\hat g} \|\hat b-b\|_{s,\hat g}\|\hat w-1\|_{s/(s-2),\hat g}\|w\|_\infty^{1/r}.$$
Because $\hat w$ is assumed bounded away from 0 and infinity,
this is of the same order as the first term in the upper bound on the
bias (if $r$ replaces $s$).

Because the function $\chi_{\hat p}^{(1)}$ is uniformly bounded,
the (conditional) variance of $\UU_n\chi_{\hat p}^{(1)}$ is of the 
order $O(1/n)$. Thus for the variance bound it suffices
to consider the (conditional) variance of $\UU_n\chi_{\hat p}^{(2)}$.
In view of Lemma~\ref{LemVarUBound}
this is bounded above by a multiple of
$$\frac 1{n^2}P^2(\chi_{\hat p}^{(2)})^2+\max_{i\in\{1,2\}}\frac 1n P\Bigl(\E \bigl(\chi_{\hat p}^{(2)}\given X_i\bigr)\Bigr)^2.$$
The variables $A\bigl(Y-\hat b(Z)\bigr)$ and $\bigl(A\hat a(Z)-1\bigr)$
are uniformly bounded. Hence the first term 
is bounded above by a multiple of 
$n^{-2} \int \hat \Pi^2\,(\hat g\times \hat g)\,d\n\times\n$,
which is equal to $k/n^2$, by Lemma~\ref{LemmaNormProjection}.
The conditional expectations in the second  term can be written
$A_1\bigl(Y_1-\hat b(Z_1)\bigr) \Pi_{\hat p}\bigl(\hat a/a-1)g/\hat g\bigr)(Z_1)$ and
$\Pi_{\hat p}\bigl(\hat b-b)g/\hat g\bigr)(Z_2)(A_2\hat a(Z_2)-1) $, for $i=1$ and $i=2$,
respectively, where $\Pi_p$ is the operator defined by the kernel. Because the second
moments of these variables under $\hat p$ are uniformly bounded, the second term
contributes  a factor of order $1/n$ only.

\subsection{Proof of Theorems~\ref{TheoremParametricRate} 
and~\ref{TheoremParametricRateEfficient}}
\label{SectionProofofTheoremParametricRate}
Let $\hat A$ and $\hat Y$ be $\tilde A$ and $\tilde Y$ 
as in (\ref{EqDeftildeYtildeA}) with
$a$ and $b$ in their definitions replaced by
$\hat a$ and $\hat b$. Because $\hat a$ and $\hat b$ are
projected onto themselves under the map $(a,b,g)\mapsto (\tilde a,\tilde b,g)$
(see Theorem~\ref{TheoremCharacTildes}),
we actually obtain the same variables by replacing $\tilde a$ and $\tilde b$
by $\hat a$ and $\hat b$, respectively:
$\hat A=\bigl(A\hat a(Z)-1\bigr)\bb(Z)$ and 
$\hat Y=A\bigl(Y-\hat b(Z)\bigr)\aa(Z)$. Furthermore,
let $\Pi$ and $\hat \Pi$ denote the operators $\Pi_p$ and $\Pi_{\hat p}$,
respectively, and $\Pi_{i,j}$ and $\hat \Pi_{i,j}$ their kernels
evaluated at $(Z_i,Z_j)$.

By explicit calculations,
\begin{equation}
\label{EqproofFirstOrder}
\chi(\hat p)+\hat\E_p \tilde\chi_{\hat p}^{(1)}-\chi(p)
=-\int (\hat a-a)(\hat b-b)\, g\,d\n
=\hat\E \hat A_1\Pi_{1,2}\hat Y_2 -\hat R,
\end{equation}
for $\hat R$ defined by
$$\hat R=\int \Bigl(\frac{\hat a-a}\aa\Bigr)(I-\Pi)
\Bigl(\frac{\hat b-b}\bb\Bigr) \aa\bb g\,d\n.$$
The variable $\hat R$ is bounded by the second term in 
the expression for $\hat\E_p\hat \chi_n-\chi(p)$ in the statement of the theorem.
We next show by induction on $m$ that
\begin{align}
&\label{EqInductionBias}
\qquad\qqqquad\qqqquad\hat R+\chi(\hat p)+\hat\E \tilde\chi_{\hat p}^{(1)}
+\cdots+\frac 1{m!}\hat\E \tilde\chi_{\hat p}^{(m)}-\chi(p)\\
&=(-1)^{m-1}\hat\E \hat A_1(\hat\Pi-\Pi)_{1,2}
\AA_2(\hat\Pi-\Pi)_{2,3}\times\cdots\times
\AA_{m-1}(\hat\Pi-\Pi)_{m-1,m}\hat Y_m.\nonumber
\end{align}
The analysis of the bias can then be concluded by showing
that the right side of (\ref{EqInductionBias})
is of the order as the first term
given in the theorem.

Equation (\ref{EqproofFirstOrder}) 
and the definition of $\tilde \chi_p^{(2)}$ readily show
that identity (\ref{EqInductionBias})
is true for $m=2$. We proceed to general $m$ by induction.
Relative to its value for $m$ the left
side receives for $(m+1)$ the extra term
$\hat\E \chi_{\hat p}^{(m+1)}/(m+1)!$, which is equal to
$(-1)^m$ times $\hat\E \hat A_1\hat\Pi_{1,2}\AA_2\hat\Pi_{2,3}\times\cdots\times
\AA_m\hat\Pi_{m,m+1}\hat Y_{m+1}$ minus
a sum of terms resulting from projections of
this leading term. This extra term without the factor $(-1)^m$
(but including the projections) can be written (cf.\
(\ref{EqIFsConcrete}) and (\ref{EqMakeDegenerate}))
\begin{equation}
\label{EqHulpBias}
\sum_{i=0}^{m-1} {m-1\choose i}
\hat\E \hat A_1\hat\Pi_{1,2}\AA_2\hat\Pi_{2,3}\times\cdots\times
\AA_{m-i}\hat\Pi_{m-i,m-i+1}\hat Y_{m-i+1}(-1)^{i}.
\end{equation}
To prove the induction hypothesis for $m+1$ it suffices
to show that this is equal to 
\begin{align}
\nonumber
&\hat\E \hat A_1(\hat\Pi-\Pi)_{1,2}
\AA_2(\hat\Pi-\Pi)_{2,3}\times\cdots\times
\AA_{m-1}(\hat\Pi-\Pi)_{m-1,m}\hat Y_m\\
&\qquad+\hat\E \hat A_1(\hat\Pi-\Pi)_{1,2}
\AA_2(\hat\Pi-\Pi)_{2,3}\times\cdots\times
\AA_{m}(\hat\Pi-\Pi)_{m,m+1}\hat Y_{m+1}.
\label{EqSumTwoTerms}
\end{align}
To achieve this we expand the two terms of the preceding
display into sums of expressions of the form,
with each $K_{j,j+1}^{(j)}$ equal to $\hat\Pi_{j,j+1}$ or $\Pi_{j,j+1}$
and $l$ the number of $j$ for which the first alternative is true,
\begin{equation}
\label{EqExpansionProduct}
B_l:=(-1)^{m-1-l}\,\hat\E \hat A_1 K_{1,2}^{(1)}
\AA_2 K_{2,3}^{(2)}\times\cdots\times\AA_{m-1}K_{m-1,m}^{(m-1)}\hat Y_m,
\end{equation}
and of the same form with $m+1$ replacing $m$ for the second term
of (\ref{EqSumTwoTerms}).  As the notation suggests the
expression in (\ref{EqExpansionProduct}) depends on $l$ (and $m$,
but this is fixed), but not on which $K$ are equal to $\hat\Pi$ or $\Pi$.
To see this we use that $\Pi$ is a projection
onto $L$ in $L_2(\aa\bb g)$, so that 
$\int \Pi_{1,2}\, \g(z_2)\,(\aa\bb g)(z_2)\,d\n(z_2)=\g(z_1)$ for every $\g\in L$;
and $\hat \Pi$ is also a projection onto $L$, so that
as a function of one argument $\hat\Pi_{1,2}$ is contained in $L$.
This observation yields the identities, for
$K$ equal to $\hat \Pi$ or $\Pi$,
$$\hat\E_{Z_j} \Pi_{j-1,j}\AA_jK_{j,j+1}=K_{j-1,j+1}
=\hat\E_{Z_j} K_{j-1,j}\AA_j\Pi_{j,j+1}.$$
This allows to reduce (\ref{EqExpansionProduct}) to
\begin{align*}
B_l&=(-1)^{m-1-l}\,\hat\E \hat A_1 \hat\Pi_{1,2}
\AA_2\hat\Pi_{2,3}\times\cdots\times\AA_{l}\hat\Pi_{l,l+1}\hat Y_{l+1},\qquad l\ge 1,\\
B_0&=(-1)^{m-1}\hat\E \hat A_1 \Pi_{1,2}\hat Y_2.
\end{align*}
Thus after expanding the two terms of (\ref{EqSumTwoTerms}) 
in the quantities $B_l$, and simplifying these quantities,
we can write their sum (\ref{EqSumTwoTerms}) 
$$(B_0-B_0)+\sum_{l=1}^{m-1}\left({m\choose l}-{m-1\choose l}\right)
B_l(-1)^{m-l}+B_m.$$
The difference of the binomial coefficients is ${m-1\choose l-1}$.
The expression is equal to (\ref{EqHulpBias}), as claimed.
This completes the proof of (\ref{EqInductionBias}).

Next we bound the right side of (\ref{EqInductionBias}),
by taking the expectation in turn with respect to
$X_m, X_{m-1},\ldots, X_1$. For  $M_{\hat w}$ multiplication 
by the  function $\hat w=g/\hat g$, 
$$\hat\E_{X_m}(\hat \Pi-\Pi)_{m-1,m}\hat Y_m
=(\hat\Pi M_{\hat w}-\Pi)\Bigl(\frac{\hat b-b}\bb\Bigr)(Z_{m-1}).$$
Next, for any function $h$ and $i=m-1,m-2,\ldots,2$, 
$$\hat\E_{X_i}(\hat\Pi-\Pi)_{i-1,i} \AA_i h(Z_i)=(\hat\Pi M_{\hat w} -\Pi) h(Z_{i-1}).$$
Combining these equations, we can write
the right side of (\ref{EqInductionBias}) in the form
$$(-1)^{m-1}\int \Bigl(\frac{a-\hat a}\aa\Bigr)
\biggl[\bigl(\hat\Pi M_{\hat w}-\Pi\bigr)^{m-1}
\Bigl(\frac{\hat b-b}{\bb}\Bigr)\biggr]\,\aa\bb g\,d\n.$$
We bound this by first applying H\"older's inequality, with conjugate pair $(\t,t)$ with $\t$ equal to
$r$ as in the statement of the theorem, and next Lemma~\ref{LemmaDifferenceOfProjections}(iii),
with $\hat\Pi$ and $\Pi$ viewed as weighted orthogonal projections 
in $L_2(\aa\bb\hat g)$ with weights 1 and $\hat w$, respectively,
and $r=\t(m-1)/(m+\t-3)$, $p=(m+\t-3)/(\t-2)$ and $q=(m+\t-3)/(m-1)$, so that  $rp=(m-1)\t/(\t-2)$
and $rq=\t$ (and $m$ of the lemma taken equal to the present $m$ minus 1).

By Lemma~\ref{LemVarUBound} the (conditional) variance of  $(j!)^{-1}\UU_n\chi_{\hat p}^{(j)}$ is bounded above by 
$$\sum_{l=1}^j\frac{2^jj^{2l}}{n^l}
\E_p\biggl(\E_p \Bigl(S_j D_{\hat p} \bigl[\hat A_1\hat\Pi_{1,2}\AA_2\hat\Pi_{2,3}\cdots\hat\Pi_{j-1,j}\hat Y_j\bigr]
\given X_1,\ldots,X_l\Bigr)\biggr)^2.$$
Here $\hat \Pi_{i,j}$ are the estimated kernels (the original ones with $p$ replaced by $\hat p$) and $D_{\hat p}$ is
the operation of making degenerate under $\hat p$ (not under $p$!).

By Lemma~\ref{LemSumAverage}(ii), for any function $h(X_1,\ldots, X_j)$
 the second moment of $\E \bigl(S_jh(X_1,\ldots, X_j)\given X_1,\ldots, X_l\bigr)
=(j!)^{-1}\sum _\s \E \bigl(h(X_{\s(1)},\ldots, X_{\s(j)})\given X_1,\ldots, X_l\bigr)$ is bounded above by
the maximum over all permutations $\s$  of the second moments of  
the variables $\E \bigl(h(X_{\s(1)},\ldots, X_{\s(j)})\given X_1,\ldots, X_l\bigr)$.
Because $X_1,\ldots, X_j$ are i.i.d., we can move the permutation $\s$ from the argument of $h$ to the conditioning variables,
and conclude that the second moment in the $l$th term on the right side is bounded above by 
$$\max_{{B\subset\{1,\ldots, j\}:}\atop {|B|=l}}
\E_p\biggl(\E_p \Bigl(D_{\hat p}\bigl[\hat A_1\hat\Pi_{1,2}\AA_2\hat\Pi_{2,3}\cdots\AA_{j-1}\hat\Pi_{j-1,j}\hat Y_j\bigr]
\given X_B\Bigr)\biggr)^2.$$
These are bounded in Lemma~\ref{LemSecondMoment}.

We complete the proof of Theorem~\ref{TheoremParametricRate} by bounding the square of 
$\hat\chi_n-\E_p\hat \chi_n$
by $\sum_{j=1}^m2^j\bigl((\UU_n-P^n)\tilde\chi_{\hat p}^{(j)}/j!\bigr)^2 \sum_j 2^{-j}$.
The extra factor $2^j$ can be incorporated in the constant $c$ in the
theorem. We finish by changing the order of summation in the double sum $\sum_{j=2}^m\sum_{l=1}^j$, so as
to collect the terms by the order of $k^{l-1}/n^l$.

For the proof of Theorem~\ref{TheoremParametricRateEfficient} it clearly suffices to
show that
\begin{align*}
\hat\E_p \sqrt n\bigl(\hat\chi_n-\chi(p)-\PP_n \tilde\chi_p^{(1)}\bigr)&\prob 0,\\
\hat\var_p \sqrt n\bigl(\hat\chi_n-\chi(p)-\PP_n \tilde\chi_p^{(1)}\bigr)&\prob 0.
\end{align*}
Because an influence function is centered at mean zero, the first
is simply $\sqrt n$ times the bias of $\hat\chi_n$.
By Theorem~\ref{TheoremParametricRate} the bias is of the order
$$\Bigl(\frac{\log n}n\Bigr)^{\a/(2\a+d)+\b/(2\b+d)+\g (m-1)/(2\g+d)}
+\Bigl(\frac 1k\Bigr)^{(\a+\b)/d}.$$
The first term is trivially $o(n^{-1/2})$, as $m_n\ra\infty$.
In the second we write $(\a+\b)/d=r/2$, where $r>1$ by assumption, and see that
it is $o(n^{-1/2})$, since $k n^{-1/r}\ra\infty$.

To handle the variance we split the estimator $\hat\chi_n$ in its linear and
higher order terms. By Lemma~\ref{LemSumAverage}(i) and the argument given previously,
for $c$ a sufficiently large constant, the variance of the higher order terms satisfies
\begin{align*}
&\var \sum_{j=2}^m\frac 1{j!}\UU_n\tilde \chi_{\hat p}^{(j)}
\le \sum_{j=2}^m\sum_{l=1}^j\frac{c^j j^{2l}}{n^l} k^{l-1} \e_n^{2(j-l)}\\
&\qquad=\sum_{j=2}^m\frac{j^2\e_n^{2j}c^j}{n\e_n^2}\sum_{l=1}^j\Bigl(\frac{j^2k}{n\e_n^2}\Bigr)^{l-1}
=\sum_{j=2}^m\frac{j^2\e_n^{2j}c^j}{n\e_n^2}\frac{(\bigl(j^2k/(n\e_n^2)\bigr)^j-1}{j^2k/(n\e_n^2)-1}.
\end{align*}
%\le \sum_{j=2}^m\frac{c^jj^2}n\Bigl(\sum_{l=1}^{j-1}\Bigl(\frac{j^2k}{n\e_n^2}\Bigr)^{l-1}\e_n^{2j-2}+\Bigl(\frac{j^2k}n\Bigr)^{j-1}\Bigr).
By assumption  $\e_n=O(n^{-\h})$ for some $\h>0$ and $k\sim n/(\log n)^3$. Thus  $j^2k/(n\e_n^2)\ra \infty$ uniformly in $j\ge 1$,
and the preceding display is bounded above by
$$2\sum_{j=2}^m \frac{j^2\e_n^{2j}c^j}{n\e_n^2}\Bigl(\frac{j^2k}{n\e_n^2}\Bigr)^{j-1}
\lesssim \frac1n\sum_{j=2}^m \Bigl(\frac{2cj^2k}{n}\Bigr)^{j-1}\lesssim \frac 1n\sum_{j=2}^\infty \Bigl(\frac1{\log n}\Bigr)^{j-1}
\lesssim \frac 1{n\log n},$$
since $j^2\lesssim 2^{j-1}$ and $2cj^2k/n\le 2cm^2k/n\le 1/\log n$, for every $j\le m$.
Finally the linear term in $\hat\chi_n$ gives the contribution 
$$\hat\var_p \sqrt n\bigl(\PP_n \tilde\chi_{\hat p}^{(1)}-\chi(p)-\PP_n \tilde\chi_p^{(1)}\bigr)
=\hat\var (\tilde\chi^{(1)}_{\hat p}-\tilde\chi^{(1)}_{p}).$$ 
From the explicit expression (\ref{EqMARFOIF}) 
for the first order influence function (or (\ref{EqMARIFApprox}) 
in the case of $\hat p$, which gives an identical function), this is seen to tend to zero
by the dominated convergence theorem.

\subsection{Proof  of Theorem~\ref{TheoremMinimaxRate} for $m=3$}
\label{SectionProofofTheoremMinimaxRate3}
The theorem asserts that the bias of the estimator $\hat\chi_n$ is equal to
the sum of four terms, the first two of which also arise 
in the bias of the estimator considered in 
Theorem~\ref{TheoremParametricRate}. Therefore, we can prove the
assertion on the bias by showing that the expected values
of the current estimator $\hat\chi_n$ (for $m=3$) and the
estimator in Theorem~\ref{TheoremParametricRate} differ by less
than the additional bias terms in Theorem~\ref{TheoremMinimaxRate}.

The two estimators differ only in their
third order influence functions, where the present estimator
retains only the terms  in the double sum (\ref{EqDoubleThirdOrderIF})
with $r=0$, $s=0$, or $r+s\le D$. Thus the difference
of the expectations of the two estimators is equal to 
\begin{align*}
\dsum_{{r+s>D}\atop{r,s\ge1}} \hat\E_p\hat A_1\Bigl[\hat\Pi^{(k_{r-1},k_r]}(Z_1,Z_2)\AA_2&\hat\Pi^{(l_{s-1},l_s]}(Z_2,Z_3)\\
&-\hat\Pi^{(k_{r-1}\vee l_{s-1},k_r\wedge l_s]}(Z_1,Z_3)\Bigr]\hat Y_3.
\end{align*}
The expectation $\hat\E_p$ refers to the variable $(X_1,X_2,X_3)$ for fixed values of the
preliminary samples, which are indicated in the ``hat'' symbols on $\hat A_1, \hat Y_3$ and the
kernels, and hence is an integral relative to the density $(x_1,x_2,x_3)\mapsto p(x_1)p(x_2)p(x_3)$.
If we replace $p(x_2)$ in this density by $\hat p(x_2)$, then the integral will be zero, as the
kernel is degenerate under $\hat P$. Thus we may integrate against
$(x_1,x_2,x_3)\mapsto p(x_1)(p-\hat p)(x_2)p(x_3)$. In that case the projection term 
$\hat A_1\hat\Pi^{(k_{r-1}\vee l_{s-1},k_r\wedge l_s]}(Z_1,Z_3)\hat Y_3$ integrates to zero, as it does not depend
on $X_2$ and $\int (p-\hat p)(x_2)\,d\m(x_2)=0$, and hence can be dropped. Next we condition
$\hat A_1$ and $\hat Y_3$ on $Z_1,Z_2,Z_3$ and write the preceding display in the form
\begin{align*}
\dsum_{{r+s>D}\atop{r,s\ge1}} \int\!\!\int\!\!\int \frac{\hat a-a}\aa(z_1)
\hat\Pi^{(k_{r-1},k_r]}(z_1,z_2)&\hat\Pi^{(l_{s-1},l_s]}(z_2,z_3) \frac {\hat b-b}\bb(z_3)\\
&\times\,d\r(z_1)\, d(\r-\hat\r)(z_2)\, d\r(z_3).
\end{align*}
for $\r$ and $\hat\r$ the measures defined by $d\r=\aa\bb g\,d\n$ and 
$d\hat\r=\aa\bb \hat g\,d\n$. 
The double sum can be rewritten as the sum over $r$ running
from $1$ to $R$ and over $s$ from $D-r+1$ to $S$, which gives the equivalent
representation, with the $\times$ referring to ``tensor products'' as explained in Section~\ref{SectionNotation},
\begin{align*}
&\sum_{r=1}^R \int 
\Bigl(\frac{\hat a-a}\aa\times 1\times\frac {\hat b-b}\bb\Bigr)\,
\, \Bigl(\hat\Pi^{(k_{r-1},k_r]}\times\hat\Pi^{(l_{D-r},k]}\Bigr) \,d\bigl(\r\times(\r-\hat\r)\times\r\bigr).
\end{align*}
We write $\hat \Pi^{(k_{r-1},k_r]}=\hat\Pi^{(k_{r-1},k]}-\hat \Pi^{(k_{r},k]}$, 
and next arrive at the difference of two expressions of the
type, with $k_r'=k_{r-1}$ and $k_r'=k_r$, respectively,
\begin{align*}
&\sum_{r=1}^R \int \Bigl(\frac{\hat a-a}\aa\times 1\times\frac {\hat b-b}\bb\Bigr)\,
\Bigl(\hat\Pi^{(k_r',k]}\times\hat\Pi^{(l_{D-r},k]}\Bigr)\,d\bigl(\r\times(\r-\hat\r)\times\r\bigr).
\end{align*}
If the measure of integration were $\hat\r\times(\r-\hat\r)\times\hat\r$ (with $\hat\r$ instead of
$\r$), then we could perform the integrals on $z_1$ and $z_3$
and next apply H\"older's inequality to bound the resulting expression in absolute value by
$$\sum_{r=1}^R\Bigl\|\hat\Pi^{(k_r',k]}\Bigl(\frac{\hat a-a}\aa\Bigr)\Bigr\|_{r}
\Bigl\|\hat\Pi^{(l_{D-r},k]}\Bigl(\frac{\hat b-b}\bb\Bigr)\Bigr\|_{r} 
\Big\|\frac{g}{\hat g}-1\Bigr\|_{r/(r-2)},$$
where the norms are those of $L_2(\aa\bb \hat g)$, which are
equivalent to those of $L_2(\n)$, by assumption. We can write
$\hat\Pi^{(l,k]}=\hat\Pi^{(0,k]}(I-\hat\Pi^{(0,l]})$ and use the assumed boundedness of $\hat\Pi^{(0,l]}$
as an operator on $L_r(\aa\bb g)$ to bound this by the third term in the bias. 

Replacing $\r\times(\r-\hat\r)\times\r$ by $\hat\r\times(\r-\hat\r)\times\hat\r$ can be achieved
by writing the first and last occurrence of $\r$
as $\r=\hat\r+(\r-\hat\r)$ and expanding the resulting expression on the $+$ signs
into four terms. One of these has the measure $\hat\r\times(\r-\hat\r)\times\hat\r$. 
The other three terms have two or three occurrences of $\r-\hat \r$, and can be bounded 
by the first term in the bias (with $m=3$). This is argued precisely 
under (\ref{EqRemainingEstimationBias}) below.

Because the first and second order influence functions
are equal to those of the estimator considered in 
Theorem~\ref{TheoremParametricRate}, the (conditional) variances
of $\UU_n\tilde \chi_{\hat p}^{(j)}$ for $j=1,2$
can be seen to be of the orders $O(1/n)$ and $O(k/n^2)$, respectively,
by the same proof. 
By Lemma~\ref{LemVarUBound} the variance for $j=3$ satisfies (see (\ref{EqVarDegU}))
$$\var \UU_n\chi_{\hat p}^{(3)}
\lesssim \sum_{l=1}^3\frac 1{n^l}\E_p\biggl(\E_p\Bigl(\tilde\chi_{\hat p}^{(3)}(X_1,X_2,X_3)
\given X_1,\ldots,X_l\Bigr)\biggr)^2.$$
Here the influence function is given in (\ref{EqDoubleThirdOrderIF}) and can also be written
\begin{align*}
\frac 16\chi_{\hat p}^{(3)}(X_1,X_2,X_3)
&=\dsum_{{(r,s):r+s\le D}\atop{ \vee r=0 \vee s=0}}
\!\!S_3D_{\hat p}\Bigl[\hat A_1\hat \Pi_{1,2}^{(k_{r-1},k_r]}\AA_2\hat \Pi_{2,3}^{(l_{s-1},l_s]}\hat Y_3\Bigr]\\
&=\sum_{r=0}^R
\!\!S_3D_p\Bigl[\hat A_1\hat\Pi_{1,2}^{(k_{r-1},k_r]}\AA_2\hat\Pi_{2,3}^{(0,l_{D-r}']}\hat Y_3\Bigr],
\end{align*}
where $l_{D-r}'=l_{D-r}\vee n$. The degeneracy $D_{\hat p}$ operator works on $X_2$ only.

In the term for $l=3$ we change measure from $p$ to $\hat p$, bound out $\hat A_1$ and $\hat Y_3$,
and  pull out the degeneracy operator to obtain the upper bound a multiple of
\begin{align*}
&\frac 1{n^3}\Big\|\frac p{\hat p}\Bigr\|_\infty^2
\hat P^3 \Bigl(\sum_{r=0}^R\hat\Pi_{1,2}^{(k_{r-1},k_r]}\AA_2 \hat\Pi_{2,3}^{(0, l_{D-r}']}\Bigr)^2.
\end{align*}
After bounding out $\hat A_1^2$ and $\hat Y_3^2$, we write the
squared sum as a double sum. From the fact that the projections $\hat\Pi^{(k_{r-1},k_r]}$ are
orthogonal for different $r$, it follows that the off-diagonal terms of the double sum vanish (the
expectation with respect to $X_1$ is zero). Thus 
the preceding display is bounded above by a multiple of 
$$\frac 1{n^3} \sum_{r=0}^R \hat P^3\bigl(\hat\Pi^{(k_{r-1},k_r]}(Z_1,Z_2)\AA_2\hat\Pi^{(0,l_{D-r}']}(Z_2,Z_3)\bigr)^2.$$
By Lemmas~\ref{LemmaVarianceProductOfKernels} and~\ref{LemmaNormProjection} and the assumption
that $\sup_z\hat\Pi^{(0,l]}(z,z)\lesssim l$ this is bounded by a multiple of 
\begin{align*}
&\frac 1{n^3}\sum_{r=0}^R(k_r-k_{r-1})l_{D-r}'\le \frac 1{n^3} \Bigl(nk+\sum_{r=1}^R(k_r-k_{r-1})(l_{D-r}+n)\Bigr).
\end{align*}
By (\ref{EqDefinitionGridk}) $k_r-k_{r-1}=(1-2^{-\a})k_r\lesssim k_r=n2^{r/\a}$ for $r\ge 1$. 
On substituting this in the display, and noting that $l_{D-r}=0$ if $r>D$, we see that this is
bounded by a multiple of $k/n^2+ 2^{D/\a\vee D/\b}/n$ if $\a\not=\b$ and bounded by a multiple of
$k/n^2+ D2^{D/\a}/n$ if $\a=\b$.

The second moment in the right side  for $l=1$ or $l=2$ is bounded above by a multiple of
$$\max_{{B\subset \{1,2,3\}}\atop{|B|=l}} 
\E_{\hat p}\biggl(\E_p\Bigl(\dsum_{{(r,s):r+s\le D}\atop{ \vee r=0 \vee s=0}}
\!\!D_{\hat p}\Bigl[\hat A_1\hat \Pi_{1,2}^{(k_{r-1},k_r]}\AA_2\hat \Pi_{2,3}^{(l_{s-1},l_s]}\hat Y_3\Bigr]
\given X_B\Bigr)\biggr)^2.$$
We consider the various subsets $B$ separately: $\{1,2\}, \{1,3\},\{2,3\}, \{1\},\{2\},\{3\}$.
Abbreviate $\hat w=g/\hat g$ and $(\aa\bb)(z_i)\,d\hat G(z_i)=\hat g_i\,di$.

$B=\{1,2\}$. Taking first the conditional expectation given $Z_3$ reduces $\hat Y_3$ to $(b-\hat b)/\bb$,
and hence
\begin{align*}
&\E_p\biggl(\sum_{r=0}^RD_{\hat p}\Bigl[\hat A_1\hat \Pi_{1,2}^{(k_{r-1},k_r]}\AA_2\hat \Pi_{2,3}^{(0,l_{D-r}']}\hat Y_3\Bigr]
\given X_1,X_2\biggr)\\
&\qquad=\hat A_1D_{\hat p}^2\Bigl[\int\sum_{r=0}^R\hat \Pi_{1,2}^{(k_{r-1},k_r]}\AA_2\hat \Pi_{2,3}^{(0,l_{D-r}']}
\frac{b_3-\hat b_3}{\bb_3}\hat w_3 \hat g_3\,d3\Bigr]\\
&\qquad=\hat A_1D_{\hat p}^2\Bigl[\sum_{r=0}^R\hat \Pi_{1,2}^{(k_{r-1},k_r]}\AA_2\Bigl(\hat \Pi^{(0,l_{D-r}']}
\Bigl(\frac{b-\hat b}{\bb}\hat w\Bigr)\Bigr)_2 \Bigr].
\end{align*}
When taking the second moment under $\hat p$, we can bound out $\hat A_1$, and leave off the
degeneracy operator, as this is a projection. Furthermore, the terms of the sum are orthogonal
as functions of $Z_1$: we have
$\int \hat \Pi_{1,2}^{(k_{r-1},k_r]}\hat \Pi_{1,2}^{(k_{r'-1},k_{r'}]}\,\hat g_1\,d1=0$, for $r\not= r'$.
Therefore, the second moment is bounded above by a multiple of 
$$\sum_{r=0}^R\E_{\hat p}\biggl(\hat \Pi_{1,2}^{(k_{r-1},k_r]}\AA_2\hat \Pi^{(0,l_{D-r}']}
\Bigl(\frac{b-\hat b}{\bb}\hat w\Bigr)_2 \biggr)^2
\lesssim \sum_rk_r\,\E_{\hat p}\biggl(\hat \Pi^{(0,l_{D-r}']}
\Bigl(\frac{b-\hat b)}{\bb}\hat w\Bigr)\biggr)^2,$$
where we peeled off the square of the kernel $\hat \Pi_{1,2}^{(k_{r-1},k_r]}$ by integrating this over $Z_1$,
reducing this to $\hat \Pi_{2,2}^{(k_{r-1},k_r]}\le \hat \Pi_{2,2}^{(0,k_r]}\le C k_r$.
We finish by leaving off the projection $\hat \Pi^{(0,l_{D-r}']}$
and bounding $\hat w$ by its uniform norm, giving the bound $\sum_rk_r\e_n^2\lesssim k\e_n^2$,
by the definition of $k_r$, which implies $k_r\asymp k_r-k_{r-1}$.

$B=\{1,3\}$. Integrating out $X_2$ gives, as a special case of Lemma~\ref{LemmaIntegrateOutKernels},
\begin{align}
&\E_p\biggl(\sum_{r=0}^RD_{\hat p}\Bigl[\hat A_1\hat \Pi_{1,2}^{(k_{r-1},k_r]}\AA_2\hat \Pi_{2,3}^{(0,l_{D-r}']}\hat Y_3\Bigr]
\given X_1,X_3\biggr)\nonumber\\
&\qquad=\hat A_1\sum_{r=0}^R\int\hat \Pi_{1,2}^{(k_{r-1},k_r]}\hat \Pi_{2,3}^{(0,l_{D-r}']}
(\hat w_2-1) \hat g_2\,d2\,\hat Y_3.\label{EqCondGiven13}
\end{align}
The terms of the sum are again orthogonal relative to integration on $Z_1$ and hence the second moment of
the right side is  bounded above by a multiple of
\begin{align*}
&\sum_{r=0}^R\E_{\hat p}\biggl(\int\hat \Pi_{1,2}^{(k_{r-1},k_r]}\hat \Pi_{2,3}^{(0,l_{D-r}']}(\hat w_2-1) \hat g_2\,d2 \biggr)^2\\
&\qquad 
\lesssim\sum_{r=0}^R\E_{\hat p}\biggl(\hat \Pi^{(0,l_{D-r}']}\Bigl(\hat \Pi_{1,\cdot}^{(k_{r-1},k_r]}(\hat w_\cdot-1)\Bigr)\biggr)_3^2\\
&\qquad 
\le \sum_{r=0}^R\E_{\hat p}\biggl(\hat \Pi_{1,3}^{(k_{r-1},k_r]}(\hat w_3-1)\biggr)^2
\le \sum_{r=0}^R(k_r-k_{r-1})\|\hat w-1\|_\infty^2=k\|\hat w-1\|_\infty^2,
\end{align*}
by first bounding out $\hat w-1$ and next applying the formula for the $L_2$-norm of a projection
kernel (see Lemma~\ref{LemmaNormProjection}).

$B=\{2,3\}$. This is analogous to $B=\{1,2\}$.

$B=\{1\}$. Taking the conditional expectation of \eqref{EqCondGiven13} given $X_1$ gives
\begin{align*}
&\E_p\biggl(\sum_{r=0}^RD_{\hat p}\Bigl[\hat A_1\hat \Pi_{1,2}^{(k_{r-1},k_r]}\AA_2\hat \Pi_{2,3}^{(0,l_{D-r}']}\hat Y_3\Bigr]
\given X_1\biggr)\\
&\qquad=\hat A_1\sum_{r=0}^R\int\!\!\int\hat \Pi_{1,2}^{(k_{r-1},k_r]}\hat \Pi_{2,3}^{(0,l_{D-r}']}
(\hat w_2-1) \hat g_2\,d2\,\frac{b_3-\hat b_3}{\bb_3}\hat w_3 \hat g_3\,d3\\
&\qquad=\hat A_1\sum_{r=0}^R\hat \Pi^{(k_{r-1},k_r]}\biggl(\hat\Pi^{(0,l_{D-r}']}\Bigl(\frac{b-\hat b}{\bb}\hat w\Bigr)(\hat w-1)\biggr)_1.
\end{align*}
The terms of the sum are orthogonal and hence the second moment is the sum of the second moments,
which is bounded by $R$ times a multiple of the maximum of the second moments, which is bounded
above by $R \|(b-\hat b)/\bb\|_2^2\,\|\hat w\|_\infty^2\,\|\hat w-1\|_\infty^2\lesssim R\e_n^4$.

$B=\{2\}$.
\begin{align*}
&\E_p\biggl(\sum_{r=0}^RD_{\hat p}\Bigl[\hat A_1\hat \Pi_{1,2}^{(k_{r-1},k_r]}\AA_2\hat \Pi_{2,3}^{(0,l_{D-r}']}\hat Y_3\Bigr]
\given X_2\biggr)\\
&\qquad=\sum_{r=0}^R\int\!\!\int \frac{\hat a_1-\hat a_1}{\aa_1}
D_{\hat p}^2\Bigl[\hat \Pi_{1,2}^{(k_{r-1},k_r]}\AA_2\hat \Pi_{2,3}^{(0,l_{D-r}']}\Bigr]
\,\frac{b_3-\hat b_3}{\bb_3}\hat w_1\hat w_3 \hat g_1\hat g_3\,d1\,d3\\
&\qquad=D_{\hat p}^2\Bigl(\Bigl[\sum_{r=0}^R\hat \Pi^{k_{r-1},k_r]}\Bigl(\frac{\hat a-a}\aa\hat w\Bigr)\Bigr)_2
\hat\Pi^{(0,l_{D-r}']}\Bigl(\Bigl(\frac{b-\hat b}{\bb}\hat w\Bigr)\Bigr)_2\Bigr].
\end{align*}
The degeneracy operator decreases second moment (it merely subtracts the mean in this case), and can be left out.
The terms of the sum appear not
to be orthogonal, but by two applications of the Cauchy-Schwarz inequality the second moment of the sum
can be bounded by
\begin{align*}
&\E_{\hat p}\sum_{r=0}^R\biggl(\hat \Pi^{(k_{r-1},k_r]}\Bigl(\frac{\hat a-a}\aa\hat w\Bigr)_2\biggr)^2
\sum_{r=0}^R\biggl(\hat\Pi^{(0,l_{D-r}']}\Bigl(\frac{b-\hat b}{\bb}\hat w\Bigr)_2\biggr)^2\\
&\qquad\le R^2\max_r 
\biggl(\E_{\hat p}\biggl(\hat \Pi^{(k_{r-1},k_r]}\Bigl(\frac{\hat a-a}\aa\hat w\Bigr)\biggr)^4
\E_{\hat p}\biggl(\hat\Pi^{(0,l_{D-r}']}\Bigl(\frac{b-\hat b}{\bb}\hat w\Bigr)\biggr)^4\biggr)^{1/2}.
\end{align*}
We can decompose $\hat\Pi^{(k_{r-1},k_r]}=\hat\Pi^{(0,k_r]}-\hat\Pi^{(0,k_{r-1}]}$ to see that the
norm of $\hat\Pi^{(k_{r-1},k_r]}: L_4(\aa\bb \hat g)\to L_4(\aa\bb \hat g)$ is bounded above by 
a multiple of the
maximum of the corresponding norms of the operators $\hat \Pi^{(0,l]}$, for $l\le k$. Thus the
expression is bounded by a multiple of $R^2\e_n^4$.

The case $B=\{3\}$ is analogous to the case $B=\{1\}$.

\begin{supplement}%[id=suppA]
  \sname{Supplement}
  \stitle{ {\em Estimation of a Functional on a Structured Model under Low Regularity} }
  \slink[doi]{COMPLETED BY THE TYPESETTER}
  \sdatatype{.pdf}
  \sdescription{The remainder of the paper is given in the supplement.}
\end{supplement}

\bibliographystyle{acmtrans-ims}

\bibliography{U}

%\end{document}
\newpage

\begin{frontmatter}
\title{Supplement to ``Minimax Estimation of a Functional on a Structured High-dimensional Model''}
\runtitle{Minimax Estimation on a Structured Model}

\begin{aug}
\author{James M. Robins\ead[label=e3]{robins@hsph.harvard.edu}},
\author{Lingling Li\ead[label=e2]{lingling07.li@gmail.com}},
\author{Lin Liu\ead[label=e5]{linliu.tju@gmail.com}},
\author{Rajarshi Mukherjee\ead[label=e5]{rajmrt23@gmail.com}},
\author{Eric Tchetgen Tchetgen\ead[label=e1]{etchetge@hsph.harvard.edu}}
\and 
\author{Aad van der Vaart\ead[label=e4]{avdvaart@math.leidenuniv.nl}}
\thankstext{}{The research leading to these results has 
received funding from the European Research Council 
under ERC Grant Agreement 320637, an NWO Spinoza grant,
and from the National Institutes of Health grants AI112339,  AI32475, AI113251 and ES020337.}
\runauthor{Robins et al.}
\affiliation{Harvard University, Karyopharm Therapeutics Inc.,
CMA-Shanghai, Harvard University
and TU Delft}
\address{James M. Robins\\
Harvard T.H. Chan School of Public Health\\
\printead{e3}}
\address{Lingling Li\\
Karyopharm Therapeutics Inc.\\
\printead{e2}}
\address{Lin Liu\\
Institute of Natural Sciences\\
School of Mathematical Sciences\\
CMA-Shanghai, MOE-LSC\\
\printead{e5}}
\address{Rajarshi Mukherjee\\
Department of Biostatistics\\
Harvard Unversity\\
\printead{e5}}
\address{Eric Tchetgen Tchetgen \\
The Wharton School of Business\\
University of Pennsylvania\\
\printead{e1}}
\address{Aad van der Vaart\\
Delft institute of Applied Mathemats\\
TU Delft\\
The Netherlands\\
\printead{e4}}
\end{aug}

\begin{abstract}
This supplement contains the proof of Theorem~\ref{TheoremMinimaxRate} in the
case that $m>3$, and it contains three appendices.
\end{abstract}

\begin{keyword}[class=AMS]
\kwd[Primary ]{62G05, 62G20, 62G20, 62F25}
%\kwd[; secondary ]{}
\end{keyword}

\begin{keyword}
\kwd{Nonlinear functional, nonparametric estimation,
$U$-statistic, influence function, tangent space}
\end{keyword}

\end{frontmatter}

\maketitle

\subsection{Proof  of Theorem~\ref{TheoremMinimaxRate}}
\label{SectionProofofTheoremMinimaxRate}
As in the proof for $m=3$ it suffices to compare
the bias with the bias of the estimator in Theorems~\ref{TheoremParametricRate}. 
In the estimator of order $m>5$ not every of the additional bias terms of orders $j=4,\ldots, m-1$
is individually small, but the sum is small due to a cancellation among these terms. The analysis therefore
requires careful bookkeeping, for which we introduce the following notation.

\def\lowhat{_\wedge}
\def\plowhat{\prescript{}{\wedge}}
\def\pd{\prescript{}{d}}
\def\dash{\hbox{\hskip1pt-\hskip1pt}}
A string $\d^1\d^2\cdots\d^{j-1}$ of symbols $\d^i\in \{0,1,\dash\}$ refers to an expectation of
a variable
\begin{equation}
\label{EqLeadingTermjth}
\hat A_1 \hat \Pi^{\d_1}_{1,2}\AA_2\hat\Pi^{\d^2}_{2,3}\AA_3\cdots\AA_{j-1}\hat\Pi_{j-1,j}^{\d^j}\hat Y_j,
\end{equation}
where
$$\hat\Pi^{\d}_{i,j}=\begin{cases}
 \Pi_{\hat p}^{(0,n]}(Z_i,Z_j),& \text{ if } \d=0,\\
 \Pi_{\hat p}^{(n,k]}(Z_i,Z_j),& \text{ if } \d=1,\\
 \Pi_{\hat p}^{(0,k]}(Z_i,Z_j),& \text{ if } \d=\dash.
\end{cases}$$
Furthermore, a string $\d^1\cdots\d^{i-1}\overline H\d^{i+2}\cdots\d^{j-1}$ refers to the expectation of
a variable
$$\hat A_1 \hat \Pi^{\d_1}_{1,2}\cdots \hat\Pi^{\d^{i-1}}_{i-1,i}\AA_i
\Bigl(\dsum_{r+s>D \atop{r,s\ge1}}\hat\Pi^{(k_{r-1},k_r]}_{i,i+1}\AA_{i+1}\hat\Pi^{(l_{s-1},l_s]}_{i+1,i+2}\Bigr)
\AA_{i+2} \hat\Pi^{\d^{i+2}}_{i+2,i+3}\cdots\hat\Pi_{j-1,j}^{\d^j}\hat Y_j.$$
Thus the symbol $\overline H$ refers to terms in pairs of projection kernels above the hyperbola,
as involved in the construction of the estimator. Similarly, we let
the same strings but with the symbol $H$ instead of $\overline H$
refer to the complementary terms below the hyperbola, but above $n$.

Every of these strings will stand for an expected value; for the first and last variables $X_1$ and $X_j$ this is computed
relative to $p$, but for the middle variables $X_2,\ldots, X_{j-1}$ this is relative to $p-\hat p$. 
We add further notation for expectations on $X_1$ and $X_j$ taken relative to $p-\hat p$ or $\hat p$, by preceding
(for $X_1$) or succeeding (for $X_j$) the string with a subscript ${}_d$ (for the difference $p-\hat p$) or $\lowhat$
(for $\hat p$). This gives, for instance,
\begin{align*}
01\dash 1\lowhat&=\int \frac{a-\hat a}{\aa}(z_1) \hat\Pi^{(0,n]}_{1,2}\hat\Pi^{(n,k]}_{2,3}\hat\Pi^{(0,k]}_{3,4}\hat\Pi^{(n,k]}_{4,5}
\frac{b-\hat b}{\bb}(z_5)\\
&\qqqquad\qqqquad\qqqquad\times\,d\r(z_1)\,d\prod_{i=2}^4(\r-\hat \r)(z_i)\,d\hat\r(z_5),\\
\pd1\dash 1\lowhat&=\int \frac{a-\hat a}{\aa}(z_1) \hat\Pi^{(n,k]}_{1,2}\hat\Pi^{(0,k]}_{2, 3}\hat\Pi^{(n,k]}_{3,4}
\frac{b-\hat b}{\bb}(z_4)\,d\prod_{i=1}^3(\r-\hat \r)(z_i)\,d\hat\r(z_4).
\end{align*}
The notations $d\r=\aa\bb g\, d\nu$, $d\hat\r=\aa\bb \hat g\, d\nu$,  and $\hat \Pi_{i,j}$ are as in
the proof of the theorem for $m=3$, and the (five- and four-fold) integrals arise after conditioning
(\ref{EqLeadingTermjth}) on the variables $Z_i$, as in the same proof.

The $j$th order kernel of the estimator in Theorem~\ref{TheoremParametricRate} corresponds
to the product of $j-1$ projection kernels with ranges $(0,k]$, and 
is represented by a string of $j-1$ dashes: $\dash\dash\cdots\dash$.
To construct the estimator of Theorem~\ref{TheoremMinimaxRate} we partition the range of a single kernel as $(0,k]=(0,n]\cup (n,k]$, or
the range of a contiguous pair of kernels as $(0,k]^2=(0,n]^2\cup H\cup\overline H$. By expanding the corresponding product of sums
of two or three (pairs of) kernels we obtain a decomposition of $\dash\dash\cdots\dash$ into 
sequences with symbols $0,1,\dash, H,\overline H$.
The terms retained in the estimator of Theorem~\ref{TheoremMinimaxRate} are represented
by the sequences $\d^1\ldots \d^{j-1}\in\{0,1,\dash\}^{j-1}$ with $j-1$ or $j-2$ symbols 0, and 
the sequences $H0\cdots0, 0H0\cdots 0, 0\cdots 0 H$. All other terms are left out;
for instance, for $j=3,4,5,6,7$ the terms that are left out are given by the sequences
$$\begin{tabular}{c|c|c|c|c}
$j=3$\quad& $j=4$\quad& $j=5$\quad& $j=6$&$j=7$\\
\hline
\noalign{\vskip2pt}
$\overline H$&1\dash1&1\dash\dash1&1\dash\dash\dash1&1\dash\dash\dash\dash1\\
&$\overline H$0&01\dash1&01\dash\dash1&01\dash\dash\dash1\\
&0$\overline H$&1\dash10&1\dash\dash10&1\dash\dash\dash10\\
&&$\overline H$00&001\dash1&001\dash\dash1\\
&&0$\overline H$0&01\dash10&01\dash\dash10\\
&&$\overline H$00&1\dash100&1\dash\dash100\\
&&&000$\overline H$&0001\dash1\\
&&&00$\overline H$0&001\dash10\\
&&&0$\overline H$00&01\dash100\\
&&&$\overline H$000&1\dash1000\\
&&&&0000$\overline H$\\
&&&&00$\overline H$00\\
&&&&0$\overline H$000\\
&&&&$\overline H$0000
\end{tabular}$$
In this table the strings are categorized by the numbers of 0s on their left and right sides. The
nonzero middle part of a string
always has the form $1\cdots\dash\cdots1$, with at least one $\dash$, or $\overline H$, which may be considered as
taking the place of $11$.

We claim that the difference of the biases of the estimators in Theorems~\ref{TheoremParametricRate} 
and~\ref{TheoremMinimaxRate} is the alternating (on the order) sum of these strings 
(or, rather, of the expectations they represent).
For instance, the extra bias for $m=5$ is equal to the sum of all strings in the table under $j=5$ minus
the strings under $j=4$ plus the string under $j=3$. 

To see this we note first that the leading factorial $j!$ in the definition of the $j$th order influence function $\tilde\chi_p^{(j)}$
in Theorem~\ref{TheoremMARHigherOrderIF} and its reduced version $\chi_p^{(j)}$ in
Theorem~\ref{TheoremMinimaxRate} cancels the factorial in the
definition of the estimator (\ref{EqEstimatorGeneral}), while the factor $(-1)^{j-1}$ causes alternation of signs
between the orders. The extra bias is the sum over $j$ of the expectation under $P^j$ of the sum of the 
terms left out of the $j$th influence function. Because by its construction the influence function is degenerate 
relative to $X_2,\ldots, X_{j-1}$ with respect to $\hat P$, the expectation can be equivalently
taken relative to $p-\hat p$ for these variables. Following this substitution, the projection of the leading term
(\ref{EqLeadingTermjth}), which creates the degeneracy, can be dropped, and the expectation reduces to a number as represented by
one of the strings  $\d^1\ldots \d^{j-1}$ or $\cdots \overline H\cdots$ introduced previously.
This last reasoning is similar as in the proof for $m=3$, where the projection
is shown explicitly.

%Furthermore, terms of the types on the right can be related across orders by using the identities$$\plowhat0\e\cdots\d_s=\pd\e\cdots\d_s,\qquad \prescript{}{s}\e\cdots\d0\lowhat=\e\cdots\d_d,$$ where $s$ can be empty, or $d$ or $\lowhat$, and similarly identities when two of the $\d$s are replaced by $\overline H$. 
We proceed to bound the alternating sum of the ``left-out strings''.
There is cancellation of expectations between terms that are one or two orders apart, i.e.\ of strings
that differ by one or two symbols. The relevant reduction formulas are, for any $\e,\d\in\{0,1,\dash, \overline H\}$ and
any intermediate symbols $\cdots$,
\begin{align}
\plowhat0\e\cdots\d0\lowhat
&=\plowhat0\e\cdots\d_d+\pd\e\cdots\d0\lowhat-\pd\e\cdots\d_d+\plowhat {\bar1}\e\cdots\d{\bar 1}\lowhat,\nonumber\\
\lowhat \e\cdots\d0\lowhat
&=\plowhat\e\cdots\d_d-\plowhat\e\cdots\d{\bar1}\lowhat,\label{EqReductionFormulas}\\
\lowhat0\e\cdots\d\lowhat
&=\pd\e\cdots\d\lowhat-\plowhat{\bar1}\e\cdots\d\lowhat. \nonumber
\end{align}
The rightmost term of each formula is a remainder term, which we may view as being defined by the formula.
The idea of the equations is to remove a symbol 0 at the beginning or end of a string with its mark
$\lowhat$, where the new, shorter string carries mark $_d$. (The second and third formulas, even though valid,
will be used only with $\e$ or $\d$ unequal to 0, respectively.) The first formula is true if the
remainder string $\plowhat {\bar1}\e\cdots\d{\bar 1}\lowhat$ is interpreted as
\begin{align*}
&\int \Bigl[\bigl(\hat \Pi^{(0,n]}-I)\frac{a-\hat a}\aa\Bigr](z_2)\,\hat\Pi^\e_{2,3}\cdots\hat\Pi^\d_{j-2,j-1}\times\\
&\qqqquad\qqqquad\times\Bigl[\bigl(\hat \Pi^{(0,n]}-I)\frac{b-\hat b}\bb\Bigr](z_{j-1})\,d\prod_{i=2}^{j-1}(\r-\hat \r)(z_i).
\end{align*}
Indeed, the four expectations obtained by expanding this last integral on the minus
signs in the two appearances of $\hat\Pi^{(0,n]}-I$ are the four strings 
$\plowhat0\e\cdots\d0\lowhat$, $\plowhat0\e\cdots\d_d$, $\pd\e\cdots\d0\lowhat$, and
$\pd\e\cdots\d_d$ in the first reduction formula, with positive, negative, negative and positive signs.
This follows by integrating the latter strings on the first and/or last variables,
and using identities such as $\int \a(z_1)\hat\Pi^{(0,n]}_{1,2}\,d\hat\r(z_1)=\hat\Pi^{(0,n]}\a(z_2)$.
The second and third formulas are obtained similarly, and more easily, with the appropriate definitions
of the remainder strings.
(The notation $\bar 1$ is motivated by the fact that $\int \a(z_1)\hat\Pi^{(n,k]}_{1,2}\,d\hat\r(z_1)$, 
which is represented by a 1, is equal to $\hat\Pi^{(n,k]}\a(z_2)$, which is $(I-\hat\Pi^{(0,n]})\a(z_2)$
up to terms ``above $k$''.)

We now proceed in two steps to rewrite all strings that make up the difference of the influence functions for $j=4,\ldots, m$.
First we write $p(x_1)=\hat p(x_1)+(p-\hat p)(x_1)$ and similarly for the density of $X_j$, and expand on the plus signs,
to rewrite every string $\e\cdots\d$ as:
\begin{equation}
\label{EqInitialStringExpansion}
\e\cdots \d=\plowhat\e\cdots\d\lowhat+\plowhat\e\cdots\d_d+\pd\e\cdots\d\lowhat
+\pd\e\cdots\d_d.
\end{equation}
Second, if one or both of $\e$ and $\d$ are 0, we expand the first string on the right side (with two $\lowhat$)
using the reduction formulas (\ref{EqReductionFormulas}), where we use the first formula if both
$\e=\d=0$ and the second or third if one of $\e, \d$ is 0. After doing this for all strings up to some order, we end up with:
\begin{enumerate}
\item[(i)] strings of the type $\plowhat\e\cdots\d\lowhat$, with both $\e,\d\in\{1,\overline H,\bar 1\}$.
\item[(ii)] strings of the type $\plowhat\e\cdots\d_d$.
\item[(iii)] strings of the type  $\pd\e\cdots\d\lowhat$.
\item[(iv)] strings of the type $\pd\e\cdots\d_d$.
\end{enumerate}
Note that the reduction formulas (\ref{EqReductionFormulas}), as we applied them, produce a string
$\plowhat \e\cdots\d\lowhat$ with two $\lowhat$ only if $\e\cdots\d$ is of type (i).
We shall show that the strings of type (i) are individually small, while the contributions of the other types are small
after cancellation.

Strings of type (i) can be bounded with the help of Lemma~\ref{LemmaMeanProductOfKernels}.
For $\e\cdots\d$ a string of length $j-3$ with $j\ge 4$, with $\overline H$ counted as having length 2,
$$\bigl|\plowhat 1\e\cdots\d1\lowhat\bigr|
\le \Bigl\|\hat\Pi^{(n,k]}\frac{a-\hat a}\aa\Bigr\|_r  \Bigl\|\hat\Pi^{(n,k]}\frac{b-\hat
  b}\bb\Bigr\|_r\|g-\hat g\|_{(j-2)r/(r-2)}^{j-2}.$$
For $\bar 1$ instead of $1$ the similar statement is true, but with $I-\hat\Pi^{(0,n]}$ replacing
$\hat\Pi^{(n,k]}$. Since the projections are norm-decreasing up to a constant by assumption (and
Lemma~\ref{LemmaProjectionInLs}), the norms in the display are bounded up to a constant
by the norms of the functions $(I-\hat\Pi^{(0,n]})\a$, for $\a=(a-\hat a)/\aa$ or  $\a=(b-\hat
b)/\bb$, respectively.  
Strings of type (i) starting or ending with $\overline H$ and at least one other symbol can be treated in the same manner, as the
kernels in $\overline H$ all start above $n$. These strings minimally give a \emph{square} estimation norm
$\|g-\hat g\|_{(m-2)r/(r-2)}^2$, and are accounted for in the fourth term of the bound on the bias
in Theorem~\ref{TheoremMinimaxRate}. The only string of type (i) of length 2 is 
$\overline H$. In the proof for $m=3$ this was shown to be accounted for by the third term of the bias bound.

Every string of type (ii) arises both from the initial expansion (\ref{EqInitialStringExpansion}) of $\e\cdots\d$, and from the secondary
expansion (\ref{EqReductionFormulas}) of $\plowhat\e\cdots\d0\lowhat$. In the expansions they carry the same sign, but
as they arise at different orders, the alternation of signs in the orders makes them cancel.
The same analysis applies to strings of type (iii). Finally, strings of type (iv) arise from 
the initial expansion of the string $\e\cdots\d$ and from the secondary expansion of the 
string $\plowhat0\e\cdots\d0\lowhat$, with the opposite sign. As the latter strings arise at
orders that differ by two, they also cancel.

If we consider terms up to order $m$, then the strings that cancel versus strings at orders $m+1$ or $m+2$ are left.
These are the strings $\pd\e\cdots\d_d$ of length $m-2$ and $m-1$, with $\overline H$ counted as two symbols,
and the strings $\pd\e\cdots\d\lowhat$ and $\plowhat\e\cdots\d_d$ of length $m-1$.
In view of Lemma~\ref{LemmaMeanProductOfKernels}, for $\e\cdots\d$ of length $j-1$,
\begin{align}
\nonumber |\pd\e\cdots\d_d|
&\lesssim \Bigl\|\hat\Pi^{\e}\Bigl(\frac{a-\hat a}\aa \frac{g-\hat g}{\hat g}\Bigr)\Bigr\|_s
\Bigl\|\hat\Pi^{\d}\Bigl(\frac{b-\hat b}\bb \frac{g-\hat g}{\hat g}\Bigr)\Bigr\|_s\|g-\hat g\|_{(j-2)s/(s-2)}^{j-2}\\
&\lesssim \Bigl\|\frac{a-\hat a}\aa\bigr\|_{sp}\Bigl\| \frac{g-\hat g}{\hat g}\Bigr\|_{sq}^2
\Bigl\|\frac{b-\hat b}\bb\Bigr\|_{sp} \|g-\hat g\|_{(j-2)s/(s-2)}^{j-2}.
\label{EqRemainingEstimationBias}
\end{align}
The choices $s=rj/(j+r-2)$, $p=(j+r-2)/j$ and $q=(j+r-2)/(r-2)$ give $sp=r$ and
$sq=(j-2)s/(s-2)=jr/(r-2)$, and then this term is bounded above by the first term in the bias of 
Theorem~\ref{TheoremMinimaxRate}. The strings $\pd\e\cdots\d\lowhat$ and $\plowhat\e\cdots\d_d$ 
can be handled similarly; only one of the two extremes yields a factor $g-\hat g$, but we need
to consider these strings only of length $m-1$.

This concludes the derivation of the bias. The variance is bounded by a weighted sum of the variances of
the third order estimator, and the variances of the variables $\UU_n\chi^{(j,i)}_{\hat p}$ over $i=1,\ldots, j-2$ and $j=4,\ldots, m$,
for the influence functions given in (\ref{EqHigherOrderCutIF}). 
By Lemma~\ref{LemVarUBound}
$$\var \UU_n\chi_{\hat p}^{(i,j)}
\lesssim \sum_{l=1}^j\frac 1{n^l}\E_p\biggl(\E_p\Bigl(\tilde\chi_{\hat p}^{(i,j}(X_1,\ldots,X_j)
\given X_1,\ldots,X_l\Bigr)\biggr)^2.$$
The second moment in the right side is bounded above by a multiple of
\begin{align*}
&\max_{{B\subset \{1,\ldots,j\}}\atop{|B|=l}} 
\E_{\hat p}\biggl(\E_p\biggl(D_{\hat p}\Bigl[\hat A_1 \hat\Pi_{1,2}^{(0,n]}\AA_2\times\cdots\times\AA_{i-1}\hat\Pi_{i-1,i}^{(0,n]}\AA_i\times\\
&\times\Bigl[\sum_{r=0}^R\hat\Pi_{i,i+1}^{(k_{r-1},k_r]}\AA_{i+1}\hat\Pi_{i+1,i+2}^{{0,l_{D-r}']}}\Bigr]
\AA_{i+2}\hat\Pi_{i+2,i+3}^{(0,n]}\times\cdots\times\AA_{j-1}\hat\Pi_{j-1,j}^{(0,n]}\hat Y_j\Bigr]
\given X_B\Bigr)\biggr)^2.
\end{align*}
This can be bounded by a combination of the arguments
used in the proofs of Theorem~\ref{TheoremParametricRate}  and Theorem~\ref{TheoremMinimaxRate} for $m=3$
in the preceding sections. 
Here we can use the Cauchy-Schwarz inequality
$(\sum_r x_r)^2\le R \sum_r x_r$ to handle the sum over $r$, 
and the bound $\hat \Pi^{(k_{r-1},k_r]}(z,z)\le \hat \Pi^{(0,k_r]}(z,z)\lesssim k_r$,
which is asymptotic to $n2^{r/\a}\asymp k_r-k_{r-1}$.

The case that $|B|=j$ can be handled by changing measure from $p$ to $\hat p$, bounding out $\hat A_1$ and $\hat Y_j$, 
and dropping the degeneracy operators. Then left is the square of the third order kernel at the variables
$i, i+1, i+2$ premultiplied and postmultiplied by product of square kernels. In the spirit of
Lemma~\ref{LemmaVarianceProductOfKernels} the opening and closing kernels can be integrated out, both from the
left and the right, and bounded out by the supremum on their diagonal, until only the
hyperbolic, third order part of the influence function remains. The suprema on the diagonal are bounded
above by a multiple of $n$ by assumption, and the hyperbolic part can be bounded exactly 
as in the proof for $m=3$. 

Next consider first the case that $|B|<j$ and $B$ contains both 1 and $j$. Then we bound out $\hat A_1$ and $\hat Y_1$,
and apply Lemma~\ref{LemmaIntegrateOutKernels} to rewrite the conditional expectation. The degeneracy operators
$D_{\hat p}^{b_r+1,b_{r+1}-1}$ in the right side of the lemma commute with the integrals and can  be left
off when taking the second moment to obtain the bound, with $\hat\z_l=\hat w_l-1$ if $l\notin B$ and
$\hat \z_l=\hat w_l$ if $l\in B$,
\begin{align*}
&R\,\E_{\hat p}\sum_r\biggl(\int\cdots\int\hat\Pi_{1,2}^{(0,n]}\hat\z_2\times\cdots\times\hat\z_{i-1}\hat\Pi_{i-1,i}^{(0,n]}\hat\z_i\\
&\times\hat\Pi_{i,i+1}^{(k_{r-1},k_r]}\hat\z_{i+1}\hat\Pi_{i+1,i+2}^{(0,l_{D-r}']}
\hat\z_{i+2}\hat\Pi_{i+2,i+3}^{(0,n]}\times\cdots\times\hat\z_{j-1}\hat\Pi_{j-1,j}^{(0,n]}
\prod_{l\notin B}(\hat g_l\,dl)\biggr)^2.\end{align*}
The last factor is bounded above by $\|\hat w\|_\infty^{2|B|}\|\hat w-1\|_\infty^{2(j-|B|)}$.
We peel this off from right to left. Every $l\in B$, which is  integrated outside the square, turns a square kernel
 $\hat\Pi_{l-1,l}^2$ into $\hat\Pi_{l-1,l-1}$, which is bounded above by $n$, $l_{D-r}'$ or $k_r$, depending on
whether it concerns a kernel $\hat\Pi^{(0,n]}$, $\Pi^{(0,l_{D-r}']}$ or $\hat\Pi^{(k_{r-1},k_r]}$. Every $l\notin B$,
which is integrated within the square, can be viewed as applying the operator to a function and can simply
be bounded out. This results in the upper bound
$$R\sum_r n^{|B-\{i+1,i+2\}|-1}k_r^{1\{i+1\in B\}}(l_{D-r}')^{1\{i+2\in B\}}\prod_{l=2}^j\|\hat\z_l\|_\infty^2.$$
If both $i+1$ and $i+2$ are contained in $B$, then we use that
$\sum_r k_rl_{D-r}'\lesssim nk+n^2 D2^{D/\a\vee D/\b}$, as used in the paper, and obtain
the upper bound $R(k+n^2D2^{D/\a\vee D/\b})n^{|B|-3}\e_n^{2(j-|B|)}$.
If $i+1$ is contained in $B$ and $i+2$ is not, then we bound $\sum_r k_r\lesssim \sum_r(k_r-k_{r-1})\le k$,
and we obtain the upper bound $Rk n^{|B|-2}\e_n^{2(j-|B|)}$.
If neither $i+1$ nor $i+2$ is contained in $B$, then we use that $\sum_r1=R$, and we obtain
the upper bound $R^2n^{|B|-1}\e_n^{2(j-|B|)}$. These upper bounds divided by $n^{|B|}$ contribute to the
variance bound.

The case that $B$ contains $j$ but not 1 can be handled by the same argument, except that we do not
bound out $\hat A_1$, but perform the integral 
 $\int (\hat a_1- a_1)/\aa_1\hat\Pi_{1,2}\hat w_1\hat g_1\,d1=\hat \Pi \bigl((\hat a-a)\hat w/\aa\bigr)_2$.
This then replaces the kernel $\hat\Pi_{1,2}$ in the preceding argument.

The case that $B$ contains $1$ but not $j$ is similar, if we perform  the peeling argument from left
to right.

The case that $B$ contains neither $1$ nor $j$ seems to be different in that the peeling cannot
start with a kernel but must deal with the term 
$$\E_p\bigl( \hat\Pi_{j-2,j-1}\AA_{j-1}\hat\Pi_{j-1,j} \hat Y_j\given X_{j-1}\bigr)
= \hat\Pi_{j-2,j-1}(\hat\Pi \bigl((b-\hat b)\hat w/\bb\bigr)_{j-1}.$$
If the operator $\hat \Pi$ would be continuous for the uniform norm, then we could
bound out the function $(\hat\Pi \bigl(b-\hat b)\hat w/\bb\bigr)$ and the argument could proceed 
as before. We can relax this assumption on the operator to continuity relative to $L_4$, by the argument as
in the proof of Lemma~\ref{LemSecondMoment}. We first integrate both on the left and the right side
over all variables $1,2,\ldots, r$ and $j,j-1,\ldots, j-l$ that do not belong to $B$, where
$r+1\in B$ and $j-l-1\in B$ to reduce to

$$\E_{\hat p}\biggl(\Phi\Bigl(\frac{\hat a-a}\aa\Bigr)_{r+1}\int\!\!\cdots\!\int \!\hat\Pi_{r+1,r+2}\cdots
\hat\Pi_{b_{r'}-2,b_{r'}-1}\,\prod_{i=r}^{r'-1}(\hat w_{b_i}-1)\hat g_{b_i}\,db_i\,
\Psi\Bigl(\frac{b-\hat b}\bb\Bigr)_{b_{r'}-1}\biggr)^2.$$
Here $\Phi$ is the repeated operator $a\mapsto \hat\Pi\bigl(\cdots\hat\Pi(a\hat w)(\hat w-1)\bigr)$ appearing
in the preceding display, and $\Psi$ is defined likewise from right to left. The kernels are 
$\hat\Pi^{(0,n]}$ or $\hat \Pi^{(k_{r-1},k_r]}$ or $\hat \Pi^{(0,l_{D-r}']}$.
By the Cauchy-Schwarz inequality this is bounded above by the square root of
\begin{align*}
&\E_{\hat p}\biggl(\Phi\Bigl(\frac{\hat a-a}\aa\Bigr)_{r+1}^2
\int\!\cdots\!\int \hat\Pi_{r+1,r+2}\cdots\hat\Pi_{b_{r'}-1,b_{r'}}
\,\prod_{i=r}^{r'}(\hat w_{b_i}-1)\hat g_{b_i}\,db_i\biggr)^2\\
&\qquad\times\E_{\hat p}\biggl(\int\!\cdots\!\int \hat\Pi_{r+1,r+2}\!\cdots\!\hat\Pi_{b_{r'}-2,b_{r'}-1}
\,\prod_{i=r}^{r'-1}(\hat w_{b_i}-1)\hat g_{b_i}\,db_i\,
\hat\Psi\Bigl(\frac{b-\hat b}\bb\Bigr)_{b_{r'}-1}^2\biggr)^2.
\end{align*}
The two expectations in this display can be bounded by peeling off the kernels from right to left,
and from left to right, respectively, in the same way as before. At the end of the process
this leaves the fourth moments of $\Phi\bigl((\hat a-a)/\aa\bigr)$
and $\Psi\bigl((b-\hat b)/\bb\bigr)$. The roots of 
these moments  can be bounded by $\|\Phi\|_4^2 \|(\hat a-a)\hat w/\aa\bigr\|_4^2$
and $\|\Psi\|_4^2 \|(\hat b-b)\hat w/\bb\bigr\|_4^2$, for $\|\Phi\|_4$ and $\|\Psi\|_4$ the norms
of these operators in $L_4$, which can be bounded in terms of the $L_4$ 
norms of the projections $\Pi_{\hat p}$ and the uniform norm of $\hat w$.

\subsection{Auxiliary lemmas concerning variances}
A difficulty in bounding the variances of the estimators is that the expectations
are under $p$, but the influence function is evaluated under $\hat p$, 
so that the degeneracy operator involved in the definition of the influence function (see Theorem~\ref{TheoremMARHigherOrderIF})
is applied under $\hat p$. This mismatch creates many additional terms in the
upper bound. The degeneracy operator (see (\ref{EqMakeDegenerate}) consists of subtracting
an alternating sum of conditional expectations relative to subsets of variables. In the context
of Theorem~\ref{TheoremMARHigherOrderIF} it works only on the variaables $X_2,\ldots, X_{m-1}$, as the
product displayed in the theorem is already degenerate relative to $X_1$ and $X_m$.
Taking a conditional expectation is equivalent to
integrating out the other variables. In the next two lemmas
 we see that every integration ``removes'' one of the kernels $\hat \Pi_{i,i+1}$.
Because a square kernel has second moment of order $k$, this reduces the second moment of the integrated kernel
from order $k^{j-1}$ to $k^{l-1}$, if $l$ variables remain. In the variance bound these powers are divided by 
the reduced powers $n^l$.

The following lemma gives the formula for
integrating out one or more of the middle variables. For instance, with $\hat w_i=(g/\hat g)(Z_i)$,

{\footnotesize{
\begin{align*}
\E_p\Big(\hat D_p[\hat\Pi_{1,2}\AA_2\hat\Pi_{2,3}]\given X_1,X_3\Bigr)
&=\E_{\hat p}\bigl(\hat\Pi_{1,2}(\hat w_2-1)\hat\Pi_{2,3}]\given X_1,X_3\bigr),\\
\E_p\Big(\hat D_p[\hat\Pi_{1,2}\AA_2\hat\Pi_{2,3}\AA_3\hat\Pi_{3,4}]\given X_1,X_4\Bigr)
&=\E_{\hat p}\bigl(\hat\Pi_{1,2}(\hat w_2-1)\hat\Pi_{2,3}(\hat w_3-1)\hat\Pi_{3,4}]\given X_1,X_4\bigr),\\
\E_p\Big(\hat D_p[\hat\Pi_{1,2}\AA_2\hat\Pi_{2,3}\AA_3\hat\Pi_{3,4}]\given X_1,X_3,X_4\Bigr)
&=\E_{\hat p}\bigl(\hat\Pi_{1,2}(\hat w_2-1)(\hat\Pi_{2,3}\AA_3\hat\Pi_{3,4}-\hat\Pi_{2,4})\given X_1,X_3,X_4\bigr).
\end{align*}\par}}

\begin{lemma}
\label{LemmaIntegrateOutKernels}
For $2\le b_1<\cdots <b_s\le j-1$, with $\hat w=g/\hat g$, and $D_p^{k,l}$ the operation of making
degenerate relative to the variables $X_k,\ldots, X_l$ relative to their law given by  $p$,
\begin{align*}
&\E_p\Bigl(D_{\hat p}[\hat\Pi_{1,2}\AA_2\hat\Pi_{2,3}\cdots\AA_{j-1}\hat \Pi_{j-1,j}\bigr]\given X_{\{1,\ldots,j\}-\{b_1,\ldots,b_s\}}\Bigr)\\
&\  =\int\cdots\int\biggl(D_{\hat p}^{2,b_1-1}\Bigl[\hat\Pi_{1,2}\!\!\!\prod_{i=2}^{b_1-1}\!\!\AA_i\hat\Pi_{i,i+1}\Bigr]
\bigl(\hat w(z_{b_1})-1\bigr)\times\\
&\quad\times D_{\hat p}^{b_1+1,b_2-1}\Bigl[\hat\Pi_{b_1,b_1+1}\!\!\!\prod_{i=b_1+1}^{b_2-1}\!\!\AA_i\hat\Pi_{i,i+1}\Bigr]
\bigl(\hat w(z_{b_2})-1\bigr)\times\\
&\qquad
\cdots\times D_{\hat p}^{b_s+1,j-1}\Bigl[\hat\Pi_{b_s,b_s+1}\!\!\!\prod_{i=b_s+1}^{j-1}\!\AA_i\hat\Pi_{i,i+1}\Bigr]
\biggr)\,\prod_{i=1}^s (\aa\bb)(z_{b_i})\,d\hat G(z_{b_1})\cdots d\hat G(z_{b_s}).
\end{align*}
\end{lemma}

\begin{proof}
Consider first the case $s=1$, and abbreviate $b_1=b$ and $g_b\,db=(\aa\bb)(z_b)\,dG(z_b)$.
Then the left side of the theorem is
the expectation of 
\begin{align}
&D_{\hat p}[\hat\Pi_{1,2}\AA_2\hat\Pi_{2,3}\cdots\AA_{j-1}\hat \Pi_{j-1,j}\bigr]\label{EqDegnerateChij}\\
&\quad=\sum_{l=0}^{j-2}(-1)^{j-l}\!\!\!\sum_{1<i_1<\cdots<i_l<j}
\Pi_{1,i_1}\AA_{i_1}\Pi_{i_1,i_2}\cdots\AA_{i_l}\Pi_{i_l,j}\Bigr]\nonumber
\end{align}
with respect to $X_b$ only, keeping the other
variables fixed. Taking first the conditional expectation on $X_b$ given $Z_b$ reduces $\AA_b$ to $(\aa\bb/a)(Z_b)$.
Next  integrating out $Z_b$ gives, with $\hat w_b=g_b/\hat g_b$,
\begin{align*}
&\sum_{l=0}^{j-2}(-1)^{j-l}\Bigl(\sum_{\substack{1<i_1<\cdots<i_l<j\\ b\in \{i_1,\ldots, i_l\}}}
\int \hat\Pi_{1,i_1}\AA_{i_1}\hat\Pi_{i_1,i_2}\cdots\AA_{i_l}\hat\Pi_{i_l,j}\hat w_b\hat g_b\,db\\
&\qquad\qquad\qquad\qquad\qquad+ \sum_{\substack{1<i_1<\cdots<i_l<j\\ b\notin \{i_1,\ldots, i_l\}}}
\hat\Pi_{1,i_1}\AA_{i_1}\Pi_{i_1,i_2}\cdots\AA_{i_l}\Pi_{i_l,j}\Bigr).
\end{align*}
If $b=i_k$, then in view of Lemma~\ref{LemmaMARMakingDegenerate}(iii)),
$$\int \hat\Pi_{i_{k-1},b}\hat w_b\hat\Pi_{b_,i_{k+1}}\hat g_b\,db
=\int \hat\Pi_{i_{k-1},b}(\hat w_b-1)\hat\Pi_{b_,i_{k+1}}\hat g_b\,db+\hat\Pi_{i_{k-1},i_{k+1}}.$$
We substitute this in the first sum in the display, splitting this into two sums. The second
sum has one factor $\hat\Pi_{r,s}$ less and does not involve $b$. This sum is the same as the
second term in the second last display, but with $l-1$ instead of $l$ and different sign. Hence when summed over $l$
these sums cancel. (The sum over $l$ of the first term can be restricted to $l=1..j-2$, of the second to $l=0..j-1$,
as $b$ is assumed to be one of the $i_k$ or not.)
The remaining sum can be reorganised in two sums and written 
\begin{align*}
&\sum_{k=0}^{b-1}\sum_{l=k+1}^{k+j-1-b}(-1)^{b-(k-1)}\!\!\!\!\!\!\sum_{1<i_1<\cdots<i_{k-1}<b}(-1)^{j-b-(r-k)}\!\!\!\!\!\!\sum_{b<i_{k+1}<\cdots<i_l<j}\\
&\qquad\qquad\qquad\qquad\int \hat\Pi_{1,i_1}\AA_{i_1}\cdots\hat\Pi_{i_{k-1},b}(\hat w_b-1)\hat\Pi_{b,i_{k+1}}\cdots\AA_{i_l}\hat\Pi_{i_l,j}\hat g_b\,db.
\end{align*}
This can be written in the form as in the statement of the lemma by interchanging the sums and the integral.

The case $s>1$ can be handled by induction on $s$.
If we integrate out the second variable $X_{b_2}$, then this affects only the term 
$D_{\hat p}\bigl[\hat\Pi_{b_1,b_1+1}\AA_{b_1+1}\cdots \hat\Pi_{j-1,j}\bigr]$. It follows by induction that
this is transformed as claimed.
\end{proof}

The following lemma is used to bound the variance in the proof of Theorems~\ref{TheoremParametricRate} 
and~\ref{TheoremParametricRateEfficient}.

\begin{lemma}
\label{LemSecondMoment}
For any $B\subset \{1,\ldots, j\}$ and a constant $M$ that depends on the supremum norms
of $\aa,\bb, 1/a, b, p/\hat p,g/\hat g$ and the norm of $\Pi_{\hat p}: L_4(\aa\bb \hat g)\to L_4(\aa\bb \hat g)$ only,
\begin{align*}
&\E_p\biggl(\E_p \Bigl(D_{\hat p}\bigl[\hat A_1\hat\Pi_{1,2}\AA_2\hat\Pi_{2,3}\cdots\AA_{j-1}\hat\Pi_{j-1,j}\hat Y_j\bigr]
\given X_B\Bigr)\biggr)^2\\
&\qquad\qquad\le M^j
%\Bigl\|\frac p{\hat p}\Bigr\|_\infty^{|B|}
\Bigl(\Bigl\|\frac{a-\hat a}\aa\Bigr\|_4\vee\Bigl\|\frac{b-\hat b}\bb\Bigr\|_4
\vee\Bigl\|\frac{g-\hat g}{\hat g}\Bigr\|_\infty\Bigr)^{2(j-|B|)} k^{|B|-1}.
\end{align*}
\end{lemma}

\begin{proof}
We change measure from $p$ to $\hat p$ in the leftmost expectation $\E_p$, replacing this by $\E _{\hat p}$,
bounding out the quotient $p/\hat p$ by its supremum norm. The multiplicative constant 
$\|p/\hat p\|_\infty^{|B|}$ can be incorporated in the factor $M^j$.

Consider first the case that $B$ contains $1$ and $j$.
Then the conditional expectation in the left side takes the form as in the right side of
Lemma~\ref{LemmaIntegrateOutKernels}, with
$\{b_1,\ldots, b_s\}=\{1,\ldots j\}-B$, premultiplied by $\hat A_1$ and postmultiplied by $\hat Y_j$.
After squaring,  we bound out the factors $\hat A_1$ and $\hat Y_j$ by their supremum norms. 
We are left with bounding the second moment
of the right side of Lemma~\ref{LemmaIntegrateOutKernels} under $\hat p$.
By linearity, the degeneracy operators $D_{\hat p}^{2,b_1-1},\ldots, D_{\hat p}^{b_s+1,j-1}$ in this right side can swap order
with the integrals relative to $z_{b_1},\ldots z_{b_s}$. It follows that the expectation of the square of the
integrals relative to the variables $X_B$
becomes bigger when removing the degeneracy operators, as making degenerate is a projection and
projection cuts second moment. Thus
we consider the expectation of the square of the right side of Lemma~\ref{LemmaIntegrateOutKernels}, 
without the degeneracy operators: the square of 
\begin{align*}
&\int\cdots\int\hat\Pi_{1,2}\prod_{i=2}^{b_1-1}\AA_i\hat\Pi_{i,i+1}(\hat w_{b_1}-1)
\hat\Pi_{b_1,b_1+1}\prod_{i=b_1}^{b_2-1}\AA_i\hat\Pi_{i,i+1}(\hat w_{b_2}-1)\times\cdots\\
&\qquad\qquad\qquad
\times(\hat w_{b_s}-1)\hat\Pi_{b_s,b_s+1}\prod_{i=b_s+1}^{j-1}\AA_i\hat\Pi_{i,i+1}
\,\prod_{i=1}^{s}(\aa\bb)(z_{b_i})\,d\hat G(b_1)\cdots d\hat G(b_{s}).
\end{align*}
We bound this by peeling off the factors 
indexed by $b_s, b_{s-1},\ldots, b_1$ from right to left (for definiteness), as follows.

For fixed values of $X_1,\ldots, X_{b_s-1}$, let $h$ be the function of $X_{b_s}$ given by the
first $s-1$ factors of the right side of Lemma~\ref{LemmaIntegrateOutKernels}, i.e.
\begin{align*}
&h(Z_{b_s})= \int\cdots\int\biggl(\hat\Pi_{1,2}\!\!\!\prod_{i=2}^{b_1-1}\!\!\AA_i\hat\Pi_{i,i+1}(\hat w_{b_1}-1)
\times\cdots\\
&\qquad
\times \hat\Pi_{b_{s-1},b_{s-1}+1}\!\!\!\prod_{i=b_{s-1}+1}^{b_s-1}\!\AA_i\hat\Pi_{i,i+1}
\bigl(\hat w_{b_s}-1)\biggr)\,\prod_{i=1}^{s-1}(\aa\bb)(z_{b_i})\,d\hat G(b_1)\cdots d\hat G(b_{s-1}).
\end{align*}
The argument $Z_{b_s}$ is hidden in the rightmost kernel $\hat\Pi_{i,i+1}$ for $i=b_s-1$ and in the factor $\hat w_{b_s}-1$. 
Let $(\hat\Pi h)_w=\int h(v)\hat\Pi(v,w)\,(\aa\bb)(w)\,d\hat G(w)$ 
be the operator defined by the kernel $\hat \Pi$ acting on the function $h$.\footnote{We denote by
$\hat\Pi_{i,j}$ a kernel evaluated at $(Z_i,Z_j)$ and by $(\hat \Pi h)_i$  the operator applied
to $h$ evaluated at $Z_i$.} Then
\begin{align*}
&\int h(z_{b_s})\hat\Pi_{b_s,b_s+1}\!\!\!\prod_{i=b_s+1}^{j-1}\!\AA_i\hat\Pi_{i,i+1}\, 
(\aa\bb)(z_{b_s})\,d\hat G(z_{b_s})
= (\hat\Pi h)_{b_s+1}\!\!\!\prod_{i=b_s+1}^{j-1}\!\AA_i\hat\Pi_{i,i+1}.
\end{align*}
Thus the left side of the lemma is bounded above by a multiple of 
$$\E_{\hat p} \Bigl((\hat\Pi h)_{b_s+1}\!\!\!\prod_{i=b_s+1}^{j-1}\!\AA_i\hat\Pi_{i,i+1}\Bigr)^2.$$
In view of the formula
$\int \hat\Pi(u,v)\hat\Pi(v,w)\,(\aa\bb)(v)\,d\hat G(v)=\hat\Pi(u,w)$, 
taking the expectation under $\hat p$ relative to $X_j$, turns the kernel $\hat \Pi_{j-1,j}^2$
into $\hat\Pi_{j-1,j-1}$. Since $\hat \Pi_{j-1,j-1}\le k$, by assumption,
we can remove the square kernel $\hat\Pi_{i,i+1}^2$ for $i=j-1$ at the cost of a multiplicative factor $k$. We repeat this on 
the kernels $\hat \Pi_{i,i+1}^2$, for $i=j-2, j-3,\ldots, b_s+1$,
 until we arrive at the upper bound $\E_{\hat p} \bigl((\hat\Pi h)_{b_{s+1}}\bigr)^2\, k^{j-1-b_s}$.
The latter expression increases if we leave off the projection operator $\hat \Pi$. Next we bound out the factor
$\hat w_{b_s}-1$ hidden in the function $h$  by its uniform norm . 
We then have succeeded in removing the $b_s$ term at the cost of the multiplicative factor $\|\hat w-1\|_\infty^2k^{j-1-b_s}$.
Repeating this process for $b_{s-1}, \cdots, b_2$, we end up with the upper bound
$$\E_{\hat p}\Bigl(\hat\Pi_{1,2}\!\!\!\prod_{i=2}^{b_1-1}\!\!\AA_i\hat\Pi_{i,i+1}\Bigr)^2
\|\hat w-1\|_\infty^{2s}k^{j-s-b_1}.$$
We finish by peeling off the kernels $\hat\Pi_{i,i+1}^2$, for $i=b_1-1,\ldots, 1$.

Next consider the case that $B$ contains $j$,  but not $1$.
Let $\{b_1,\ldots, b_s\}=\{2,\ldots j-1\}-B$, so that $B=\{1,\ldots j\}-\{1,b_1,\ldots, b_s\}$.
 In this case the expectation relative to $X_1$ must be taken on the expression in 
Lemma~\ref{LemmaIntegrateOutKernels} before squaring it. The degeneracy operators work only on
$X_2,\ldots, X_{j-1}$, commute with the expectation on $X_1$, and hence can again be bounded out.
Next we peel off the rightmost blocks $b_s,b_{s-1},\ldots,$
in the same way as before, but treat the leftmost blocks differently.
%\footnote{If the first $r$ numbers in
%$\{1,b_1,\ldots, b_s\}$ are $1,2,\ldots, r$ and $r+1$ is not the next number, then 
%we  perform the integrations on $X_1,\ldots,X_r$ before applying the peeling argument.
%For the present case we could also peel down to $Z_1$, but this strategy seems better 
%in the case that $B$ contains neither 1 nor $j$.}
We start by noting that
$$\E_p(\hat A_1\hat\Pi_{1,2}\given X_2)=\int \frac{\hat a_1-a_1}{\aa_1}\, \hat\Pi_{1,2}\hat w_1(\aa\bb)_1\,d\hat G(z_1)
=\biggl(\hat\Pi\Bigl(\frac{\hat a-a}\aa\hat w\Bigr)\biggr)_2.$$
If $b_1=2$, we also integrate out the kernel $\hat\Pi_{2,3}$ to obtain
$$\E_p(D_{\hat p}\bigl[\hat A_1\hat\Pi_{1,2}\AA_2\hat\Pi_{2,3}\bigr]\given X_3)=
\biggl(\hat\Pi\Bigl(\hat\Pi\Bigl(\frac{\hat a-a}\aa\hat w\Bigr)(\hat w-1)\Bigr)\biggr)_3.$$
We continue this to the first $b_r$ with $b_r>r+1$ (hence $b_1=2,\ldots, b_{r-1}=r$ and $X_{r+1}$ is the
first variable not belonging to $X_B$). The remaining kernels $\hat\Pi_{i,i+1}$, for $i\ge r$, are peeled
off from the right as before, but the peeling argument is stopped at block $b_r$, leaving the term
$$\E_{\hat p}\biggl(\hat\Pi\biggl(\cdots \hat\Pi\Bigl(\hat\Pi\Bigl(\frac{\hat a-a}\aa\hat w\Bigr)
(\hat w-1)\Bigr)(\hat w-1)\bigg)\biggr)_{r+1}^2,$$
where there are $r$ repetitions of the operator $\hat \Pi$.
The latter expectation is bounded by $\bigl\|(a-\hat a)/\aa\bigr\|_2^2 \|\hat w\|_\infty\|\hat w-1\|_\infty^{r-1}$.

The case that $B$ contains $1$, but not $j$, proceeds similarly, but now the peeling
process should work from left to right, ending up with a repeated weighted projection 
of the function $(b-\hat b)\hat w/\bb$.

If $B$ contains neither $1$ nor $j$, then we first treat the conditional expectations relative
to $X_1$ and $X_j$ as previously, transforming the second moment to be bounded into
$$\E_{\hat p}\biggl(\Phi\Bigl(\frac{\hat a-a}\aa\Bigr)_{r+1}\int\!\!\cdots\!\int \!\hat\Pi_{r+1,r+2}\cdots
\hat\Pi_{b_{r'}-2,b_{r'}-1}\,\prod_{i=r}^{r'-1}(\hat w_{b_i}-1)\hat g_{b_i}\,db_i\,
\Psi\Bigl(\frac{b-\hat b}\bb\Bigr)_{b_{r'}-1}\biggr)^2.$$
Here $\Phi$ is the repeated operator $a\mapsto \hat\Pi\bigl(\cdots\hat\Pi(a\hat w)(\hat w-1)\bigr)$ appearing
in the preceding display, and $\Psi$ is defined likewise from right to left.
By the Cauchy-Schwarz inequality this is bounded above by the square root of
\begin{align*}
&\E_{\hat p}\biggl(\Phi\Bigl(\frac{\hat a-a}\aa\Bigr)_{r+1}^2
\int\!\cdots\!\int \hat\Pi_{r+1,r+2}\cdots\hat\Pi_{b_{r'}-1,b_{r'}}
\,\prod_{i=r}^{r'}(\hat w_{b_i}-1)\hat g_{b_i}\,db_i\biggr)^2\\
&\qquad\times\E_{\hat p}\biggl(\int\!\cdots\!\int \hat\Pi_{r+1,r+2}\!\cdots\!\hat\Pi_{b_{r'}-2,b_{r'}-1}
\,\prod_{i=r}^{r'-1}(\hat w_{b_i}-1)\hat g_{b_i}\,db_i\,
\hat\Psi\Bigl(\frac{b-\hat b}\bb\Bigr)_{b_{r'}-1}^2\biggr)^2.
\end{align*}
The two expectations in this display can be bounded by peeling off the kernels from right to left,
and from left to right, respectively, in the same way as before. At the end of the process
this leaves the fourth moments of $\Phi\bigl((\hat a-a)/\aa\bigr)$
and $\Psi\bigl((b-\hat b)/\bb\bigr)$. The roots of 
these moments  can be bounded by $\|\Phi\|_4^2 \|(\hat a-a)\hat w/\aa\bigr\|_4^2$
and $\|\Psi\|_4^2 \|(\hat b-b)\hat w/\bb\bigr\|_4^2$, for $\|\Phi\|_4$ and $\|\Psi\|_4$ the norms
of these operators in $L_4$, which can be bounded in terms of the $L_4$ 
norms of the projections $\Pi_{\hat p}$ and the uniform norm of $\hat w$.
\end{proof}

\subsection{Proof of Theorem~\ref{TheoremMARHigherOrderIF}}
\label{SectionDerivationHigherOrderIFApproximate}
We compute the higher order influence functionals of the approximate
functional $\tilde\chi$ using algorithm [1]--[3] in
Section~\ref{SectionComputingIF}. This starts by computing the second
order influence function as the derivative of $p\mapsto
\tilde\chi_p^{(1)}(x_1)+\chi(\tilde p)$, for fixed $x_1$. Because the
latter functional (given in (\ref{EqMARIFApprox})) depends on
the parameters only through $\tilde a(z_1)$ and $\tilde b(z_1)$, the
following lemma does the main part of the work.

\begin{lemma}
\label{LemmaMARIFofaaAndbb}
For fixed $z_1$ influence functions of $p\mapsto \tilde a(z_1)$ and 
$p\mapsto\tilde b(z_1)$ are given by
\begin{align*}
x_2&\mapsto -\aa(z_1)\Pi_p(z_1,z_2)
\bigl(a_2\tilde a(z_2)-1\bigr)\bb(z_2),\\
x_2&\mapsto \ \bb(z_1)\Pi_p(z_1,z_2)
a_2\bigl(y_2-\tilde b(z_2)\bigr)\aa(z_2),
\end{align*}
where $\Pi_p$ is the kernel of the orthogonal projection
in $L_2(\aa\bb g)$ onto $L$.
\end{lemma}

\begin{proof}
We can write the equation (\ref{EqApproxa}) determining $\tilde a$ as
$\E \bigl(A\tilde a(Z)-1\bigr)\bb(Z)l(Z)=0$, for every $l\in L$.
Insert a sufficiently regular 
path $p_t$, given by parameters $(a_t,b_t,f_t)$,
and differentiate the equality relative to $t$ at $t=0$
to find, with $\g$ a score function of the path
$$\E \frac d{dt}_{|t=0}\tilde a_t(Z)A\bb(Z)l(Z)=
-\E \bigl(A\tilde a(Z)-1\bigr)\bb(Z)l(Z)\,\g(X).$$
Using the fact that $\E(A\given Z)=1/a(Z)$, where $a$ is bounded
away from zero, we can also write this as
$$\E \frac{\frac d{dt}_{|t=0}\tilde a_t(Z)}{\aa(Z)}
\frac{(\aa\bb)}a(Z)l(Z)=
-\E \frac{\bigl(A\tilde a(Z)-1\bigr) a(Z)\g(X)}{\aa(Z)}
\frac{(\aa\bb)(Z)}{a(Z)}l(Z).$$
Because the function $(\tilde a_t-\hat a)/\aa$ 
is contained in $L$ for every $t$ by construction, the
function $(d/dt)_{|t=0}\tilde a_t/\aa$ is also contained
in $L$. Combined with the validity of the preceding
display for every $l\in L$, we conclude
that  $(d/dt)_{|t=0}\tilde a_t(Z)/\aa(Z)$ is the
weighted projection of 
$-\bigl(A\tilde a(Z)-1\bigr) a(Z)\g(X)/\aa(Z)$ in $L_2(\Pr)$ 
onto the space $\{l(Z): l\in L\}$ relative to the weight $(\aa\bb/a)(Z)$.
The projection can be represented in terms of a
kernel operator (cf.\ Lemma~\ref{LemmaProjectionKernelExpressedInBasis}). 
If  $\Pi_p(z_1,z_2)(\aa\bb)(z_2)/a(z_2)$ denotes
the kernel, then
\begin{align*}
\frac{\frac d{dt}_{|t=0}\tilde a_t(z_1)}{\aa(z_1)}
&=-\E \Pi_p(z_1,Z_2) 
\frac {\bigl(A_2\tilde a(Z_2)-1\bigr) a(Z_2)\g(X_2)}{\aa(Z_2)}
\Bigl(\frac{\aa\bb}a\Bigr)(Z_2)\\
&=-\E \Pi_p(z_1,Z_2) \bigl(A_2\tilde a(Z_2)-1\bigr)\bb(Z_2)\, \g(X_2).
\end{align*}
This represents the derivative on the left as an inner product
of the score function $\g$ with the function on the right
of the first equation of the lemma (evaluated at $X_2$).
Thus the first assertion of the lemma is proved.

The second assertion is proved similarly. Using that $\E(Y\given Z)=b(Z)$
and $\E(A\given Z)=1/a(Z)$, we start by  writing the equation 
(\ref{EqApproxb}) defining
$\tilde b$ as $\E \bigl(\tilde b(Z)-Y\bigr)/\bb(Z)\,(\aa\bb/a)(Z)l(Z)=0$,
for every $l\in L$. By the same arguments as before we conclude that
$(d/dt)_{|t=0}\tilde b_t(Z)$ is the weighted projection
of $\bigl(Y-\tilde b(Z)\bigr)\g(X)/\bb(Z)$ in $L_2(\Pr)$ onto 
the space $\{l(Z): l\in L\}$,
relative to the weight $(\aa\bb/a)(Z)$. 
\end{proof}

The first order influence function (\ref{EqMARIFApprox}) depends on $p$ only
through $\tilde a$ and $\tilde b$ and hence the chain rule and the preceding
lemma imply that a second order influence
function of $\tilde\chi$  is given by the degenerate part of
\begin{align}
\tilde\chi_p^{(2)}(X_1,X_2)
&=-\Pi_p(Z_1,Z_2) \Bigl[A_1\bigl(Y_1-\tilde b(Z_1)\bigr)\aa(Z_1)
\bigl(A_2\tilde a(Z_2)-1\bigr)\bb(Z_2)\nonumber\\
&\qqqquad+\bigl(A_1\tilde a(Z_1)-1\bigr)\bb(Z_1)
A_2\bigl(Y_2-\tilde b(Z_2)\bigr)\aa(Z_2)\Bigr].
\label{EqMARSecondOrderIFTruncated}
\end{align}
(Note that this function is symmetric in $(X_1,X_2)$; $\Pi_p$ is symmetric,
because it is an orthogonal projection kernel.) 
Actually, this function is already degenerate and hence is the second
order influence function of $\tilde \chi$.

\begin{lemma}
\label{LemmaMARMakingDegenerate}
For any fixed $z_1$ and $z_3$,
\begin{enumerate}
\item[(i)] $\E_p \Pi_p(z_1,Z_2)\bigl(A_2\tilde a(Z_2)-1\bigr)\bb(Z_2)=0$.
\item[(ii)] $\E_p \Pi_p(z_1,Z_2)A_2\bigl(Y_2-\tilde b(Z_2)\bigr)\aa(Z_2)=0$.
\item[(iii)] $\E_p \Pi_p(z_1,Z_2)A_2(\aa\bb)(Z_2)\Pi_p(Z_2,z_3)=\Pi_p(z_1,z_3)$.
\end{enumerate}
\end{lemma}

\begin{proof}
Because $(\tilde a-\hat a)(Z)/\aa(Z)$ 
and $(\tilde b-\hat b)(Z)/\aa(Z)$ are the weighted
projections in $L_2(\Pr)$ of $(a-\hat a)(Z)/\aa(Z)$ 
and  $\bigl(Y-\hat b(Z)\bigr)/\bb(Z)$, respectively, onto 
$\{l(Z): l\in L\}$ relative to the weights $(\aa\bb/a)(Z)$,
\begin{align}
\E_{X_2} \Pi_p(Z_1,Z_2)\Bigl[\frac{\tilde a(Z_2)-\hat a(Z_2)}{\aa(Z_2)}
-\frac{a(Z_2)-\hat a(Z_2)}{\aa(Z_2)}\Bigr]\frac{\aa\bb}a(Z_2)&=0,\\
\E_{X_2} \Pi_p(Z_1,Z_2)\Bigl[\frac{\tilde b(Z_2)-\hat b(Z_2)}{\bb(Z_2)}
-\frac{Y_2-\hat b(Z_2)}{\bb(Z_2)}\Bigr](\aa\bb)(Z_2)A&=0.
\end{align}
These two assertions imply (i) and (ii).
The third assertion follows from the fact that
$\Pi_p$ is the kernel of the weighted projection in
$L_2(\Pr)$ onto $L$ relative to the weight $(\aa\bb/a)(Z)$.
\end{proof}

The second order influence function 
(\ref{EqMARSecondOrderIFTruncated}) depends on $p$
through $\tilde a$ and $\tilde b$ and through the kernel $\Pi_p$.
We proceed to higher orders by differentiating
the influence function relative to these components, and applying the
chain rule, where we use the influence functions of 
$p\mapsto\tilde a(x)$ and $p\mapsto\tilde b(x)$ as given
previously in Lemma~\ref{LemmaMARIFofaaAndbb},
and the influence function of $p\mapsto \Pi_p(z_1,z_2)$
as given in Lemma~\ref{LemmaIFProjection}.

\begin{proof}[Proof of Theorem~\ref{TheoremMARHigherOrderIF}]
Denote the symmetrization of the variable in the theorem
by $\bar \chi^{(m)}_p(X_1,\ldots, X_m)$. Then 
$\bar\chi_p^{(2)}$ is the function $\tilde\chi_p^{(2)}$  given by
(\ref{EqMARSecondOrderIFTruncated}), which was  seen to be a second
order influence function in the preceding discussion. We show
by induction on $m$ that 
$x_{m+1}\mapsto\bar\chi_p^{(m+1)}(x_1,\ldots, x_m,x_{m+1})$ is
an influence function of $p\mapsto\bar\chi_p^{(m)}(x_1,\ldots, x_m)$.
The theorem is then a corollary of Lemma~\ref{LemmaRecursiveIF}.

By Lemmas~\ref{LemmaMARIFofaaAndbb} and~\ref{LemmaIFProjection},
\begin{itemize}
\item[(i)] The influence function of $p\mapsto \tilde Y_1$
is $x_{m+1}\mapsto -\Pi_p(Z_1,z_{m+1})\AA_1\tilde y_{m+1}$
\item[(ii)] The influence function of $p\mapsto \tilde A_1$
is $x_{m+1}\mapsto -\Pi_p(Z_1,z_{m+1})\AA_1\tilde a_{m+1}$.
\item[(iii)] The influence function of $p\mapsto \AA_1$ is
zero.
\item[(iv)] The influence function of $p\mapsto \Pi_p(Z_1,Z_2)$
is $x_{m+1}\mapsto -\Pi_p(Z_1,z_{m+1})\AA_{m+1}\Pi_p(z_{m+1},Z_2)$.
\end{itemize}
Applying this repeatedly readily gives an expression
for the influence function of 
$p\mapsto \tilde A_1 \Pi_{1,2}\AA_2 \Pi_{2,3}\AA_3\Pi_{3,4}\AA_4
\times\cdots\times \AA_{m-1}\Pi_{m-1,m}\tilde Y_m$.
The symmetrization of this expression is 
the same expression, but then with $m$ replaced by $m+1$ and an added
minus sign.
\end{proof}

\section{Other examples}
\label{SectionOtherExamples}
In this section we briefly indicate a number of other examples for which our
general heuristics have been worked out, leading to well known or novel estimators.

\subsection{Density estimation}
\label{SectionDensityEstimation}
Consider estimating a density $\chi(p)=p(a)$ 
at the fixed point $a$ based on a random sample
from $p$. A first order influence function of this functional
would satisfy, for every smooth path $t\mapsto p_t$ with score function $g$ at $t=0$,
$$\int \chi^{(1)}_p gp\,d\m=\frac d{dt}_{|t=0} \chi(p_t)=g(a)p(a).$$
In a nonparametric situation every zero-mean function $g$
arises as a score function, and hence $\chi_p^{(1)}$ would have to
be a ``Dirac function at $a$''. Because this does not exist
(except for very special $p$), in this example already a first order
influence function fails to exist.

We may approximate the Dirac function by the function $x\mapsto \Pi(a,x)$
for $\Pi$ the kernel of an orthogonal projection onto a given
(large) subspace $L$ of $L_2(\m)$. Because
$\int \Pi(a,x) g(x)p(x)\,d\m(x)=g(a)p(a)$ for every function
$g$ such that $gp\in L$, the function $x\mapsto \Pi(a,x)$ achieves
representation for a large set of scores. The corresponding
degenerate version is $x\mapsto \Pi(a,x)-\Pi p(a)$, 
for $\Pi p=\int \Pi(\cdot,x)p(x)\,d\m(x)$ the projection
of $p$. The corresponding first order estimator (\ref{EqEstimator}) is
$$\hat \chi_n=\chi(\hat p_n)+\PP_n \bigl(\Pi(a,\cdot)-\Pi \hat p_n(a)\bigr)
=\PP_n \Pi(a, \cdot)+ \bigl((I-\Pi)\hat p_n\bigr) (a).$$
If $\hat p_n\in L$, then the second term vanishes and the estimator
reduces to $\PP_n \Pi(a,\cdot)$. This is the usual projection
estimator (cf.\ \cite{prakasarao,tsybakov}): 
if $L$ is spanned by the orthonormal 
set $e_1,e_2,\ldots,e_k$, then $\Pi(x_1,x_2)=\sumik e_i(x_1)e_i(x_2)$
and $\hat \chi_n=\sumik (\PP_n e_i) e_i(a)$.

Alternative to viewing $x\mapsto \Pi(a,x)$ as an approximation
to the ``ideal'' influence function, we can derive it as the exact influence function of the approximate
functional $\tilde\chi(p)=\chi(\Pi p)$.
\beginskip
Indeed, for every sufficiently regular path $t\mapsto p_t$
with score function $g$,
$$\frac d{dt}_{|t=0} \tilde\chi(p_t)
=\frac d{dt}_{|t=0} \int \Pi(a,x)p_t(x)\,d\m(x)
=\int \Pi(a,x) g(x)p(x)\,d\m(x).$$
This shows that $x\mapsto \Pi(a,x)$
is an influence function of $\tilde\chi$.
\endskip

\subsection{Quadratic functionals}
\label{SectionQuadratic}
Consider estimating the functional $\chi(p)=\int p^2\,d\m$ based on
a random sample of size $n$ from the density $p$.

The first order influence function of this functional exists on the full nonparametric
model, and can be seen to take the form
$$\chi_p^{(1)}(x)= 2\bigl(p(x)-\chi(p)\bigr).$$
\beginskip
To see this, it suffices to note that 
this function is mean-zero (i.e.\ degenerate) and
satisfies 
$$\frac d{dt}_{|t=0}\chi(p_t)=\int 2p_t\dot p_t\,d\m_{|t=0}
=P 2p g=P\chi_p^{(1)}g,$$
for any sufficiently regular path $t\mapsto p_t$ with
$p_0=p$ and score function $g=\dot p_0/p_0$ at $t=0$.
\endskip
By the algorithm [1]--[3] of Section~\ref{SectionComputingIF},
a second order influence function can be computed as the degenerate part of an influence
function of the functional $p\mapsto\bar\chi_p^{(1)}(x_1)=2p(x_1)$, for fixed $x_1$.
As seen in Section~\ref{SectionDensityEstimation}, point evaluation
is not a differentiable functional, but has the kernel $\Pi$
of an orthogonal projection in $L_2(\m)$ as an \emph{approximate} influence function.
Thus an approximate second order influence function of the present functional, minus its projection onto the degenerate
functions, is  given by 
$$\tilde\chi_p^{(2)}(x_1,x_2)
=2\Pi(x_1,x_2)-2\Pi p(x_1)-2\Pi p(x_2)+2\int (\Pi p)^2\,d\m.$$
This may also be derived as an exact influence function of the approximate
functional $\tilde\chi(p)=\chi(\Pi p)$.

It can be checked that the estimator (\ref{EqEstimatorParametricRate}) for $m=2$,
given an initial estimator $\hat p_n$ that is contained in the range of $\Pi$,
reduces to $\hat\chi_n=\UU_n\Pi$, which is a well known estimator
(\cite{Laurent97}). 

\subsection{Doubly robust models}
The heuristics described in Section~\ref{SectionHeuristics} ought to be applicable
in a wide range of estimation problems, but the detailed treatment of the
missing data problem in Sections~\ref{SectionMissingData}--\ref{SectionMinimaxRate} shows
that their implementation can be involved. Inspection of the proofs reveals that
the particular implementation in the latter sections is based on the structure 
(\ref{EqMARFOIF}) of the first order influence function in the missing data problem. 
The argument extends to semiparametric models with first order influence function of the form
\begin{equation}
\label{EqJamiesModel}
\chi_p^{(1)}(x)=a(z)b(z) S_1(x)+a(z) S_2(x)+ b(z)
S_3(x)+S_4(x)-\chi(p),
\end{equation}
for known functions $S_i(x)$ of the data (i.e.\  $S=(S_1,S_2,S_3,S_4)$ is a given statistic).
The full parameter may be a quadruplet $p\leftrightarrow(a,b,c,f)$, in which
$f$ is the marginal density of an observable covariate $Z$, and $c$ does not appear in (\ref{EqJamiesModel}).
Other examples of this structure are described in \cite{RobinsetalFreedman,vdVStatScience}.

\section{Appendix 1: Influence functions}
\label{SectionInfluenceFunctions}
The main aim of this section is to prove the validity of
algorithm [1]--[3] as given in Section~\ref{SectionComputingIF}
for computing influence functions. We start with a lemma
that motivates the defining property of influence functions
in Section~\ref{SectionHeuristics}.

\begin{lemma}
\label{LemmaTaylorExpansion}
Let $f: [-1,1]\to \RR$ and $g: [-1,1]\times[-1,1]\to \RR$
be $k$ times continuously differentiable with $g(t,t)=0$ and
$({d^j}/{dt^j})f(t)=({\partial^j}/{\partial s^j})_{|s=t} g(s,t)$ for
every $t$ and $j=1,\ldots, k$. Then, for $j=1,\ldots, k$ and
every $u\in(-1,1)$,
$$\frac{\partial^j}{\partial t^j} g(u,t)_{|t=u}
=-\frac{\partial^j}{\partial s^j} g(s,u)_{|s=u}.$$
\end{lemma}

\begin{proof}
The conditions show that the functions
$s\mapsto f(s)$ and $s\mapsto g(s,t)$ have the same
first $k$ derivatives at $s=t$. Because also $g(t,t)=0$, it follows
that $f(s)-f(t)-g(s,t)=o\bigl(|s-t|^k\bigr)$ as $s\ra t$.
By writing the remainder term in the form
$$\frac1{k!}(s-t)^k \Bigl[\bigl[f^{(k)}\bigl(t+\x_s(s-t)\bigr)-f^{(k)}(t)
- g_1^{(k)}\bigl(t+\x_s'(s-t),t\bigr)+g_1^{(k)}(t,t)\Bigr],$$
for $f^{(k)}$ the $k$th derivative of $f$ and $g_1^{(k)}$ the
$k$th partial derivative of $g$ relative to its first
argument, we see that 
$f(s)-f(t)-g(s,t)=o\bigl(|s-t|^k\bigr)$ as $|s- t|\ra 0$, uniformly
in $(s,t)$, by the assumed (uniform) continuity of the $k$th derivatives.
Now the difference $f(s)-f(t)$ can also be expanded
as $-\bigl[f'(s)(t-s)+\cdots +f^{(k)}(s)(s-t)^k/k!\bigr] +
o\bigl(|s-t|^k\bigr)$ as $t\ra s$. A similar expansion
of $t\mapsto g(s,t)=g(s,t)-g(s,s)$ follows. The lemma
follows by uniqueness of a Taylor expansion.
\end{proof}

Let $\chi: (-1,1)\to \RR$ be $m$ times continuously
differentiable and let $t\mapsto p_t$ be a smooth
map from $(-1,1)$ to $\P$. 
Assume that $\bar\chi^{(j)}_t: \X^j\to \RR$ are symmetric functions 
such that, for $t\in (-1,1)$ and $j=1,\ldots, m-1$, and
for every $(x_1,\ldots, x_j)\in \X^j$,
\begin{align}
\label{EqRepOne}
\frac d{dt}\chi(t)&=\int \bar \chi_t^{(1)}(x){\frac{\partial}{\partial t}}
 p_t(x)\,d\m(x),\\
\label{EqRepTwo}
\frac {\partial}{\partial t}
\bar \chi_t^{(j)}(x_1,\ldots, x_j)
&=\int \bar\chi_t^{(j+1)}(x_1,\ldots,x_j,x)
{\frac{\partial}{\partial t}}p_t(x)\,d\m(x).
\end{align}

\begin{lemma}
\label{LemmaRecursiveIF}
Under (\ref{EqRepOne})-(\ref{EqRepTwo}) and regularity assumptions, 
the functions $\chi^{(j)}_0=D_p\bar \chi^{(j)}_0$
satisfy, for every $j=1,\ldots, m$,
$$\frac {d^j}{dt^j}_{|t=0}\chi(t)
=\frac {d^j}{dt^j}_{|t=0}\Bigl(P_t\chi_0^{(1)}+\frac 12P_t^2\chi_0^{(2)}
+\cdots+\frac 1{m!}P_t^m\chi_0^{(m)}\Bigr).$$
\end{lemma}

\begin{proof} 
By Leibniz's rule, for every $j$ and $i$,
$$\frac {d^j}{dt^j} P_t^i\chi_0^{(i)}
=\sum_{j_1,\cdots, j_i} \left({j\atop j_1\cdots j_i}\right)
\int\!\cdots\!\int\chi_0^{(i)} 
(p_t^{(j_1)}\times\!\cdots\!\times p_t^{(j_i)})\,d\m^i,$$
where $p_t^{(j)}$ is the $j$th partial derivative of $t\mapsto p_t$.
Upon evaluation at $t=0$ all terms in the sum with
one of the indices $j_1,\ldots, j_i$ equal to zero vanish,
by degeneracy of the function $t\mapsto \chi_0^{(i)}$.
This will happen for every $(j_1,\ldots, j_i)$ if $i>j$.
It follows that the right side of the lemma can be written as 
$$\sum_{i=1}^j \int\!\cdots\!\int\chi_0^{(i)} b_0^{(i,j)}\,d\m^i,$$
for the functions $b_t^{(i,j)}$ defined by
$$b_t^{(i,j)}=\frac 1{i!} 
\sum_{j_1,\cdots, j_i>0} \left({j\atop j_1\cdots j_i}\right)
p_t^{(j_1)}\times\cdots\times p_t^{(j_i)}.$$
We shall show that the left side of the lemma can be
written in the same form.

In fact, we prove by induction on $j$ that, for every $t$,
\begin{equation}
\label{EqHulpProofRecursive}
\frac {d^j}{dt^j}\chi(t)
=\sum_{i=1}^j \int\!\cdots\!\int\bar\chi_t^{(i)} b_t^{(i,j)}\,d\m^i.
\end{equation}
Because the functions $b_t^{(i,j)}$ are degenerate, the right is
unchanged if $\bar\chi_t^{(i)}$ is replaced by its degenerate
part (in $L_2(p_t)$), and hence the lemma follows.

Let a dot denote differentiation relative to $t$.
For $j=1$ the identity is exactly assumption (\ref{EqRepOne}),
because $b_t^{(1,1)}=\dot p_t$.
If assertion (\ref{EqHulpProofRecursive})
is true for $j$, then it follows by differentiation that,
\begin{align*}\frac {d^{j+1}}{dt^{j+1}}\bar\chi(t)
&=\sum_{i=1}^j \int\!\cdots\!\int\Bigl[\dot{\bar\chi}_t^{(i)}b_t^{(i,j)}\,d\m^i
+\bar\chi_t^{(i)} \dot b_t^{(i,j)}\Bigr]\,d\m^i\\
&=\sum_{i=1}^j \Bigl[\int\!\cdots\!\int\int\bar\chi_t^{(i+1)} 
(b_t^{(i,j)}\times \dot p_t)\,d\m^{i+1}
+\int\!\cdots\!\int\bar\chi_t^{(i)}\dot b_t^{(i,j)}\,d\m^i\Bigr],
\end{align*}
by (\ref{EqRepTwo}). Here the function
$b_t^{(i,j)}\times\dot p_t$ can be replaced by its symmetrization,
by the assumed symmetry of $\bar\chi_t^{(i+1)}$.
It follows that the assertion is true for $j+1$ if
the $b_t^{(i,j)}$ satisfy the recursion formulas,
with $S$ denoting symmetrization,
\begin{align*}
b_t^{i,j+1}&=S(b_t^{(i-1,j)}\times \dot p_t)+\dot b_t^{(i,j)},
\qquad 1<i<j+1,\\
b_t^{j+1,j+1}&=S(b_t^{(j,j)}\times\dot p_t),\\
b_t^{1,j+1}&=\dot b_t^{(1,j)}.
\end{align*}
The second and third recursions are consistent with the
first if we set $b_t^{(j+1,j)}=b_t^{(0,j)}=0$.

From  the definition of $b_t^{(i,j)}$ we see that (note that
$\dot p_t=p_t^{(1)}$)
\begin{align*}
&S(b_t^{(i-1,j)}\times \dot p_t)+\dot b_t^{(i,j)}\\
&\qquad=\frac 1{(i-1)!}
\sum_{j_1,\cdots, j_{i-1}>0} 
\left({j\atop j_1\cdots j_{i-1}}\right)
S\bigl(p_t^{(j_1)}\times\cdots\times p_t^{(j_{i-1})}\times p_t^{(1)} \bigr)\\
&\qqqquad +
\frac 1{i!} 
\sum_{j_1,\cdots, j_i>0} \left({j\atop j_1\cdots j_i}\right)
\sum_{l=1}^i p_t^{(j_1)}\times\cdots\times p_t^{(j_l+1)}\times
\cdots\times p_t^{(j_i)}.
\end{align*}
This can be seen to be equal to $b_t^{(i,j+1)}$. Indeed,
the sum defining the latter function corresponds to
the assignments of $j+1$ objects to $i$ nonempty
boxes. The two sums in the preceding display correspond to
the assignments in which the $(j+1)$th object is alone
in a box ($i$ possible boxes, the other $j$ objects
distributed over $i-1$ boxes in groups of sizes $j_1,\ldots, j_{i-1}$)
or is in a box with at least one other object ($i$ possible
boxes, the other objects distributed over the $i$ boxes in groups of
sizes $j_1,\ldots, j_i$). Note that the symmetrization
$S\bigl(p_t^{(j_1)}\times\cdots\times p_t^{(j_{i-1})}\times p_t^{(1)}\bigr)$
can be written as an average over the $i$ expressions obtained 
by placing the term $p_t^{(1)}$ before the $l$th factor
of the product $\prod_{l=1}^{i-1}p_t^{(j_l)}$, or after the
$(i-1)$th factor.
\end{proof}

%Is it restrictive to assume $\bar\chi_t^{(j)}$ to be symmetric??
%Should be automatic. Make remark??

\section{Appendix 2: Projections}
\label{SectionProjections}
In this section we collect essential properties of projections,
including representation by kernels, means and variances, and influence functions.
Throughout let $\m$ be a $\s$-finite measure onto some arbitrary
measurable space.

\subsection{Generalities}
We call a \emph{weighted projection} in $L_2(\m)$
onto a closed subspace $L$ with weight function $w$ the map 
$\Pi: L_2(\m)\to L$ given by
$$\Pi g=\argmin_{l\in L}\int (g-l)^2\,w\,d\m.$$
We assume that the weight function $w$ is bounded away from 0 and $\infty$,
so that this map is well defined.
The weighted projection is determined by:  $\Pi g\in L$ and the
orthogonality relationship
$$\int (g-\Pi g)l\,w\,d\m=0,\qquad \forall l\in L.$$
We say that the weighted projection has a \emph{kernel representation
with kernel} $\Pi$ if, for all $g\in L_2(\m)$,
$$\Pi g(x_1)=\int \Pi(x_1,x_2)g(x_2)w(x_2)\,d\m(x_2).$$
A weighted projection is of course just an orthogonal projection
onto $L$ in the space $L_2(\n)$ for the measure $\n$ defined
by $d\n=w\,d\m$, and as a kernel operator on $L_2(\n)$ it has
precisely kernel $\Pi$. On the other hand, as a kernel operator on
$L_2(\m)$ the weighted projection 
has kernel $(x_1,x_2)\mapsto \Pi(x_1,x_2)w(x_2)$,
which includes the weight function. This ambiguity is unavoidable
in our context, as we need to work with multiple weight functions,
both estimated and ``true'' ones.

The kernel of an orthogonal projection is symmetric in its
arguments.
%\footnote{A projection $\Pi$ is orthogonal in $L_2(\n)$ if and only if $\langle \Pi g,(I-\Pi)h\rangle_\n=0$ for every $g,h$, or equivalently $\langle \Pi g, h\rangle_\n=\langle \Pi^*\Pi g, \Pi^*\Pi h\rangle_\n$, which is true if and only if $\Pi^*\Pi=\Pi$, which is true if and only if $\Pi$ is self-adjoint.}
Thus with the preceding definition the ``kernel of a weighted projection'' is also symmetric.

Not all projections have kernels, but projections on
finite-dimensional spaces do.

\begin{lemma}
\label{LemmaProjectionKernelExpressedInBasis}
If $e_1,\ldots, e_k$ are arbitrary linearly independent elements
that span the linear subspace $L$ of $L_2(\m)$, then the
weighted projection onto $L$ relative to the weight function $w$ has kernel
$$\Pi(x_1,x_2)=\sum_i\sum_j (C^{-1})_{ij}e_i(x_1)e_j(x_2),$$
for $C$ the $(k\times k)$-matrix with $(i,j)$th element
$C_{ij}=\int e_ie_jw\,d\m$.
\end{lemma}

\begin{proof}
Because we can change measure from $\m$ to $\n$ given by $d\n=w\,d\m$,
if suffices to prove the lemma for the case that $w=1$.
If $\Pi g=\sum_i \g_i e_i$, then the orthogonality relationships 
$g-\Pi g\perp e_j$ give that
$\sum_i\g_i C_{ij}= \int g e_j\,d\n$ for $j=1,\ldots, k$.
We can invert this system of linear equations to see
that $\g_i=\sum_j (C^{-1})_{ij}\int ge_j\,d\n$ for every $i$.
Insert this into $\Pi g=\sum_i \g_ie_i$ and exchange the order of
summation and integration to obtain the result.
\end{proof}

We view projections mainly as operators on $L_2(\m)$, but for
a number of arguments we need control of approximation errors
in $L_s(\m)$ for $s>2$. An $L_2$-projection $\Pi$ does not necessarily
give a best approximation in $L_s(\m)$ for $s\not=2$, but it often
gives an approximation that is optimal up to a constant. 
This is the case if its norm  as an operator 
$\Pi: L_s(\m)\to L_s(\m)$ is finite. (Finiteness assumes implicitly that 
$\Pi$ maps $L_s(\m)$ in itself; the norm
$\|\Pi\|_s$ is then by definition the minimal number
$C$ such that $\|\Pi g\|_s\le C\|g\|_s$ for every $g\in L_2(\m)$.)

\begin{lemma}
\label{LemmaProjectionInLs}
Let $\Pi$ be an orthonormal projection in $L_2(\m)$ onto
a subspace $L$ that is also contained in $L_s(\m)$. If 
$\Pi: L_s(\m)\to L_s(\m)$ has bounded norm $\|\Pi\|_s$, then
$$\|g-\Pi g\|_s\le \bigl(1+\|\Pi\|_s\bigr)\|g-L\|_s.$$
\end{lemma}

\begin{proof}
The triangle inequality gives $\|g-\Pi g\|_s\le \|g-l\|_s+\|l-\Pi g\|_s$. Since
$l=\Pi l$, for every $l\in L$, the second term is bounded by $\|\Pi\|_s\|l- g\|_s$. 
We finish by taking the infimum over $l\in L$.
\end{proof}

One example are projections on a wavelet basis. 
The $L_s$-norm of a function is equivalent to the $\ell_s$-norm 
of the coefficients relative to such a basis (suitably normalized).
Because the $L_2$-projection is the wavelet expansion truncated 
at a certain level of resolution, projection
decreases the $\ell_s$-norm of the coefficients and hence
the $L_s$-norm of the function ``up to a constant''.

\subsection{Norms, means and variances}
An orthogonal projection in $L_2(\m)$ has operator norm $1$,
but the square $L_2(\m\times \m)$-norm 
$\int\int \Pi^2\,d(\m\times \m)$ of its kernel is equal to the dimension
of its projection space. 

\begin{lemma}
\label{LemmaNormProjection}
The kernel of an orthogonal projection onto a $k$-dimensional subspace
of $L_2(\m)$ has square $L_2(\m\times\m)$-norm 
$\int\int \Pi^2\,d(\m\times \m)=k$.
\end{lemma}

\begin{proof}
By writing the kernel in the form given by
Lemma~\ref{LemmaProjectionKernelExpressedInBasis} relative
to an orthonormal basis of the projection space (so that $C=I$), we find that
$$\int\int \Pi^2\,d(\m\times \m)
=\sum_i\sum_j\int\int e_i(x_1)e_i(x_2)e_j(x_1)e_j(x_2)\,d\m(x_1)\,d\m(x_2).$$
The off-diagonal elements vanish by
orthogonality, while the diagonal elements are equal to 1.
\end{proof}

Typically the square norm of a projection
 kernel can be written as $\int \Pi(x,x)\,d\m(x)$.
In fact, the projection property $\Pi^2=\Pi$ of a kernel operator on $L_2(\m)$
can be expressed in the kernel as
\begin{equation}
\label{EqProjectionKernel}
\int \Pi(x_1,x_2)\Pi(x_2,x_3)\,d\m(x_2) = \Pi(x_1,x_3),\qquad \text{a.e.}\ (x_1,x_3).
\end{equation}
If this equation holds for every $x_1=x_3$ and $\Pi$ is symmetric, then
we obtain by integration that
$\int\int \Pi^2\,d(\m\times \m)=\int \Pi(x,x)\,d\m(x)$.

For simplicity of notation we assume that the kernel 
is such that (\ref{EqProjectionKernel}) is
valid for every $x_1,x_3$, in particular on the diagonal 
$\{(x_1,x_3): x_1=x_3\}$. (This  is typically a null set, making this an assumption
of using a special representative.)
This is true in particular for the kernels
in Lemma~\ref{LemmaProjectionKernelExpressedInBasis}.

\begin{lemma}
\label{LemmaVarianceProductOfKernels}
If $\Pi_1,\ldots, \Pi_{m-1}$ are kernels of orthogonal projections in $L_2(\m)$
that satisfy (\ref{EqProjectionKernel}) identically, then, for any $j\in \{1,\ldots, m-1\}$,
$$\int\cdots\int \prod_{i=1}^{m-1} \Pi_i^2(x_i,x_{i+1})\,d\m(x_1)\cdots d\m(x_m)
\le \prod_{i=1: i\not= j}^{m-1}\sup_x \Pi_i(x,x)\! \int\! \Pi_j^2\,d(\m\times\m).$$
\end{lemma}

\begin{proof}
Equation (\ref{EqProjectionKernel}) implies that 
$\int \Pi_i(x,y)^2\,d\m(y)=\Pi_i(x,x)$, for every $x$.
If $j<m-1$, then we apply this to the integral with respect to $x_m$
of the multiple integral in the lemma, thereby turning this $m$fold
integral into an $(m-1)$fold integral of the function
$\prod_{i=1}^{m-2} \Pi_i^2(x_i,x_{i+1})\Pi_{m-1}(x_{m-1},x_{m-1})$. Next
we bound the factor $\Pi_{m-1}(x_{m-1},x_{m-1})$ by its supremum 
over $x_{m-1}$, and are left with an $(m-1)$fold integral
of the same type as before times this supremum.
We repeat the argument, removing all kernels to the right of the $j$th kernel.
Next we apply the same procedure working from the left side up, until the only remaining integral
is $\int \Pi_j(x_j,x_{j+1})\,d\m(x_j)d\m(x_{j+1})$.
\end{proof}

The preceding results show that (under (\ref{EqProjectionKernel}))
the square norms of (products of) projection kernels are 
controlled by their values on the diagonal. The following lemma
shows that these values do not differ significantly for weighted
projections with different weights.

\begin{lemma}
\label{LemmaProportionalityWeightedKernels}
The weighted projections in $L_2(\m)$ onto a finite-dimensional 
space $L$ relative to the weight functions $v$ and $w$ 
possess kernels $\Pi_v$ and $\Pi_w$ that satisfy
(\ref{EqProjectionKernel}) identically and, for every $x$,
$$\Pi_v(x,x)\le \Bigl\|\frac w v\Bigr\|_\infty \Pi_w(x,x).$$
\end{lemma}

\begin{proof}
For a fixed basis $e_1,\ldots, e_k$ of $L$ we can,
by Lemma~\ref{LemmaProjectionKernelExpressedInBasis},
represent the kernels as 
$\Pi_v(x,y)=\vec e_k(x)^TC_v^{-1}\vec e_k(y)$ for $C_v$
the matrix with $ij$th element $\int e_ie_jv\,d\m$, and similarly
for $\Pi_w$. By choosing $e_1,\ldots, e_k$ to be orthonormal
in $L_2(w)$ the matrix $C_w$ can be reduced to the identity.
The quotient $\Pi_v(x,x)/\Pi_w(x,x)$ then takes the
form $z^TC_v^{-1}z/ z^Tz$ for some $z\in \RR^k$, and it suffices
to upper bound this quotient uniformly in $z\in\RR^k$. The supremum of
this quotient over $z$ is the maximal eigenvalue of $C_v^{-1}$,
which is the inverse of the minimal eigenvalue of $C_v$.
Because 
$$z^T C_v z=\int \Bigl(\sumik z_ie_k\Bigr)^2 v\,d\m\ge \inf_x \frac vw(x)
\int \Bigl(\sumik z_ie_k\Bigr)^2 w\,d\n= \inf_x \frac vw(x) \,z^Tz,$$
this minimum eigenvalue is bigger than the minimum value of $v/w$.
\end{proof}

\begin{lemma}
\label{LemmaMeanProductOfKernels}
If $\Pi_1,\ldots, \Pi_{m-1}$ are kernels of integral operators on $L_s(\m)$ with norms $\sup_{r/(r-1)\le s\le r}\|\Pi_i\|_s\le C$,
then for arbitrary measurable functions $w_1,\ldots,w_m$ and any $r\ge 2$ (with $r/(r-2)=\infty$ if $r=2$),
\begin{align*}
&\Bigl|\int\cdots\int \prod_{i=1}^{m-1} \Pi_i(x_i,x_{i+1})\prodim w_i(x_i)\,d\m(x_1)\cdots d\m(x_m)\Bigr|\\
&\qqqquad\le C^{m-3}\|\Pi_1w_1\|_r\,\|\Pi_{m-1}w_m\|_r\prod_{i=2}^{m-1} \|w_i\|_{(m-2)r/(r-2)}.
\end{align*}
\end{lemma}

\begin{proof}
Let $M_i$ denote multiplication by $w_i$, i.e.\ $M_ig=w_ig$.
By H\"older's inequality  the left side is smaller than, for any conjugate pairs $(p_i,q_i)$,
\begin{align*}
&\|\Pi_1w_1\|_{p_1} \,\bigl\|M_2\Pi_2M_3\Pi_3M_4\cdots \Pi_{m-2}M_{m-1}\Pi_{m-1}w_m\bigr\|_{q_1}\\
&\qquad\le \|\Pi_1w_1\|_{p_1}
\,\|w_2\|_{q_1p_2}\|\Pi_2\|_{q_1q_2}\|w_3\|_{q_1q_2p_3}\times\cdots\\
&\qqqquad\qqqquad\times\|w_{m-1}\|_{q_1\cdots
  q_{m-2}p_{m-1}}\|\Pi_{m-1}w_m\|_{q_1\cdots q_{m-1}}.
\end{align*}
We finish by choosing the conjugate pairs so that $p_1=q_1q_2\cdots q_{m-1}=r$ and
$q_1p_2=q_1q_2p_3=\cdots=q_1\cdots q_{m-2}p_{m-1}$. 
Then the common value in the last string is $(m-2)r/(r-2)$ and
the indices of the operator norms satisfy $r/(r-1)=q_1\le q_1q_2\le\cdots\le q_1\cdots q_{m-1}=r$.
\end{proof}

\subsection{Approximations of weighted projections}

\begin{lemma}
\label{LemmaDifferenceOfProjections}
Let $\Pi_w$ and $\Pi$ be the weighted projections onto a fixed subspace 
$L$ of $L_2(\m)$ relative to the weight functions $w$ and $1$, 
respectively, and let
$M_w$ be multiplication by the function $w$. Then,
for any conjugate pairs $r^{-1}+s^{-1}=1$ and $p^{-1}+q^{-1}=1$,
any $t\le r$, any integer $m\ge 2$, and any $g$,
\begin{itemize}
\item[(i)] $\bigl\|(\Pi_w-\Pi M_w)g\bigr\|_r\le\|\Pi\|_s\|\Pi_wg\|_{rq} \|w-1\|_{rp}$.
\item[(ii)] $\bigl\|(\Pi_w-\Pi)g\bigr\|_r\le \|\Pi\|_s\bigl\|(I-\Pi_w)g\bigr\|_{rq}  \|w-1\|_{rp}$.
\item[(iii)]$\bigl\|(\Pi_w-\Pi M_w)^mg\bigr\|_t
\le C^{m-1}\|\Pi\|_s \|\Pi_wg\|_{rq}\|w-1\|_{rp}^m$, 
where $p=(m-1)t/(r-t)$ (with $p=\infty$ if $r=t$) and the constant $C$
is the supremum of the norms of the operator
$\Pi: L_u(\m)\to L_u(\m)$ over $u\in [t,r)$.
\item[(iv)] $\|\Pi M_wg\|_r\le \|\Pi\|_s\|\Pi_wg\|_r\bigl(2+\|w\|_\infty\bigr)$.
\end{itemize}
\end{lemma}

\begin{proof}
(i). The orthogonality relationships for the projections $\Pi$ and
$\Pi_w$ imply that $\int \Pi(wg) l\,d\m=\int wgl\,d\m=\int w(\Pi_wg)l\,d\m$,
for every $l\in L$ and $g$. Because $\Pi_wg-\Pi(wg)$ is contained in $L$,
it follows that, for every $k\in L_s(\m)\cap L_2(\m)$,
\begin{align*}
\int \bigl(\Pi_wg-\Pi(wg)\bigr)&k\,d\m
=\int \bigl(\Pi_wg-\Pi(wg)\bigr)\Pi k\,d\m,\\
&=\int \Pi_wg (1-w)\Pi k\,d\m
\le \|\Pi_w g\|_{qr}\|1-w\|_{pr}\|\Pi\|_s\|k\|_s,
\end{align*}
by H\"older's inequality. By approximating a general element $k\in L_s(\m)$
by a sequence in $L_s(\m)\cap L_2(\m)$
(truncate $k$ by a constant and restrict it to sets of 
finite $\m$-measure that increase to the whole space), it is seen
that the far left side is bounded by the far right side of
the display for any $k\in L_s(\m)$. Assertion (i) follows,
because the norm $\|(\Pi_w-\Pi M_w)g)\|_r$
is the supremum of the left side over all $k\in L_s(\m)$ with
$\|k\|_s\le 1$. 

(ii). Because the function $(\Pi_w-\Pi)g$ is contained in $L$ for any
fixed $g$, the orthogonality relationships for $\Pi$ and $\Pi_w$ imply,
for any function $k$ as under (i),
\begin{align*}
\int (\Pi_w-\Pi)g\, k\,d\m&=\int (\Pi_w-\Pi)g\, \Pi k\,d\m=\int (\Pi_w-I)g\, \Pi k\,d\m\\
&=\int (\Pi_w-I)g\, \Pi k\,(1-w)\,d\m\\
&\le \|\Pi_wg-g\|_{qr}\|1-w\|_{pr}\|\Pi\|_s\|k\|_s,
\end{align*}
by H\"older's inequality. We take the supremum over $k$ to finish the proof.

(iii). The operator $\Pi_w-\Pi$ vanishes on $L$, so that
$\Pi_w-\Pi M_w=\Pi-\Pi M_w=\Pi M_{1-w}$ on this space. 
Therefore, for $m\ge 2$ and any $t'\ge t$,
\begin{align*}
\bigl\|(\Pi_w-\Pi M_w)^mg\bigr\|_t
&=\bigl\|\Pi M_{1-w}(\Pi_w-\Pi M_w)^{m-1}g\bigr\|_t\\
&\le \|\Pi\|_t\|M_{1-w}\|_{t'\ra t}
\bigr\|(\Pi_w-\Pi M_w)^{m-1}g\bigr\|_{t'}.
\end{align*}
Here $\|A\|_{r\to s}$ denotes the norm of an operator
$A: L_r(\m)\to L_s(\m)$, and $\|A\|_r=\|A\|_{r\to r}$.
Using H\"older's inequality, we see that
the norm $\|M_{1-w}\|_{t'\ra t}$ is 
bounded above by $\|1-w\|_{tt'/(t'-t)}$, for $t'>t$.

We repeat this argument $m-1$ times
with the pairs $(t,t')$ equal to $(t_i,t_{i+1})$ for
a sequence $t=t_1<t_2<\cdots<t_m=r$ such that
$1/t_{i-1}-1/t_i=(t_i-t_{i-1})/(t_{i-1}t_i)=1/(rp)$ for every $i$.
(We divide $[1/r,1/t]$ in $m-1$ equal intervals of
length $1/(rp)$.) This results in 
$$\bigl\|(\Pi_w-\Pi M_w)^mg\bigr\|_t
\le \prod_{i=1}^{m-1}\bigl(\|\Pi\|_{t_i}\|w-1\|_{rp}\bigr)
\bigr\|(\Pi_w-\Pi M_w)g\bigr\|_r.$$
Finally we apply (i) to the last term.

(iv). This is a consequence of (i) with $p=\infty$ and the triangle inequality.
\end{proof}

\subsection{Influence functions}
Let $L$ be a fixed linear space of functions
contained and closed in $\L_2(p)$ for every $p$ in a collection
$\P$ of densities relative to a fixed measure $\n$.

\begin{lemma}
\label{LemmaIFProjection}
Let $\Pi_p$ be the kernel of a weighted projection operator in $L_2(p)$
onto a finite-dimensional subspace. If the subspace and weight
function $w$ are independent of $p$, then for almost every $(x_1,x_2)$
the map $x_3\mapsto -\Pi_p(x_1,x_3) w(x_3)\Pi_p(x_3,x_2)$ is
an influence function of the functional $p\mapsto \Pi_p(x_1,x_2)$.
\end{lemma}

\begin{proof}
The projection property gives that $\Pi_p l=l$ for every $l\in L$,
which  can be written as, for almost every $x_2$, with
$d\n=w d\m$,
$$\int \Pi_p(x_2,x_3)l(x_3)p(x_3)\,d\n(x_3)=l(x_2).$$
Substitute a smooth path $t\mapsto p_t$ and differentiate at $t=0$
to conclude, with $\g=\dot p_0/p_0$ the score function
of the path and $p=p_0$,
\begin{align*}
&\int \frac d{dt}_{|t=0}\Pi_{p_t}(x_2,x_3)l(x_3)p(x_3)\,d\n(x_3)\\
&\qquad=-\int \Pi_p(x_2,x_3)l(x_3)\g(x_3)p(x_3)\,d\n(x_3).
\end{align*}
Because $\Pi_p$ projects onto the same space $L$ for all $p$,
the function $x_1\mapsto \Pi_p(x_1,x_2)$ is contained in $L$
for every $x_2$, if we use the kernel given in 
Lemma~\ref{LemmaProjectionKernelExpressedInBasis}, 
and hence also the function
$x_1\mapsto (d/dt)_{|t=0}\Pi_{p_t}(x_1,x_2)$. Therefore,
\begin{align*}
\frac d{dt}_{|t=0}\Pi_{p_t}(x_1,x_2)
&=\Pi_p\Bigl(\frac d{dt}_{|t=0}\Pi_{p_t}(\cdot, x_2)\Bigr)(x_1)\\
&=\int \frac d{dt}_{|t=0}\Pi_{p_t}(x_3,x_2)\Pi_p(x_1,x_3)p(x_3)\,d\n(x_3).
\end{align*}
Applying the second last  display with $l$ equal to the
function defined by $l(x_3)=\Pi_p(x_1,x_3)$ yields that
the right side is equal to
$-\E_{X_3} \Pi_p(x_2,X_3) \Pi_p(x_1,X_3)\g(X_3)w(X_3)$.
\end{proof}

\subsection{Wavelets}
An orthonormal wavelet basis of $L_2(\RR^d)$ is given in terms of functions $\psi_{i,j}^v$ 
indexed by a ``resolution'' (or scale) parameter 
$i\in \NN$, a ``location'' parameter $j\in \ZZ^d$, and
a ``dimension index'' $v\in\{0,1\}^d$ (e.g.\ \cite{Hardleetal, Daubechies, Cohen}). Each function 
$\psi_{i,j}^v$ is a scaled and translated version of a fixed 
base function $\psi_{0,0}^v$ through
$$\psi_{i,j}^v(z)=2^{id/2}\psi_{0,0}^v( 2^iz-j),\qquad i\in\NN,
j\in \ZZ^d, v\in\{0,1\}^d.$$ 
The multiresolution property of wavelets entails that for each resolution
level $I$ we can expand a function $g\in L_2(\RR^d)$ as
$$g=\sum_{j\in \ZZ^d}\sum_{v\in\{0,1\}^d}
\langle g,\psi_{I,j}^v\rangle \psi_{I,j}^v
+\sum_{i>I}\sum_{j\in \ZZ^d}\sum_{v\in\{0,1\}^d-\{0\}}
\langle g,\psi_{i,j}^v\rangle \psi_{i,j}^v.$$
Thus the functions $\psi_{i,j}^v$ with $i\ge I$, $j\in \ZZ^d$ and
$v\in\{0,1\}^d$, with $v\not=0$ if $i>I$, span $L_2(\RR^d)$. 
We consider for each $I$ the projection obtained by leaving out the
contributions of the base functions at resolution levels $i>I$,
and retaining only the projections on the functions $\psi_{I,j}^v$
with $j\in\ZZ^d$ and $v\in \{0,1\}^d$. 

\begin{lemma}
\label{LemmaWaveletsOne}
If the generating base functions $\psi_{0,0}^v$ are bounded and
compactly supported, then the kernel of the 
orthogonal projection in $L_2(\RR^d)$ 
onto the linear span of all base functions 
$\{\psi_{I,j}^v: i\in J_I,v\in\{0,1\}^d\}$ 
whose support intersects $[0,1]^d$ satisfies $\Pi(x,x)\le C 2^{Id}$
for some constant $C$.
\end{lemma}

\begin{proof}
The kernel can be written in the form
$$\Pi(z_1,z_2) =\sum_{j\in J_I}\sum_{v\in \{0,1\}^d}
\psi_{I,j}^{v}(z_1)\psi_{I,j}^{v}( z_2).$$ 
Here $J_I$ includes all $j\in \ZZ^d$ such that
the support of $\psi_{I,j}$ intersects $[0,1]^d$. For
a fixed vector $(z,z)$ the function 
$\psi_{I,j}^{v}(z)\psi_{I,j}^{v}( z)$ is nonzero only if
$2^Iz-j$ is contained in the support of the function
$\psi_{0,0}^v$. The number of vectors $j\in J_I$ such that this is 
the case is  bounded by a constant that depends only on the
support of $\psi_{0,0}^v$. For each
$j$ the product $\psi_{I,j}^{v}(z)\psi_{I,j}^{v}( z)$ 
is bounded by $2^{Id}$ times $\|\psi_{0,0}^v\|_\infty^2$. 
The lemma follows.
\end{proof}

It follows by Lemma~\ref{LemmaProportionalityWeightedKernels}
that the kernel of the projection onto the wavelet
bases viewed as subset of $L_2(\n)$ is similarly bounded,
for any measure $\n$ with a Lebesgue density that is bounded away from zero and infinity.

In Section~\ref{SectionMinimaxRate} the projection has been
decomposed as a sum of projections on subspaces. Within the
context of wavelet bases it is natural to choose the
blocks in this decomposition equal to unions of resolution
levels, so that all base functions at a given refinement level are
included in the same block. To this end we choose
the grids $n=k_0<k_2<\cdots<k_R=k$ 
and $n=l_0<l_2<\cdots<l_R=k$ defined in 
(\ref{EqDefinitionGridk})-(\ref{EqDefinitionGridl}),
which determine the block size, equal to dyadic numbers
$$k_r=2^{d p_r}\sim n 2^{r/\a},\qquad l_r=2^{dq_r}
\sim n 2^{r/\b}.$$
This can be achieved within a factor of $2^d$.
The basis $e_1,\ldots, e_k$ in 
Section~\ref{SectionMinimaxRate} can be taken equal to
the functions $\psi_{i,j}^v$ for $i=0,\ldots, I$, $j\in J_i$,  and
$v\in \{0,1\}^d$, with $v\not=0$ if  $i>0$. Because there are 
$2^d$ as many functions $\psi_{i,j}^v$ at resolution level
$i=i_0+1$ than there are at level $i=i_0$, the
preceding display can be satisfied.

\section{Appendix 3: $U$-statistics}
\label{SectionU}
For degenerate, symmetric functions $f: \X^m\ra\RR$
and $g: \X^{m'}\ra\RR$ we have 
\begin{align}
P^n\UU_nf&=0,\nonumber\\
P^n (\UU_nf)(\UU_ng)&=\frac1{{n\choose m}}P^m fg,\qquad \hbox{ if } m=m',
\label{EqVarDegU}\\
P^n (\UU_nf)(\UU_ng)&=0,\qquad\qqqquad \hbox{ if } m\not=m'.\nonumber
\end{align}
(If the functions $f$ and $g$ are not symmetric, then the
second equation needs correction.)
The variance of $\UU_nf$ for a general measurable symmetric function 
$f: \X^m\to \RR$ can be obtained from this formula by decomposing
$f$ in its Hoeffding decomposition
$$f(X_1,\ldots, X_m)=\sum_{A\subset \{1,\ldots, m\}}f_{|A|}(X_A),$$
where $f_{|A|}(X_A)$ is the orthogonal projection of $f(X_1,\ldots,X_m)$
onto the set of square integrable random variables that are measurable
functions of $X_A:=(X_i: i\in A)$ 
that are orthogonal to the random variables that are measurable functions
of $X_B$ for any $B\not=A$ (see e.g. \cite{vanderVaart98}, Section~11.4). 
Because the terms of this decomposition are orthogonal and each term with 
$A\not=\emptyset$ is degenerate and symmetric, we have 
$$\var_p \UU_n f=\sum_{l=1}^m {m\choose l}^2\frac1{{n\choose l}}P^l f_l^2.$$
The functions $f_l$ can be expressed in the conditional expectations
$$\bar f_l(X_1,\ldots, X_l)=\E_p\bigl(f(X_1,\ldots, X_m)\given X_1,\ldots,X_l\bigr).$$
While $\bar f_l(X_1,\ldots, X_l)$ is the projection of $f(X_1,\ldots, X_m)$ onto the linear space
of all functions of $(X_1,\ldots, X_l)$, the variable $f_l(X_1,\ldots, X_l)$ is the projection onto the
smaller space of such functions that are also orthogonal to functions of fewer than $l$ variables.
Hence $P^l f_l^2\le P^l\bar f_l^2$ and an upper bound on the variance is obtained by replacing
$f_l$ by  $\bar f_l$. An alternative direct expression of the variance in the functions $\bar f_l$ is
obtained by writing
$$\var_p \UU_n f=\frac1{\binom nm^2}\!\!\sum_{\substack{A\subset\{1,\ldots, n\}\\ |A|=m}} \sum_{\substack{A'\subset\{1,\ldots, n\}\\ |A'|=m}} \cov\bigl(f(X_A),f(X_{A'})\bigr)
=\sum_{l=1}^m\frac{\binom{n-m}{m-l}\binom{m}{l}}{\binom{n}{m}}\z_l, $$ 
for $\z_l=P^l\bar f_l^2-f_0^2$ the covariance of $f(X_A)$ and $f(X_B)$ when $A$ and $B$ have $l$ variables in common
(see e.g.\ \cite{vanderVaart98}, page~163). This leads to the following upper bound.

\begin{lemma}
\label{LemVarUBound}
For any permutation-symmetric, measurable function $f: \X^m\to\RR$ and $n\ge 2m$,
\begin{equation}
\label{EqVarU}
\var_p\UU_nf\le \sum_{l=1}^m \frac{2^lm^{2l}}{n^l} P^l\bar f_l^2.
\end{equation}
\end{lemma}

\begin{proof}
We use the formula in the display preceding the lemma, with the further bounds $\z_l\le P^l\bar f_l^2$,
$\binom ml\le m^l$,  and, for $l\le m\le n$,
\begin{align*}
\frac{\binom{n-m}{m-l}}{\binom{n}{m}}
&=\frac{1}{n(n-1)\cdots(n-l+1)}\frac{m!}{(m-l)!}\times\\
&\qquad\times\frac{(n-m)(n-m-1)\cdots (n-m-(m-l-1))}{(n-l)(n-l-1)\cdots(n-(m-1))}
\le \frac{(2m)^l}{n^l},
\end{align*}
for $n\ge 2l$. We use here that $n-k\ge n/2$ for $k\le n/2$, $m-k\le m$ for $k\ge 0$, and 
that the fraction at the beginning of the second line is a product of numbers smaller than 1.
\end{proof}

For easy bounds on sums the following lemma is useful.

\begin{lemma}
\label{LemSumAverage}
For any random variables $Y_1,\ldots, Y_m$
\begin{itemize}
\item[(i)] $\var(Y_1+\cdots+ Y_m)\le \sum_{j=1}^m 2^j\var Y_j$.
\item[(ii)]$\E \bigl(m^{-1}(Y_1+\cdots +Y_m)\bigr)^2\le \max_j \E Y_j^2$.
\end{itemize}
\end{lemma}

\begin{proof}
For (i) we apply the triangle inequality to see 
that $\sdev(Y_1+\cdots+ Y_m)\le \sum_j\sdev Y_j$. Hence 
$\var(Y_1+\cdots +Y_m)\le \sum_ja_j\var Y_j\sum_j(1/a_j)$, by the Cauchy-Schwarz inequality,
for any $a_j>0$. We take $a_j=2^j$.

For (ii) we write the square as the double sum $m^{-2}\sum_{i,j} Y_iY_j$ and use the Cauchy-Schwarz
inequality to see that the expectation of this is bounded above by $m^{-2}\sum_{i,j}(\E Y_i^2\E Y_j^2)^{1/2}$,
which is bounded by $m^{-2}\sum_{i,j} \max_k \E Y_k^2$.
\end{proof}

\end{document}